\documentclass[11pt]{scrartcl}
\usepackage[utf8]{inputenc}
\usepackage{hyperref}
\usepackage{amsmath,amssymb,amstext}
\usepackage{braket}
\usepackage{amsthm}
\usepackage{dsfont}
\usepackage{color}
\usepackage[title]{appendix}

\newtheorem{vor}{Assumption}[section]
\newtheorem{theorem}[vor]{Theorem}

\newtheorem{lem}[vor]{Lemma}
\newtheorem{cor}[vor]{Corollary}
  
\newtheorem{prop}[vor]{Proposition}

\theoremstyle{definition}
\newtheorem{defi}[vor]{Definition}

\numberwithin{equation}{section}

\begin{document}
\title{The Fr\"ohlich Polaron at Strong Coupling -- Part I: The Quantum Correction to the Classical Energy}
\author{\vspace{.3cm} \textsc{Morris Brooks and Robert Seiringer}\\ IST Austria, Am Campus 1, 3400 Klosterneuburg, Austria}
\date{October 20, 2022}
\maketitle

\begin{abstract} 
\textsc{Abstract}. We study the Fr\"ohlich polaron model in $\mathbb{R}^3$, and establish the subleading term in the strong coupling asymptotics of its ground state energy, corresponding to the quantum corrections to the classical energy determined by the Pekar approximation. 
\end{abstract}

\section{Introduction and Main Results}
This is the first part of a study of the asymptotic properties of the Fr\"ohlich polaron, which is a model describing the interaction between an electron and the optical modes of a polar crystal \cite{F37}. In the regime of strong coupling between the electron and the optical modes, also called phonons, it is a well known fact \cite{leschke,DV,LT} that the ground state energy of the Fr\"ohlich polaron is asymptotically given by the minimal Pekar energy \cite{P54}, which can be considered as the ground state energy of an electron interacting with a classical phonon field. This result is motivated by using appropriately scaled units, see e.g. \cite{S}, which demonstrates that the strong coupling regime is a semi-classical limit in the phonon field variables. In such units the Fr\"ohlich Hamiltonian, acting on the space $L^2\!\left(\mathbb{R}^3\right)\otimes \mathcal{F}\left(L^2\!\left(\mathbb{R}^3\right)\right)$, reads
\begin{align}
\label{Equation-Hamiltonian}
\mathbb{H}:=-\Delta_x-a\left(w_x\right)-a^\dagger\left(w_x\right)+\mathcal{N},
\end{align}
where the annihilation and creation operators satisfy the rescaled canonical commutation relations $\left[a(f),a^\dagger(g)\right]=\alpha^{-2}\braket{f|g}$ for $f,g\in L^2\!\left(\mathbb{R}^3\right)$ with $\alpha>0$ being the coupling strength, the interaction is given by $w_x(x'):=\pi^{-\frac{3}{2}}|x'-x|^{-2}$ and $\mathcal{N}$ is the corresponding (rescaled) particle number operator, i.e. $\mathcal{N}:=\sum_{n=1}^\infty a^\dagger(\varphi_n)a(\varphi_n)$ where $\{\varphi_n:n\in \mathbb{N}\}$ is an orthonormal basis of $L^2\!\left(\mathbb{R}^3\right)$. The definition of the Fr\"ohlich Hamiltonian in Eq.~(\ref{Equation-Hamiltonian}) has to be understood in the sense of quadratic forms, see for example \cite{S}, due to the ultraviolet singularity in the interaction $w_x$. By substituting the annihilation and creation operators $a$ and $a^\dagger$ in Eq.~(\ref{Equation-Hamiltonian}) with a (classical) phonon field $\varphi\in L^2\! \left(\mathbb{R}^3\right)$, i.e. replacing $a(f)$ with $\braket{f|\varphi}$ and $a^\dagger(f)$ with $\braket{\varphi|f}$, we arrive at the Pekar energy
\begin{align}
\label{Equation-Pekar energy}
\mathcal{E}\left(\psi,\varphi\right):&=\big\langle \psi\big|-\Delta_x-\braket{w_x|\varphi}-\braket{\varphi|w_x}+\|\varphi\|^2\, \big|\psi\big\rangle\\
\nonumber
&=\int  \left|\nabla \psi(x)\right|^2\mathrm{d}x-\iint w_x(x')\left(\varphi(x')+\overline{\varphi(x')}\right)\left|\psi(x)\right|^2\mathrm{d}x'\mathrm{d}x+\int |\varphi(x')|^2\mathrm{d}x',
\end{align}
where $\psi\in L^2\! \left(\mathbb{R}^3\right)$ is the wave-function of the electron. We further define the Pekar functional $\mathcal{F}^\mathrm{Pek}\! \left(\varphi\right):=\inf_{\|\psi\|=1}\mathcal{E}\left(\psi,\varphi\right)$ and the minimal Pekar energy $e^\mathrm{Pek}:=\inf_\varphi \mathcal{F}^\mathrm{Pek}\! \left(\varphi\right)$. It is known that the ground state energy $E_\alpha:=\inf \sigma\left(\mathbb{H}\right)$, as a function of the coupling strength $\alpha$, is asymptotically given by the minimal Pekar energy $e^\mathrm{Pek}$ in the limit $\alpha\rightarrow \infty$ \cite{leschke,DV}. More precisely, one has  $e^\mathrm{Pek} \geq E_\alpha=e^\mathrm{Pek}+O_{\alpha\rightarrow \infty}\left(\alpha^{-\frac{1}{5}}\right)$, as shown in \cite{LT}. In this work we are going to verify the prediction in the physics literature \cite{tja,A2,A} that the sub-leading term in this energy asymptotics is actually of order $\alpha^{-2}$ with a rather explicit pre-factor
\begin{align}
\label{Equation-Main Asymptotic}
E_\alpha=e^\mathrm{Pek}-\frac{1}{2\alpha^2}\mathrm{Tr}\left[1-\sqrt{H^\mathrm{Pek}}\, \right]+o_{\alpha\rightarrow \infty}\left(\alpha^{-2}\right),
\end{align}
 where $\varphi^\mathrm{Pek}$ is a minimizer of $\mathcal{F}^\mathrm{Pek}$ and $H^\mathrm{Pek}$ is the Hessian of $\mathcal{F}^\mathrm{Pek}$ at $\varphi^\mathrm{Pek}$ restricted to real-valued functions $\varphi\in L^2_\mathbb{R}\! \left(\mathbb{R}^3\right)$, i.e. $H^\mathrm{Pek}$ is an operator on $L^2\! \left(\mathbb{R}^3\right)$ defined by
\begin{align}
\label{Equation-Definition_Hessian}
\braket{\varphi|H^\mathrm{Pek}|\varphi}=\lim_{\epsilon\rightarrow 0}\frac{1}{\epsilon^2}\left(\mathcal{F}^\mathrm{Pek}\! \left(\varphi^\mathrm{Pek}+\epsilon \varphi\right)-e^\mathrm{Pek}\right)
\end{align}
for all $\varphi\in L^2_\mathbb{R}\! \left(\mathbb{R}^3\right)$. The prediction in Eq.~(\ref{Equation-Main Asymptotic}) has been verified previously for polaron models either confined to a bounded region of $\mathbb{R}^3$ \cite{FS} or to a three-dimensional torus \cite{FeS}. The methods presented there  exhibit substantial problems regarding their extension to the unconfined case, however.  
In this paper we present a new approach, which is partly based on techniques previously developed in the study of Bose--Einstein  condensation and the validity of Bogoliubov's approximation for Bose gases \cite{LNR,LNSS,BS} in the mean-field limit. We employ a localization method for the phonon field, which breaks the translation-invariance and effectively reduces the problem to the confined case, allowing for an application of some of the methods developed in \cite{FS,FeS}. Our main result is the following Theorem \ref{Theorem: Main} where we verify the lower bound on $E_\alpha$ in Eq.~(\ref{Equation-Main Asymptotic})
.

\begin{theorem}
\label{Theorem: Main}
Let $E_\alpha$ be the ground state energy of $\mathbb{H}$ in \eqref{Equation-Hamiltonian}. For any  $s<\frac{1}{29}$
\begin{align}
\label{Equation-Main}
E_\alpha\geq e^\mathrm{Pek}-\frac{1}{2\alpha^2}\mathrm{Tr}\left[1-\sqrt{H^\mathrm{Pek}}\, \right]-\alpha^{-(2+s)} 
\end{align}
for all $\alpha\geq \alpha(s)$, where $\alpha(s)>0$ is a suitable constant.
\end{theorem}

As an intermediate result, which might be of independent interest, we will establish the existence of a family of approximate ground states, by which we mean states whose energy is given by the right side of \eqref{Equation-Main Asymptotic}, exhibiting complete Bose--Einstein condensation with respect to a minimizer $\varphi^\mathrm{Pek}$ of the Pekar functional $\mathcal{F}^\mathrm{Pek}$. We refer to Theorem \ref{Theorem-Existence of a strong condensate} for a precise statement. 

In contrast to the lower bound, the proof of the upper bound on $E_\alpha$ in Eq.~(\ref{Equation-Main Asymptotic}) is essentially the same as for confined polarons \cite{FS,FeS} and can be obtained by the same methods. It is also contained as a special case in \cite{MMS}, where it has been verified that the ground state energy $E_\alpha(P)$  as a function of the (conserved) total momentum $P$ can be bounded from above by
\begin{align}
\label{Equation-Parabolic Upper Bound}
E_\alpha(P)\leq e^\mathrm{Pek}-\frac{1}{2\alpha^2}\mathrm{Tr}\left[1-\sqrt{H^\mathrm{Pek}}\, \right]+\frac{|P|^2}{2\alpha^4 m}+C_{\epsilon}\alpha^{-\frac{5}{2}+\epsilon},
\end{align} 
where $m:=\frac{2}{3}\|\nabla \varphi^\mathrm{Pek}\|^2$ and $\epsilon>0$, with $C_{\epsilon}$  a suitable constant. Since $E_\alpha = E_\alpha(0)$ \cite{gross,spohnD,JSM},  Theorem \ref{Theorem: Main} in combination with Eq.~(\ref{Equation-Parabolic Upper Bound}) for the specific case $P=0$ concludes the proof of Eq.~(\ref{Equation-Main Asymptotic}). Combining \eqref{Equation-Parabolic Upper Bound} with Theorem \ref{Theorem: Main}, one further obtains an upper bound on the increment $E_\alpha(P)-E_\alpha$, a quantity related to the effective mass of the polaron \cite{Lpekar,LSmass,Spohnmass,BPmass}. In the second part \cite{BS2} we will discuss, in analogy to Theorem \ref{Theorem: Main}, the corresponding lower bound on $E_\alpha(P)$.\\

The proof of Eq.~(\ref{Equation-Main Asymptotic}) for confined systems in \cite{FS,FeS} requires an asymptotically correct local quadratic lower bound on the Pekar functional $\mathcal{F}^\mathrm{Pek}\! \left(\varphi\right)$ for configurations close to a minimizer, as well as a sufficiently strong quadratic lower bound valid for all configurations. While our proof of Theorem \ref{Theorem: Main} makes use of a local quadratic lower bound as well, we believe that in the translation-invariant setting any globally valid quadratic lower bound cannot be sufficiently strong, and therefore new ideas are necessary. As we explain in the following, we circumvent this problem by constructing an approximate ground state $\Psi$, which is essentially supported close to a minimizer of the Pekar functional $\mathcal{F}^\mathrm{Pek}$, and consequently we only require a locally valid quadratic lower bound.\\

\textbf{Proof strategy of Theorem \ref{Theorem: Main}.}  Even though we want to verify a lower bound on $E_\alpha$, let us first discuss how test functions providing an asymptotically correct upper bound are expected to look like. In the following let $(\psi^\mathrm{Pek},\varphi^\mathrm{Pek})$ denote a minimizer of the Pekar energy $\mathcal{E}$ defined in Eq.~(\ref{Equation-Pekar energy}). It has been established in \cite{Li} that all other minimizers are given by translations $\varphi^\mathrm{Pek}_x(x'):=\varphi^\mathrm{Pek}(x'-x)$ and $\psi^\mathrm{Pek}_x(x'):=e^{i\theta}\psi^\mathrm{Pek}(x'-x)$ of $\varphi^\mathrm{Pek}$ and $e^{i\theta}\psi^\mathrm{Pek}$, where $\theta$ is an arbitrary phase. W.l.o.g. let us denote in the following by  $(\psi^\mathrm{Pek},\varphi^\mathrm{Pek})$   the unique minimizer of $\mathcal{E}$ such that $\varphi^\mathrm{Pek}$ is radial and $\psi^\mathrm{Pek}$ is non-negative. Then all the product states of the form $\psi^\mathrm{Pek}_x\otimes \Omega_{\varphi^\mathrm{Pek}_x}$ with $x\in \mathbb{R}^3$, where $\Omega_{\varphi^\mathrm{Pek}_x}$ is the coherent state corresponding to $\varphi^\mathrm{Pek}_x$ (defined by $a(w) \Omega_\varphi = \langle w|\varphi\rangle \Omega_\varphi$ for all $w\in L^2\! \left(\mathbb{R}^3\right)$), have the asymptotically correct leading term in the  energy $\braket{\psi^\mathrm{Pek}_x\otimes \Omega_{\varphi^\mathrm{Pek}_x}|\, \mathbb{H}\, |\psi^\mathrm{Pek}_x\otimes \Omega_{\varphi^\mathrm{Pek}_x}}=e^\mathrm{Pek}$. By taking convex combinations of these states on the level of density matrices, we can construct a large family of low energy states
\begin{align*}
\Gamma_\mu:=\int_{\mathbb{R}^3} \ket{\psi^\mathrm{Pek}_x\otimes \Omega_{\varphi^\mathrm{Pek}_x}}\bra{\psi^\mathrm{Pek}_x\otimes \Omega_{\varphi^\mathrm{Pek}_x}}\mathrm{d}\mu(x)
\end{align*}
for any given probability measure $\mu$ on $\mathbb{R}^3$.
Clearly, $\Gamma_\mu$ exhibits the correct leading energy $\braket{\, \mathbb{H}\, }_{\Gamma_\mu}=e^\mathrm{Pek}$. Our proof of the lower bound given in Eq.~(\ref{Equation-Main}) relies on the observation that asymptotically as $\alpha\rightarrow \infty$, any low energy state $\Psi$ is of the form $\Gamma_\mu$ with a suitable probability measure $\mu$ on $\mathbb{R}^3$. Since we only need this statement for the phonon part of $\Psi$, we will verify the weaker statement   
\begin{align*}
\mathrm{Tr}_\mathrm{electron}\big[\ket{\Psi}\bra{\Psi}\big]\approx \int_{\mathbb{R}^3} \ket{ \Omega_{\varphi^\mathrm{Pek}_x}}\bra{\Omega_{\varphi^\mathrm{Pek}_x}}\mathrm{d}\mu(x)
\end{align*}
instead, see Theorem \ref{Theorem-Coherent States} for a precise formulation. This statement is analogous to a version of the quantum de Finetti theorem used in \cite{LNR} in order to verify the Hartree approximation for Bose gases in a general setting. The main technical challenge of this paper will be the construction of approximate ground states $\Psi$ where the corresponding measure  is a delta measure, $\mu=\delta_{0}$, i.e. the construction of states where the phonon part is essentially given by a single coherent state $\Omega_{\varphi^\mathrm{Pek}}$. The method presented here is based on a grand-canonical version of the localization techniques previously developed for translation-invariant Bose gases in \cite{BS}, and in analogy to the concept of Bose--Einstein  condensation we say that such states satisfy (complete) condensation with respect to the Pekar minimizer $\varphi^\mathrm{Pek}$. Heuristically this means that only field configurations $\varphi$ close to the minimizer $\varphi^\mathrm{Pek}$ are relevant, hence the translational degree of freedom has been eliminated and the system is effectively confined.

Based on this observation we can adapt the strategy developed for confined polarons in \cite{FS,FeS}, which starts by introducing an ultraviolet regularization in the interaction $w_x$ with the aid of  a momentum cut-off $\Lambda$, leading  to the study of the truncated Hamiltonian $\mathbb{H}_\Lambda$. Using a lower bound on the excitation energy $\mathcal{F}^\mathrm{Pek}\! \left(\varphi\right)-e^\mathrm{Pek}$ that is, up to a symplectic transformation, quadratic in the field variables $\varphi$ and valid for all $\varphi$ close to the minimizer $\varphi^\mathrm{Pek}$, one can bound the truncated Hamiltonian from below by an operator that is, up to a unitary transformation, quadratic in the creation and annihilation operators. The lower bound  is only valid, however,  if tested against a state satisfying (complete) condensation in $\varphi^\mathrm{Pek}$. Finally an explicit diagonalization of this quadratic operator yields the desired lower bound in Eq.~(\ref{Equation-Main}). 

The symplectic transformation on the phase space $L^2\! \left(\mathbb{R}^3\right)$, respectively the corresponding unitary transformation on the Hilbert space $\mathcal{F}\left(L^2\! \left(\mathbb{R}^3\right)\right)$, is one of the key novel ingredients in our proof. It turns out to be necessary due to the presence of the translational symmetry, which makes it impossible to find a non-trivial positive semi-definite quadratic lower bound on $\mathcal{F}^\mathrm{Pek}\! \left(\varphi\right) - e^\mathrm{Pek}$. This issue has already been addressed in the study of a polaron on the three dimensional torus \cite{FeS}, where a  different coordinate transformation is used, however. The symplectic/unitary transformation presented in this paper is an adaptation of the one used in the study of translation-invariant Bose gases in \cite{BS}.\\

\textbf{Outline.} The paper is structured as follows. In Section \ref{Section-The cut-off Model} we will introduce an ultraviolet cut-off as well as a discretization in momentum space, and provide estimates on the energy cost associated with such approximations. Section \ref{Section-Construction of a Condensate} then contains our main technical result Theorem \ref{Theorem-Existence of a strong condensate}, in which we verify the existence of approximate ground states satisfying (complete) condensation with respect to a minimizer $\varphi^\mathrm{Pek}$ of the Pekar functional $\mathcal{F}^\mathrm{Pek}$. Subsequently we will discuss a large deviation estimate for such condensates in Section \ref{Section-Large Deviation Principle for strong Condensates}, quantifying the heuristic picture that only configurations close to the point of condensation matter. In Section \ref{Section-Properties of the Pekar Functional} we then discuss properties of the Pekar functional $\mathcal{F}^\mathrm{Pek}$. In particular, we will discuss quadratic approximations around the minimizer $\varphi^\mathrm{Pek}$ as well as lower bounds that are, up to a coordinate transformation, quadratic in $\varphi$. Together with  the error estimates from Section \ref{Section-The cut-off Model} and the large deviation estimate from Section \ref{Section-Large Deviation Principle for strong Condensates}, applied to the approximate ground state constructed in Section \ref{Section-Construction of a Condensate}, this will allow us to verify our main Theorem \ref{Theorem: Main} in Section \ref{Section-Proof of the lower Bound}. The subsequent Section \ref{Convergence to Coherent States}  contains the proof of Theorem \ref{Theorem-Coherent States}, which can be interpreted as a version of the quantum de Finetti theorem adapted to our setting. Finally,  Appendices \ref{Properties of the Pekar Minimizer} and \ref{Appendix-Estimates on Operator Norms} contain auxiliary results concerning the Pekar minimizer $\varphi^\mathrm{Pek}$ and the projections introduced in Section \ref{Section-The cut-off Model}, respectively.

\section{Models with Cut-off}
\label{Section-The cut-off Model}
In this section we will estimate the effect of the introduction of an  ultraviolet cut-off, as well as a discretization in momentum space, on the ground state energy, following similar ideas as in \cite{LT,FS,FeS}. We will eventually apply these results 
for two different levels of coarse graining, a rough scale used in the proof of Theorem \ref{Theorem-Coherent States} in Section \ref{Convergence to Coherent States}, which applies to low energy states with energy $e^\mathrm{Pek} + o_{\alpha\rightarrow \infty}(1)$, and a fine scale precise enough to yield the correct ground state energy up to errors of order $o_{\alpha\rightarrow \infty}\left(\alpha^{-2}\right)$, see the proof of Theorem \ref{Theorem: Main} in Section \ref{Section-Proof of the lower Bound}.
\begin{defi}
\label{Definition-Pi}
Given parameters $0< \ell<\Lambda$, let us define for $z\in 2\ell\, \mathbb{Z}^3\setminus \{0\}$ the cubes $C_z:=\left[z_1-\ell,z_1+\ell\right)\times \left[z_2-\ell,z_2+\ell\right)\times \left[z_3-\ell,z_3+\ell\right)$, and let $z^1,..,z^N$ be an enumeration of the set of all $z=(z_1,z_2,z_3)\in 2\ell\, \mathbb{Z}^3\setminus \{0\}$ such that $C_z\subset B_{\Lambda}(0)$, where $B_r(0)$ is the (open) ball of radius $r$ around the origin. Then we define the orthonormal system $e_n\in L^2\!\left(\mathbb{R}^3\right)$ as
\begin{align*}
e_n(x):=\frac{1}{\sqrt{(2\pi)^3 \int_{C_{z^n}}\frac{1}{|k|^2}\,\mathrm{d}k}}\int_{C_{z^n}}\frac{e^{i\, k\cdot x}}{|k|}\,\mathrm{d}k,
\end{align*}
as well as the translated system $e_{y,n}(x):=e_n(x-y)$ and the orthogonal projection $\Pi^y_{\Lambda,\ell}$ onto the space spanned by $\{e_{y,1},\dots,e_{y,N}\}$. Furthermore we denote with $\Pi_\Lambda$ the projection onto the spectral subspace of momenta $|k|\leq \Lambda$.
\end{defi}

\begin{lem}
\label{Lemma-Norm Estimate}
Let $w_x(x'):=\pi^{-\frac{3}{2}}|x'-x|^{-2}$. Then we obtain for $0<\ell<\Lambda$ and $x,y\in \mathbb{R}^3$ the following estimate on the $L^2$ norm
\begin{align*}
\left\|\Pi_\Lambda w_x-\Pi^y_{\Lambda,\ell} w_x\right\|\lesssim  |x-y|\ell \sqrt{\Lambda}+\sqrt{\ell}.
\end{align*}
\begin{proof}
With $\widehat{\cdot}$ denoting  Fourier transformation, we have
\begin{align*}
\sqrt{2\pi^2}\, \widehat{\Pi^y_{\Lambda,\ell} w_x}(k)=\sum_{n=1}^N \frac{1}{\int_{C_{z^n}}\frac{1}{|k'|^2}\mathrm{d}k'}\int_{C_{z^n}}\frac{e^{ik'\cdot (y-x)}}{|k'|^2}\mathrm{d}k'\, \frac{1}{|k|}\mathds{1}_{C_{z^n}}(k),
\end{align*}
where we have used that $\widehat{\Pi_\Lambda w_x}(k)=\frac{1}{\sqrt{2\pi^2}|k|}\mathds{1}_{B_\Lambda(0)}(k)$. Defining the function $\sigma_n(k,x,y):=\frac{1}{\int_{C_{z^n}}\frac{1}{|k'|^2}\mathrm{d}k'}\int_{C_{z^n}}\frac{e^{ik'\cdot (y-x)}-e^{ik\cdot (y-x)}}{|k'|^2}\mathrm{d}k'$, we further have
\begin{align*}
\sqrt{2\pi^2}\left(\widehat{\Pi^y_{\Lambda,\ell} w_x}(k)-\widehat{\Pi_\Lambda w_x}(k)\right)=\sum_{n=1}^N \sigma_n(k,x,y)\frac{1}{|k|}\mathds{1}_{C_{z^n}}(k)-\frac{1}{|k|}\mathds{1}_{A}(k)
\end{align*}
with $A:=B_\Lambda(0)\setminus \left(\bigcup_{n=1}^N C_{z^n}\right)$. Making use of the estimate $|\sigma_n(k,x,y)|^2\leq |y-x|^2\max_{k'\in C_{z^n}}|k'-k|^2\leq 12|x-y|^2\ell^2 $ for $k\in C_{z^n}$, we therefore obtain
\begin{align*}
\sum_{n=1}^N \int_{C_{z^n}} |\sigma_n(k,x,y)|^2\frac{1}{|k|^2}\mathrm{d}k\leq 12|x-y|^2\ell^2\int_{|k|\leq \Lambda}\frac{1}{|k|^2}\mathrm{d}k=48\pi |x-y|^2\ell^2 \Lambda.
\end{align*}
Since $A\subset B_{2\ell}\cup B_\Lambda\setminus B_{\Lambda-4\ell}$ we consequently have $\int_A \frac{1}{|k|^2}\mathrm{d}k\lesssim \ell$.
\end{proof}
\end{lem}

\begin{defi}
\label{Definition-Cut off Hamiltonian general}
For $y\in \mathbb{R}^3$, $0<\ell<\Lambda$, let us define the cut-off Hamiltonians
\begin{align}
\label{Equation-Cut off Hamiltonian general}
\mathbb{H}^y_{\Lambda,\ell}&:=-\Delta_x-a\left(\Pi^y_{\Lambda,\ell}  w_x\right)-a^\dagger\left(\Pi^y_{\Lambda,\ell} w_x\right)+\mathcal{N},\\
\label{Equation-Momentum Cut off Hamiltonian}
\mathbb{H}_\Lambda&:=-\Delta_x-a\left(\Pi_{\Lambda}  w_x\right)-a^\dagger\left(\Pi_{\Lambda} w_x\right)+\mathcal{N}.
\end{align}
\end{defi}

These Hamiltonians can be interpreted as the restriction of $\mathbb{H}$ (in the quadratic form sense) to states where only the phonon modes in $\Pi^y_{\Lambda,\ell}L^2\! \left(\mathbb{R}^3\right)$, respectively $\Pi_{\Lambda}L^2\! \left(\mathbb{R}^3\right)$, are occupied. In particular, this implies that $\inf \sigma(\mathbb{H}^y_{\Lambda,\ell}) \geq E_\alpha$ as well as $\inf \sigma(\mathbb{H}_{\Lambda}) \geq E_\alpha$.  In the following we shall quantify the energy increase due to the introduction of the cut-offs.

Note that the $\alpha$-dependence of the Hamiltonians $\mathbb{H}$, $\mathbb{H}^y_{\Lambda,\ell}$ and $\mathbb{H}_\Lambda$ only enters through the rescaled canonical commutation relations $\left[a(f),a^\dagger(g)\right]=\alpha^{-2}\braket{f|g}$ satisfied by the creation and annihilation operators $a^\dagger$ and $a$, and we will usually suppress the $\alpha$ dependency in our notation for the sake of readability. In the rest of this paper, we will always assume that $\alpha$ is a parameter satisfying $\alpha\geq 1$ and, in case it is not stated otherwise, estimates hold uniformly in this parameter for $\alpha\rightarrow \infty$, i.e. we write $X\lesssim Y$ in case there exist constants $C,\alpha_0>0$ such that $X\leq C\, Y$ for all $\alpha\geq \alpha_0$. 

The proof of the subsequent Lemma \ref{Lemma-Lieb-Yamazaki bound} closely follows the arguments in \cite{LY,LT}, where it was shown that $\mathbb{H}$ is bounded from below and well approximated by an operator containing only finitely many phonon modes. For the sake of completeness we will illustrate the proof, which is based on the Lieb--Yamazaki commutator method, see \cite{LY}. In the following Lemma \ref{Lemma-Lieb-Yamazaki bound}, we will use the identification $L^2\!\left(\mathbb{R}^3\right)\otimes \mathcal{F}\left(L^2\!\left(\mathbb{R}^3\right)\right)\cong L^2\!\left(\mathbb{R}^3,\mathcal{F}\left(L^2\!\left(\mathbb{R}^3\right)\right)\right)$, in order to represent elements $\Psi\in L^2\!\left(\mathbb{R}^3\right)\otimes \mathcal{F}\left(L^2\!\left(\mathbb{R}^3\right)\right)$ as functions $x\mapsto \Psi(x)$ with values in $\mathcal{F}\left(L^2\!\left(\mathbb{R}^3\right)\right)$, allowing us to define the support $\mathrm{supp}\left(\Psi\right)$ as the closure of $\{x\in \mathbb{R}^3: \Psi(x)\neq 0\}$.

\begin{lem}
\label{Lemma-Lieb-Yamazaki bound}
We have for all $0<\ell<\Lambda\leq K$ and $L>0$, and states $\Psi$ with $\mathrm{supp}\left(\Psi\right)\subset B_L(y)$ the estimate
\begin{align}
\label{Equation-General estimate for cut-off's}
\left|\braket{\Psi|\mathbb{H}_K-\mathbb{H}^y_{\Lambda,\ell}|\Psi}\right|\lesssim  \! \left(L\ell \sqrt{\Lambda}+\sqrt{\ell}+\sqrt{\frac{1}{\Lambda}-\frac{1}{K}}\right)\! \braket{\Psi|-\Delta_x+\mathcal{N}+1|\Psi}.
\end{align}
Furthermore, there exists a constant $d>0$ such that 
\begin{align}
\label{Equation-Boundedness from below_first line}
\mathbb{H}_{K}&\geq -\frac{d}{t^2}-t\, \left(\mathcal{N}+\alpha^{-2}\right),\\
\label{Equation-Boundedness from below}
\mathbb{H}_{K}&\geq -d+\frac{1}{2}\left( -\Delta_x+\mathcal{N}\right)
\end{align}
for all $t>0$, $K\geq 0$ and $\alpha\geq 1$.
\end{lem}
\begin{proof}
Let us define the functions $u^n_x$ by $\widehat{u^n_x}(k):=\frac{1}{\sqrt{2\pi^2}}\mathds{1}_{B_K(0)\setminus B_\Lambda(0)}(k)\frac{k_n e^{ik\cdot x}}{|k|^3}$. 
We have  $a\left(\partial_{x_n}u^n_x\right)- a^\dagger\left(\partial_{x_n}u^n_x\right)  = \left[ \partial_{x_n} ,a\left(u^n_x\right)\! -\! a^\dagger\left(u^n_x\right)\right]$ and 
\begin{align*}
&\pm i \left[\partial_{x_n},a\left(u^n_x\right)\! -\! a^\dagger\left(u^n_x\right)\right]\leq - 2\epsilon \partial_{x_n}^2\! \! +\! \frac{1}{\epsilon}\left(a(u^n_x)^\dagger a(u^n_x) \!+\! a(u^n_x)a(u^n_x)^\dagger\right)\\
&\ \ \ \  \leq - 2 \epsilon\, \partial_{x_n}^2 + \frac{\|u^n_x\|^2}{\epsilon}\left(2\mathcal{N} + \alpha^{-2}\right)=2\|u^n_x\|\left( - \partial_{x_n}^2+\mathcal{N}+\frac{1}{2}\alpha^{-2}\right),
\end{align*}
where we have applied the Cauchy--Schwarz inequality in the first line and used the specific choice $\epsilon:=\|u^n_x\|$ in the last identity. Note that the $L^2$-norm $\|u^n_x\|$ is independent of $x$, and furthermore we can express $\pm\left(\mathbb{H}^y_{\Lambda,\ell} - \mathbb{H}_K\right)$ as
\begin{align*}
&\pm  a\left(\Pi_\Lambda w_x\!-\!\Pi^y_{\Lambda,\ell} w_x\right)\! \pm\! a^\dagger\left(\Pi_\Lambda w_x\!-\!\Pi^y_{\Lambda,\ell} w_x\right)\! \pm i\sum_{n=1}^3 \left(a\left(\partial_{x_n}u^n_x\right)\! -  a^\dagger\left(\partial_{x_n}u^n_x\right)\right)\! \\
&\ \ \ \ \leq 2\left\|\Pi_\Lambda w_x-\Pi^y_{\Lambda,\ell} w_x\right\|\left(1+\mathcal{N}\right)+2\max_{n\in \{1,2,3\}}\|u^n_x\|\left(-\Delta_x+3\mathcal{N}+\frac{3}{2}\alpha^{-2}\right).
\end{align*}
This concludes the proof of Eq.~(\ref{Equation-General estimate for cut-off's}), since we have $\left\|\Pi_\Lambda w_x-\Pi^y_{\Lambda,\ell} w_x\right\|\lesssim L\ell \sqrt{\Lambda}+\sqrt{\ell}$ for all $x\in \mathrm{supp}\left(\Psi\right)$ by Lemma \ref{Lemma-Norm Estimate} and $\|u^n_x\|^2\lesssim \frac{1}{\Lambda}-\frac{1}{K}$. The other statements in Eqs.~(\ref{Equation-Boundedness from below_first line}) and~(\ref{Equation-Boundedness from below}) can be verified similarly, using the decomposition $\Pi_{K} w_x=\Pi_{K'}w_x+\sum_{n=1}^3 \frac{1}{i}\partial_{x_n} g^n_x$ with $\widehat{g^n_x}(k):=\frac{1}{\sqrt{2\pi^2}}\mathds{1}_{B_K(0)\setminus B_{K'}(0)}(k)\frac{k_n e^{ik\cdot x}}{|k|^3}$ where $K'\leq K$ is large enough such that 
$\left\|g^n_x\right\|<\frac{1}{12}$.
\end{proof}

The subsequent Theorem \ref{Theorem-Strong cut-off} is a direct consequence of the results in \cite{FS} and \cite{M,FeS}, where multiple Lieb--Yamazaki bounds as well as a suitable Gross transformation are used in order to verify that the energy cost of introducing an ultraviolet cut-off $\Lambda=\alpha^{\frac{4}{5}(1+\sigma)}$ with $\sigma>0$ is only of order $o_{\alpha\rightarrow \infty}\left(\alpha^{-2}\right)$. 
Combined with an application of the IMS localization formula, as was also done in \cite{LT}, one obtains the following.

\begin{theorem}
\label{Theorem-Strong cut-off}
Given a constant $0<\sigma\leq \frac{1}{4}$, let us introduce the momentum cut-off $\Lambda:=\alpha^{\frac{4}{5}(1+\sigma)}$ as well as the space cut-off $L:=\alpha^{1+\sigma}$. Then there exists a sequence of states $\Psi^\diamond_\alpha$ satisfying $\braket{\Psi^\diamond_\alpha|\mathbb{H}_{\Lambda}|\Psi^\diamond_\alpha}-E_\alpha\lesssim \alpha^{-2(1+\sigma)}$ and $\mathrm{supp}\left(\Psi^\diamond_\alpha\right)\subset B_{L}(0)$, where $E_\alpha$ is the ground state energy of $\mathbb{H}$. 
\end{theorem}
\begin{proof}
We start by arguing that 
\begin{align}
\label{Equation-Ultraviolet cut-off}
\inf \sigma\left(\mathbb{H}_{\Lambda}\right)-E_\alpha\lesssim \Lambda^{-\frac{5}{2}}+\alpha^{-1}\Lambda^{-\frac{3}{2}}+\alpha^{-2}\Lambda^{-\frac{1}{2}}
\end{align}
 for large $\alpha$. An analogous bound was shown in \cite[Prop.~7.1]{FS} in the confined case, where additional powers of  $\ln \Lambda$ appear due to complications coming from the boundary. In the translation-invariant setting on a torus, \eqref{Equation-Ultraviolet cut-off} is shown \cite[Prop.~4.5]{FeS}, 
and that proof applies verbatim also in the unconfined case considered here (as has been worked out also in \cite{M}). 

 By our choice of $\Lambda=\alpha^{\frac{4}{5}(1+\sigma)}$, we immediately obtain $\inf \sigma\left(\mathbb{H}_{\Lambda}\right)-E_\alpha\lesssim \alpha^{-2(1+\sigma)}$. Hence there exists a state $\Psi$ satisfying $\braket{\Psi|\mathbb{H}_{\Lambda}|\Psi}-E_\alpha\lesssim   \alpha^{-{2(1+\sigma)}}$. In order to construct a state which is furthermore supported on the ball $B_{L}(0)$, let $\chi$ be a non-negative $H^1\! \left(\mathbb{R}^3\right)$ function with $\int \chi(y)^2\mathrm{d}y=1$ and $\mathrm{supp}\left(\chi\right)\subset B_1(0)$. We define $\Psi_y(x):=L^{-\frac{3}{2}}\chi\left(L^{-1}(x-y)\right)\Psi(x)$ for $y\in \mathbb{R}^3$ and compute, using the IMS identity,
 \begin{align*}
 \int \braket{\Psi_y|\mathbb{H}_\Lambda|\Psi_y}\, \mathrm{d}y&=\braket{\Psi|\mathbb{H}_\Lambda|\Psi}+L^{-3}\iint \left|\nabla_x\chi\left(L^{-1}(x-y)\right)\right|^2 \mathrm{d}y\, \|\Psi(x)\|^2\mathrm{d}x\\
 &=\braket{\Psi|\mathbb{H}_\Lambda|\Psi}+L^{-2}\|\nabla\chi\|^2=E_\alpha+O_{\alpha\rightarrow \infty}\left(\alpha^{-2(1+\sigma)}\right),
 \end{align*}
 see also \cite{LT} where an explicit choice of $\chi$ is used. Since $\int \|\Psi_y\|^2\mathrm{d}y=1$, there clearly exists a $y\in \mathbb{R}^3$ such that the state $\Psi^\diamond_\alpha:=\|\Psi_y\|^{-1}\Psi_y$ satisfies $\braket{\Psi^\diamond_\alpha|\mathbb{H}_\Lambda|\Psi^\diamond_\alpha}-E_\alpha\lesssim \alpha^{-2(1+\sigma)}$. By the translation invariance of $\mathbb{H}_\Lambda$ we can assume that $y=0$.
\end{proof}

\section{Construction of a Condensate}
\label{Section-Construction of a Condensate}

The purpose of this section is to construct a sequence of approximate ground states $\Psi_\alpha$, i.e. states with $\braket{\Psi_\alpha|\mathbb{H}_\Lambda|\Psi_\alpha}=E_\alpha+o_{\alpha\rightarrow \infty}\left(\alpha^{-2}\right)$ and $\Lambda$ as in Theorem \ref{Theorem-Strong cut-off}, that additionally satisfy complete condensation with respect to a minimizer $\varphi^\mathrm{Pek}$ of the Pekar functional $\mathcal{F}^\mathrm{Pek}$, i.e. the phonon part of $\Psi_\alpha$ is in a suitable sense close to a coherent state $\Omega_{\varphi^\mathrm{Pek}}$ with $\Omega_{\varphi^\mathrm{Pek}}:=e^{ \alpha^2 a^\dagger\left(\varphi^\mathrm{Pek}\right)-\alpha^2 a\left(\varphi^\mathrm{Pek}\right)}\Omega$, where $\Omega$ is the vacuum in $\mathcal{F}\left(L^2\!\left(\mathbb{R}^3\right)\right)$, see Lemma \ref{Lemma-Existence of a condensate} and Theorem \ref{Theorem-Existence of a strong condensate}. The construction will be based on various localization procedures of the phonon field with respect to operators of the form $\widehat{F}$ defined in the subsequent Definition \ref{Definition-Hat operators}. Before we start with the localization procedures, we will discuss an asymptotic formula for the expectation value $\braket{\Psi_\alpha|\widehat{F}|\Psi_\alpha}$ in Theorem \ref{Theorem-Coherent States} as well as the energy cost of localizing with respect to such an operator $\widehat{F}$ in Lemma \ref{Lemma-IMS Formula}. 

\begin{defi}
\label{Definition-Hat operators}
Given a function $F:\mathcal{M}\left(\mathbb{R}^3\right)\longrightarrow \mathbb{R}$, where $\mathcal{M}\left(\mathbb{R}^3\right)$ is the set of finite (Borel) measures on $\mathbb{R}^3$, let us define the operator $\widehat{F}$ acting on the Fock space $\mathcal{F}\left(L^2\! \left(\mathbb{R}^3\right)\right)=\bigoplus\limits_{n=0}^\infty L^2_{\mathrm{sym}}\! (\mathbb{R}^{3\times n})$ as $\widehat{F}\bigoplus\limits_{n=0}^\infty \Psi_n:=\bigoplus\limits_{n=0}^\infty F^n\Psi_n$, where 
\begin{align*}
\left(F^n\Psi_n\right)(x^1,\dots ,x^n):=F\left(\alpha^{-2}\sum_{k=1}^n\delta_{x^k}\right)\Psi_n(x^1,\dots ,x^n)
\end{align*}
and $F_0\Psi_0=F(0)\Psi_0$, i.e. $\widehat{F}$ acts component-wise on $\bigoplus\limits_{n=0}^\infty L^2_{\mathrm{sym}}\! (\mathbb{R}^{3\times n})$ by multiplication with the real valued function $(x^1,\dots ,x^n)\mapsto F\left(\alpha^{-2}\sum_{k=1}^n\delta_{x^k}\right)$.
\end{defi}
In order to keep the notation simple, we will allow $F:\mathcal{M}\left(\mathbb{R}^3\right)\longrightarrow \mathbb{R}$ to act on non-negative $L^1\left(\mathbb{R}^3\right)$ functions $q:\mathbb{R}^3\longrightarrow  [0,\infty)$ as well by identifying them with the corresponding measure $\lambda\in \mathcal{M}\left(\mathbb{R}^3\right)$ defined as $\frac{\mathrm{d}\lambda}{\mathrm{d}x}=q(x)$.

Before we discuss the asymptotic formula for the expectation value $\braket{\Psi_\alpha|\widehat{F}|\Psi_\alpha}$, let us introduce a family of cut-off functions $\chi^\epsilon\left(a\leq f(x)\leq b\right)$ where $\epsilon\geq 0$ determines the sharpness of the cut-off. In the following let $\alpha,\beta:\mathbb{R}\longrightarrow [0,1]$ be $C^\infty$ functions such that $\alpha^2+\beta^2=1$, $\mathrm{supp}\left(\alpha\right)\subset(-\infty,1)$ and $\mathrm{supp}\left(\beta\right)\subset(-1,\infty)$. For a given function $f:X\longrightarrow \mathbb{R}$ and constants $-\infty\leq a<b\leq \infty$, let us define the function $\chi^\epsilon\left(a\leq f\leq b\right):X\longrightarrow [0,1]$ as
\begin{align}
\label{Equation-Epsilon cut-off}
\chi^\epsilon\left(a\leq f(x)\leq b\right):=
\begin{cases}
\alpha\left(\frac{f(x)-b}{\epsilon}\right)\beta\left(\frac{f(x)-a}{\epsilon}\right),\text{ for }\epsilon>0\\
\mathds{1}_{[a,b]}\left(f(x)\right),\text{ for }\epsilon=0.
\end{cases}
\end{align}
Note that $\sum_{j\in J}\chi^\epsilon\left(a_j\leq f(x)\leq b_j\right)^2=1$ in case the intervals $[a_j,b_j)$ are a disjoint partition of $\mathbb{R}$ with $-\infty\leq a_j<b_j\leq \infty$.

 Similarly, we define the operator $\chi^\epsilon\left(a\leq T\leq b\right):=\int \chi^\epsilon\left(a\leq t\leq b\right)\, \mathrm{d}E(t)$, where $T$ is a self-adjoint operator and $E$ is the spectral measure with respect to $T$. Furthermore we will write $\chi\left(a\leq f\leq b\right)$, respectively $\chi\left(a\leq T\leq b\right)$, in case $\epsilon=0$ as well as $\chi^{\epsilon}\left(a\leq \cdot \right)$ and $\chi^{\epsilon}\left(\cdot \leq b\right)$ in case $b=\infty$ or $a=-\infty$, respectively.

The proof of the following Theorem \ref{Theorem-Coherent States} will be carried out in Section \ref{Convergence to Coherent States}. It is reminiscent of the quantum de-Finetti Theorem, and establishes in addition that for low energy states phonon field configurations are necessarily close to the set of Pekar minimizers given by $\{\varphi^\mathrm{Pek}_x\}_{x\in \mathbb{R}^3}$. 

\begin{theorem}
\label{Theorem-Coherent States}
Given $m\in \mathbb{N},C>0$ and $g\in L^2\! \left(\mathbb{R}^3\right)$, we can find a constant $T>0$ such that for all $\alpha\geq 1$ and states $\Psi$ satisfying $\chi\left(\mathcal{N}\leq C\right)\Psi=\Psi$ and $\braket{\Psi|\mathbb{H}_K|\Psi}\leq e^\mathrm{Pek}+\delta e$ with $\delta e\geq 0$ and $K\geq \alpha^{\frac{8}{29}}$, there exists a probability measure $\mu$ on $\mathbb{R}^3$, with the property
\begin{align}
\label{Equation-Coherent state formula for F}
\left|\big\langle \Psi\big|\widehat{F}\, \big|\Psi \big\rangle -\int_{\mathbb{R}^3} F\left(|\varphi_x^\mathrm{Pek}|^2\right)\, \mathrm{d}\mu(x)\right|\leq  T\|f\|_\infty\max\Big\{\sqrt{\delta e},\alpha^{-\frac{2}{29}}\Big\}
\end{align}
for all $F:\mathcal{M}\left(\mathbb{R}^3\right)\longrightarrow \mathbb{R}$ of the form $F\left(\rho\right)=\int\dots \int f(x_1,\dots , x_m)\, \mathrm{d}\rho(x_1)\dots \mathrm{d}\rho(x_m)$ with bounded $f:\mathbb{R}^{3\times m}\longrightarrow \mathbb{R}$, and furthermore
\begin{align}
\label{Equation-Translated F version}
&\left|\Big\langle \Psi\Big|W_g^{-1}\mathcal{N}W_g\Big|\Psi\Big\rangle-\int_{\mathbb{R}^3} \left\|\varphi_x^\mathrm{Pek}-g\right\|^2\, \mathrm{d}\mu\left(x\right)\right|\leq T\max\Big\{\sqrt{\delta e},\alpha^{-\frac{2}{29}}\Big\},
\end{align}
where $W_g$ is the Weyl operator characterized by $W_g^{-1}a(h)W_g=a(h)-\braket{h|g}$.
\end{theorem}

In the subsequent Lemma \ref{Lemma-IMS Formula} we introduce a generalized IMS-type estimate quantifying the energy cost of localizing with respect to an $\widehat{F}$-operator, similar to the generalized IMS results in \cite[Theorem A.1]{LS} and \cite[Proposition 6.1]{LNSS}. In order to formulate the result, let us define for a given subset $\Omega\subset \mathcal{M}\left(\mathbb{R}^3\right)$ and a (quadratic) partition of unity $\mathcal{P}=\{F_j:\mathcal{M}\left(\mathbb{R}^3\right)\longrightarrow \mathbb{R}:j\in J\}$, i.e. $0\leq F_j\leq 1$ and $\sum_{j\in J}F_j^2=1$, the variation of this partition on $\Omega$ as
\begin{align*}
V_\Omega\left(\mathcal{P}\right):=\alpha^{4}\underset{\rho\in \Omega,y\in \mathbb{R}^3}{\sup}\sum_{j\in J}\left|F_j\left(\rho+\alpha^{-2}\delta_y\right)-F_j\left(\rho\right)\right|^2.
\end{align*}

\begin{lem}
\label{Lemma-IMS Formula}
There exists a constant $c>0$, such that for any partition of unity $\mathcal{P}=\{F_j:\mathcal{M}\left(\mathbb{R}^3\right)\longrightarrow \mathbb{R}:j\in J\}$, $\Omega\subset \mathcal{M}\left(\mathbb{R}^3\right)$, $K>0$, $\alpha\geq 1$ and state $\Psi$ with $\widehat{\mathds{1}_{\Omega}}\Psi=\Psi$
\begin{align}
\label{Equation-hat F IMS}
\left|\sum_{j\in J }\big\langle \widehat{F_j}\Psi\big|\mathbb{H}_K\big|\widehat{F_j}\Psi\big\rangle-\braket{\Psi|\mathbb{H}_K|\Psi}\right|\leq c\sqrt{K} \alpha^{-4}V_\Omega\left(\mathcal{P}\right) \big\langle \Psi\big|\sqrt{\mathcal{N}+\alpha^{-2}}\, \big|\Psi  \big\rangle.
\end{align}
Furthermore given $M>0$, there exists a constant $c'>0$ such that we have for any $\varphi\in L^2\! \left(\mathbb{R}^3\right)$ satisfying $\|\varphi\|\leq M$, partition of unity $\{f_j:\mathbb{R}\longrightarrow \mathbb{R}:j\in J\}$, $K\geq 1$, $\alpha\geq 1$ and state $\Psi$
\begin{align*}
\left|\sum_{j\in J }\braket{\Psi_j|\mathbb{H}_K|\Psi_j}-\braket{\Psi|\mathbb{H}_K|\Psi}\right|\leq c'\sqrt{K}   \alpha^{-4} V_{\mathcal{M}\left(\mathbb{R}^3\right)}\left(\mathcal{P}'\right)\Big\langle  \Psi\Big| \sqrt{\mathcal{N}+1}\Big| \Psi \Big\rangle,
\end{align*}
where we define $\Psi_j:=f_j\left(W_\varphi^{-1}\mathcal{N}W_\varphi\right)\Psi$ with $W_\varphi$ being the corresponding Weyl operator and $\mathcal{P}':=\{F'_j:\mathcal{M}\left(\mathbb{R}^3\right)\longrightarrow \mathbb{R}:j\in J\}$ with $F'_j(\rho):=f_j(\int \mathrm{d}\rho)$.
\end{lem}
\begin{proof}
By applying the IMS identity, we obtain
\begin{align*}
\sum_{j\in J }\widehat{F_j}\mathbb{H}_K\widehat{F_j}-\mathbb{H}_K=\frac{1}{2}\sum_{j\in J } \left[\left[\widehat{F_j},\mathbb{H}_K\right],\widehat{F_j}\right] =-\sum_{j\in J }\mathfrak{Re}\,  \left[\left[\widehat{F_j},a\left(\Pi_K w_x\right)\right],\widehat{F_j}\right] ,
\end{align*}
where we have used the fact that $F_j$ commutes with $-\Delta_x$ and $\mathcal{N}$ in the last identity. Since a state $\Psi$ is a function with values in $\mathcal{F}\left(L^2\! \left(\mathbb{R}^3\right)\right)=\bigoplus\limits_{n=0}^\infty L^2_{\mathrm{sym}}\! (\mathbb{R}^{3\times n})$, we can represent it as $\Psi=\bigoplus_{n=0}^\infty \Psi_n$ where $\Psi_n(y,x^1,\dots ,x^{n})$ is a function of the electron variable $y$ and the $n$ phonon coordinates $x^k\in \mathbb{R}^3$. In order to simplify the notation, we will suppress the dependence on the electron variable $y$ in our notation. By an explicit computation, we obtain $\left[\left[\widehat{F},a\left(v\right)\right],\widehat{F}\right]\bigoplus_{n=0}^\infty \Psi_n=-\bigoplus_{n=0}^\infty \sqrt{\frac{n+1}{\alpha^2}}\Psi'_n$ with
\begin{align*}
\Psi'_n(x^1, \dots  ,x^n)\! =\! \! \! \int\! \! \Bigg[\! F\bigg(\alpha^{-2}\sum_{k=1}^{n+1}\delta_{x^k}\bigg)\!-\!F\bigg(\alpha^{-2}\sum_{k=1}^{n}\delta_{x^k}\bigg)\!\Bigg]^2 \! \! \! v(x^{n\!+\!1})\Psi_{n\!+\!1}(x^1,\dots ,x^{n\!+\!1})\mathrm{d}x^{n\!+\!1},
\end{align*}
for $v\in L^2\! \left(\mathbb{R}^3\right)$ and $F:\mathcal{M}\left(\mathbb{R}^3\right)\longrightarrow \mathbb{R}$. By the definition of $V_\Omega\left(\mathcal{P}\right)$ we obtain that 
\begin{align*}
\sigma(x^1,\dots,x^{n+1}):=\sum_{j\in J}\bigg[ F_j\Big(\alpha^{-2}\sum_{k=1}^{n+1}\delta_{x^k}\Big)-F_j\Big(\alpha^{-2}\sum_{k=1}^{n}\delta_{x^k}\Big)\bigg]^2\leq \alpha^{-4}V_\Omega\left(\mathcal{P}\right)
\end{align*}
for all $x^{n+1}\in \mathbb{R}^3$ and every $(x^1,\dots,x^n)\in \mathbb{R}^{3n}$ with $\alpha^{-2}\sum_{k=1}^n \delta_{x^k}\in \Omega$. Hence we can estimate $\left|\Big\langle \Psi\Big|  \sum_{j\in J}\mathfrak{Re}\left[\left[\widehat{F}_j,a\left(v\right)\right],\widehat{F}_j\right]\Big| \Psi \Big\rangle\right|$, using the notation $X=(x^1,\dots,x^n)$, by
\begin{align*}
&\sum_{n=0}^\infty \sqrt{\frac{n\! +\! 1}{\alpha^2}} \int  |\Psi_n(X)| \int \sigma(X,x^{n\!+\! 1})|v(x_{n+1})\Psi_{n+1}(X,x_{n\!+\! 1})|\mathrm{d}x_{n + 1}\mathrm{d}X\\
&\ \ \ \ \leq \alpha^{-5} V_\Omega\left(\mathcal{P}\right)\sum_{n=0}^\infty \sqrt{n+1}\int \!|\Psi_n(X)|\int |v(x_{n+1})\Psi_{n+1}(X,x_{n+1})|\mathrm{d}x^{n+1}\mathrm{d}X\\
&\ \ \ \ \leq \alpha^{-5} V_\Omega\left(\mathcal{P}\right)\|v\|\sum_{n=0}^\infty \sqrt{n+1}\big\|\Psi_n\big\|\, \big\|\Psi_{n+1}\big\|\leq \alpha^{-4}V_\Omega\left(\mathcal{P}\right)\|v\|\Big\langle\Psi\Big| \sqrt{\mathcal{N}+\alpha^{-2}}\,\Big|\Psi \Big\rangle.
\end{align*}
This concludes the proof of Eq.~(\ref{Equation-hat F IMS}), using the concrete choice $v:=\Pi_K  w_x$, since $\left\|\Pi_K  w_x\right\|^2=\frac{1}{2\pi^2} \int_{|k|\leq K}\frac{1}{|k|^2}=\frac{2}{\pi} K$. 

In order to verify the second statement we apply the unitary transformation $W_\varphi$ to the operator $\mathbb{X}:=\sum_{j\in J }f_j\left(W_\varphi^{-1}\mathcal{N}W_\varphi\right)\mathbb{H}_Kf_j\left(W_\varphi^{-1}\mathcal{N}W_\varphi\right)-\mathbb{H}_K$ and compute
\begin{align*}
W_\varphi  \mathbb{X}\, W_\varphi^{-1} &=\frac{1}{2}\sum_{j\in J } \Big[\Big[f_j\left(\mathcal{N}\right),W_\varphi\mathbb{H}_K W_\varphi^{-1}\Big],f_j\left(\mathcal{N}\right)\Big]\\
& =\sum_{j\in J } \mathfrak{Re}\Big[\Big[f_j\left(\mathcal{N}\right),a\left(\varphi-\Pi_K w_x\right)\Big],f_j\left(\mathcal{N}\right)\Big]=\sum_{j\in J } \mathfrak{Re}\Big[\Big[\widehat{F}_j',a\left(v\right)\Big],\widehat{F}_j'\Big],
\end{align*}
where we defined $v:=\varphi-\Pi_K w_x$ and applied the definition $F_j'(\rho)=f_j\left(\int \mathrm{d}\rho\right)$. We know from the previous estimates that
\begin{align*}
\pm \sum_{j\in J } \mathfrak{Re}\Big[\Big[f_j\left(\mathcal{N}\right),a\left( v\right)\Big],f_j\left(\mathcal{N}\right)\Big]\leq \alpha^{-4}V_{\mathcal{M}\left(\mathbb{R}^3\right)}\left(\mathcal{P}'\right)\|v\|\sqrt{\mathcal{N}+\alpha^{-2}}.
\end{align*}
Clearly $\|v\|\leq \|\varphi\|+\|\Pi_K w_x\|\lesssim \sqrt{K}$ for $K\geq 1$, and consequently
\begin{align*}
&\left|\sum_{j\in J }\braket{\Psi_j|\mathbb{H}_K|\Psi_j}-\braket{\Psi|\mathbb{H}_K|\Psi}\right|\lesssim \sqrt{K}\alpha^{-4}V_{\mathcal{M}\left(\mathbb{R}^3\right)}\left(\mathcal{P}'\right)\Big\langle\Psi\Big| \sqrt{W_\varphi^{-1}\mathcal{N}W_\varphi+\alpha^{-2}}\,\Big|\Psi \Big\rangle\\
& \ \ \ \ \ \ \lesssim \sqrt{K}\alpha^{-4}V_{\mathcal{M}\left(\mathbb{R}^3\right)}\left(\mathcal{P}'\right)\Big\langle\Psi\Big| \sqrt{\mathcal{N}+1}\,\Big|\Psi \Big\rangle,
\end{align*}
where we have used that $W_\varphi^{-1}\mathcal{N}W_\varphi\leq 2\left(\mathcal{N}+\|\varphi\|^2\right)$ and the operator-monotonicity of the square root.
\end{proof}

In the following let $L:=\alpha^{1+\sigma}$ and $\Lambda:=\alpha^{\frac{4}{5}(1+\sigma)}$ with $0<\sigma\leq \frac{1}{4}$, and let $\Psi_\alpha^\bullet$ be a  sequence of states satisfying $\mathrm{supp}\left(\Psi_\alpha^\bullet\right)\subset B_L(0)$ and $\widetilde{E}_\alpha- E_\alpha\lesssim \alpha^{-\frac{4}{29}}$, where
\begin{align}
\label{Equation-Energy_of_given_sequence}
\widetilde{E}_\alpha:=\braket{\Psi_\alpha^\bullet|\mathbb{H}_\Lambda|\Psi_\alpha^\bullet}.
\end{align}
The exponent $\frac{4}{29}$ is chosen for convenience, as it allows  to simplify the right hand side of Eq.~(\ref{Equation-Coherent state formula for F}) to $\|f\|_\infty \alpha^{-\frac{2}{29}}$ (using that $E_\alpha \leq e^\mathrm{Pek}$). For the proof of Theorem~\ref{Theorem: Main} we shall use the specific choice $\Psi_\alpha^\diamond$ from Theorem~\ref{Theorem-Strong cut-off} for $\Psi_\alpha^\bullet$, but it will  be useful in the second part to have the first two localization procedures in Lemma \ref{Lemma-Admissible sequence} and \ref{Lemma-Control of variance} formulated for a more general sequence
 $\Psi_\alpha^\bullet$.

In the following Eq.~(\ref{Equation:First Localization}) and Eq.~(\ref{Equation:Second Localization}), we will apply localizations procedures to the given sequence $\Psi_\alpha^\bullet$ in order to construct states having additional useful properties, which we will use in Lemma \ref{Lemma-Existence of a condensate} in order to construct a sequence of approximate ground states satisfying complete condensation. Furthermore we will quantify the energy cost of these localizations by $\braket{\Psi_\alpha|\mathbb{H}_\Lambda|\Psi_\alpha}-\widetilde{E}_\alpha\lesssim \alpha^{-3}$ in the Lemmata \ref{Lemma-Admissible sequence} and \ref{Lemma-Control of variance}. In Theorem \ref{Theorem-Existence of a strong condensate} we will then apply a final localization procedure, in order to lift the (weak) condensation from Lemma \ref{Lemma-Existence of a condensate} to a strong one, following the argument in \cite{LNSS}.

Having Lemma \ref{Lemma-IMS Formula} at hand, we can verify our first localization result in Lemma \ref{Lemma-Admissible sequence}, which allows us to restrict our attention to states $\Psi'_\alpha$ having a (rescaled) particle number $\mathcal{N}$ between some fixed constants  $c_-$ and $c_+$. To be precise, for given $c_-,c_+$ and $\epsilon'$ we use the function $F_*(\rho):=\chi^{\epsilon'}\left(c_-+\epsilon'\leq \int \mathrm{d}\rho\leq c_+-\epsilon'\right)$ in order to define the states
\begin{align}
\label{Equation:First Localization}
\Psi_\alpha':=Z_{\alpha}^{-1}\widehat{F}_*\Psi^\bullet_\alpha,
\end{align}
with the corresponding normalization constants $Z_{\alpha}:=\|\widehat{F}_*\Psi^\bullet_\alpha\|$. By construction we  have $\chi\left(c_-\leq \mathcal{N}\leq c_+\right)\Psi'_\alpha=\Psi'_\alpha$  as well as $\mathrm{supp}\left(\Psi'_\alpha\right)\subset B_{L}(0)$. In the following Lemma \ref{Lemma-Admissible sequence} we derive an upper bound on the energy of $\Psi_\alpha'$, and in addition we will investigate the large $\alpha$ behavior of $Z_\alpha$, which will be useful in the second part.

\begin{lem}
\label{Lemma-Admissible sequence}
Let $\Psi_\alpha^\bullet$ be the sequence introduced above Eq.~(\ref{Equation-Energy_of_given_sequence}). Then there exist $\alpha$-independent constants $c_-,c_+,\epsilon'>0$ such that the corresponding states $\Psi_\alpha'$ defined in Eq.~(\ref{Equation:First Localization}) satisfy $\braket{\Psi'_\alpha|\mathbb{H}_\Lambda|\Psi'_\alpha}-\widetilde{E}_\alpha\lesssim \alpha^{-\frac{7}{2}}$. Furthermore, $Z_\alpha\underset{\alpha\rightarrow \infty}{\longrightarrow} 1$.
\end{lem}
\begin{proof}
In the following let $F_*$ be the function defined above Eq.~(\ref{Equation:First Localization}) and let us complete it to a quadratic partition of unity $\mathcal{P}:=\{F_-,F_*,F_+\}$ with the aid of the functions $F_-(\rho):=\chi^{\epsilon'}\left(\int \mathrm{d}\rho\leq c_-+\epsilon'\right)$ and $F_+(\rho):=\chi^{\epsilon'}\left(c_+-\epsilon'\leq \int \mathrm{d}\rho\right)$. Making use of Lemma \ref{Lemma-IMS Formula} and $\Lambda=\alpha^{\frac{4}{5}(1+\sigma)}\leq \alpha$, we then obtain
\begin{align}
\nonumber
&Z_{\alpha,-}^2 \braket{\Psi_{\alpha,-}|\mathbb{H}_\Lambda|\Psi_{\alpha,-}}+ Z_{\alpha}^2 \braket{\Psi_{\alpha}'|\mathbb{H}_\Lambda|\Psi_{\alpha}'} +  Z_{\alpha,+}^2  \braket{\Psi_{\alpha,+}|\mathbb{H}_\Lambda|\Psi_{\alpha,+}}\\
\label{Equation-IMS for lower and upper bound}
& \ \ \ \ \ \  \leq   \braket{\Psi^\bullet_\alpha|\mathbb{H}_\Lambda|\Psi^\bullet_\alpha} + c\, \alpha^{-\frac{7}{2}} V_{\mathcal{M}\left(\mathbb{R}^3\right)} \left(\mathcal{P}\right) \big\langle \Psi^\bullet_\alpha\big|\sqrt{\mathcal{N} + \alpha^{-2} }\big|\Psi^\bullet_\alpha \big\rangle,
\end{align}
where $\Psi_{\alpha,\pm}:=Z_{\alpha,\pm}^{-1}\widehat{F}_{(\pm)}\Psi^\bullet_\alpha$, with corresponding normalization factors $Z_{\alpha,\pm}:=\|\widehat{F}_{(\pm)}\Psi^\bullet_\alpha\|$. By Eq.~(\ref{Equation-Boundedness from below}) there exists a constant $d$ s.t. $\big\langle \Psi^\bullet_\alpha\big|\mathcal{N}  \big|\Psi^\bullet_\alpha \big\rangle\leq \big\langle \Psi^\bullet_\alpha\big|2\mathbb{H}_\Lambda+d  \big|\Psi^\bullet_\alpha \big\rangle\lesssim d+\alpha^{-\frac{4}{29}}$, where we have used the assumption $\big\langle \Psi^\bullet_\alpha\big|\mathbb{H}_\Lambda\big|\Psi^\bullet_\alpha \big\rangle=\widetilde{E}_\alpha\leq \widetilde{E}_\alpha-E_\alpha\lesssim \alpha^{-\frac{4}{29}}$.  The first derivative of the functions $\chi^{\epsilon'}(\cdot \leq c_-+{\epsilon'}),\chi^{\epsilon'}(c_-+{\epsilon'}\leq \cdot \leq c_+-{\epsilon'})$ and $\chi^{\epsilon'}(\cdot \leq c_+-{\epsilon'})$ is uniformly bounded by some $\epsilon'$-dependent constant $D$, and consequently we have for all finite measures $\rho$ and $\rho':=\rho+\alpha^{-2}\delta_y$ with $y\in \mathbb{R}^3$, and $\diamond\in \{-,*,+\}$,
\begin{align*}
\left|F_\diamond\left(\rho'\right)-F_\diamond\left(\rho\right)\right|\leq D \left|\int \mathrm{d}\rho'-\int \mathrm{d}\rho\right|=D\alpha^{-2}.
\end{align*}
This implies that $V_{\mathcal{M}\left(\mathbb{R}^3\right)}\left(\mathcal{P}\right)\lesssim 1$, and therefore the right hand side of Eq.~(\ref{Equation-IMS for lower and upper bound}) is bounded by $\braket{\Psi^\bullet_\alpha|\mathbb{H}_\Lambda|\Psi^\bullet_\alpha}+C\alpha^{-\frac{7}{2}}$ for a suitable $C>0$. Since $Z_{\alpha,-}^2+Z_{\alpha}^2+Z_{\alpha,+}^2=1$, this means that at least one of the terms $\braket{\Psi_{\alpha,-}|\mathbb{H}_\Lambda|\Psi_{\alpha,-}},\braket{\Psi_{\alpha}'|\mathbb{H}_\Lambda|\Psi_{\alpha}'}$ or $\braket{\Psi_{\alpha,+}|\mathbb{H}_\Lambda|\Psi_{\alpha,+}}$ is bounded from above by $\braket{\Psi^\bullet_\alpha|\mathbb{H}_\Lambda|\Psi^\bullet_\alpha}+C\alpha^{-\frac{7}{2}}=\widetilde{E}_\alpha+C\alpha^{-\frac{7}{2}}$. We can however rule out that $\braket{\Psi_{\alpha,-}|\mathbb{H}_\Lambda|\Psi_{\alpha,-}}$, respectively $\braket{\Psi_{\alpha,+}|\mathbb{H}_\Lambda|\Psi_{\alpha,+}}$, satisfy this upper bound for all small $c_-,{\epsilon'}$ and large $\alpha,c_+$, since $\widetilde{E}_\alpha\leq E_\alpha+C'\alpha^{-\frac{4}{29}}\leq e^\mathrm{Pek}+C'\alpha^{-\frac{4}{29}}< \frac{e^\mathrm{Pek}}{2}<0$ for $\alpha$ large enough and a suitable $C'$, and since we have by Eqs.~(\ref{Equation-Boundedness from below_first line}) and~(\ref{Equation-Boundedness from below}) for all $t>0$
\begin{align}
\label{Equation-FirstContradiction}
 \braket{\Psi_{\alpha,-}|\mathbb{H}_\Lambda|\Psi_{\alpha,-}}&\! \geq\!  \braket{\Psi_{\alpha,-}|\! -\! \frac{d}{t^2}\! -\! t\!  \left(\mathcal{N}\! +\! \alpha^{-2}\right)\! |\Psi_{\alpha,-}}\! \geq \! -\frac{d}{t^2}\! -\! t\! \left(c_-\! +\! 2{\epsilon'}\! +\! \alpha^{-2}\right)\! \geq \! -\frac{e^\mathrm{Pek}}{2},\\
\label{Equation-SecondContradiction}
\braket{\Psi_{\alpha,+}|\mathbb{H}_\Lambda|\Psi_{\alpha,+}}&\geq \braket{\Psi_{\alpha,+}|-d+\frac{1}{2}\mathcal{N}|\Psi_{\alpha,+}}\geq -d+\frac{1}{2}(c_+-2{\epsilon'})\geq 0,
\end{align} 
where the last inequality in Eq.~(\ref{Equation-FirstContradiction}), respectively Eq.~(\ref{Equation-SecondContradiction}), holds for small $c_-,{\epsilon'}$ and large $\alpha,c_+$ with the concrete choice $t:=\left(\frac{d}{c_-+2{\epsilon'}+\alpha^{-2}}\right)^{\frac{1}{3}}$. Using again that the right hand side of Eq.~(\ref{Equation-IMS for lower and upper bound}) is bounded by $\braket{\Psi^\bullet_\alpha|\mathbb{H}_\Lambda|\Psi^\bullet_\alpha}+C\alpha^{-\frac{7}{2}}$ together with Eqs.~(\ref{Equation-FirstContradiction}) and~(\ref{Equation-SecondContradiction}), and the fact that $\mathbb{H}_\Lambda\geq E_\alpha$ and $E_\alpha\leq e^\mathrm{Pek}$, yields furthermore
\begin{align*}
(1-Z_\alpha^2)\left(E_\alpha-\frac{e^\mathrm{Pek}}{2}\right)+Z_\alpha^2 E_\alpha\leq (1-Z_\alpha^2)\frac{e^\mathrm{Pek}}{2}+Z_\alpha^2 E_\alpha\leq \widetilde{E}_\alpha+C\alpha^{-\frac{7}{2}},
\end{align*}
and therefore $-(1-Z_\alpha^2)\frac{e^\mathrm{Pek}}{2}\leq \widetilde{E}_\alpha-E_\alpha+C\alpha^{-\frac{7}{2}}\underset{\alpha\rightarrow \infty}{\longrightarrow}  0$. Since $e^\mathrm{Pek}<0$, this immediately implies $Z_\alpha \underset{\alpha\rightarrow \infty}{\longrightarrow}  1$.
\end{proof}

Regarding the next localization step in Lemma \ref{Lemma-Control of variance}, let us introduce for given $R$ and $\epsilon>0$ satisfying $R>2\epsilon$ the function $K_R\left(\rho\right):=\iint \chi^{\epsilon}\left(R-\epsilon\leq |x-y|\right)\mathrm{d}\rho(x)\mathrm{d}\rho(y)$, which measures how sharply the mass of the measure $\rho$ is concentrated. It will be convenient in the second part to have $K_R$ defined for arbitrary $\epsilon\geq 0$ even though we only need it for $\epsilon=0$ in the following. We also define the function $F_R\left(\rho\right):=\chi^{\frac{\delta}{3}}\Big(K_R\left(\rho\right)\leq \frac{2\delta}{3}\Big)$ for $R,\delta>0$, as well as the states
\begin{align}
\label{Equation:Second Localization}
\Psi''_\alpha:=Z_{R,\alpha}^{-1}\widehat{F}_R\Psi'_\alpha,
\end{align}
where $\Psi'_\alpha$ is as in Lemma \ref{Lemma-Admissible sequence} and $Z_{R,\alpha}:=\|\widehat{F}_R\Psi'_\alpha\|$. Since $\Psi'_\alpha$ satisfies $\mathrm{supp}\left(\Psi'_\alpha\right)\subset B_L(0)$, we have $\mathrm{supp}\left(\Psi''_\alpha\right)\subset B_L(0)$ as well. Furthermore $\chi\left(\widehat{K}_R\leq \delta\right)\Psi''_\alpha=\Psi''_\alpha$. Heuristically this means that we can restrict our attention to phonon configurations that concentrate in a ball of fixed radius $R$.

\begin{lem}
\label{Lemma-Control of variance}
Let $\Psi'_\alpha$ be the sequence from Lemma \ref{Lemma-Admissible sequence}, and let $\epsilon\geq 0$ and $\delta>0$ be given constants. Then there exists a $\alpha$ independent $R>0$, such that the states $\Psi''_\alpha$ defined in Eq.~(\ref{Equation:Second Localization}) satisfy $\braket{\Psi''_\alpha|\mathbb{H}_\Lambda|\Psi''_\alpha}-\widetilde{E}_\alpha\lesssim \alpha^{-\frac{7}{2}}$, where $\widetilde{E}_\alpha$ is defined in Eq.~(\ref{Equation-Energy_of_given_sequence}). Furthermore, $Z_{R,\alpha}\underset{\alpha \rightarrow \infty}{\longrightarrow}1$.
\end{lem}
\begin{proof}
Since $\mathcal{P}:=\{F_R,G_R\}$ with $G_R:=\sqrt{1-F_R^2}=\chi^{\frac{\delta}{3}}\left(\frac{2\delta}{3}\leq K_R\left(\rho\right)\right)$ is a partition of unity, we obtain by Lemma \ref{Lemma-IMS Formula}
\begin{align}
\label{Equation-IMS for F_L}
\braket{\widehat{F}_R\Psi'_\alpha|\mathbb{H}_\Lambda|\widehat{F}_R\Psi'_\alpha}\! +\! \braket{\widehat{G}_R\Psi'_\alpha|\mathbb{H}_\Lambda |\widehat{G}_R\Psi'_\alpha}\! \leq\!  \braket{\Psi'_\alpha|\mathbb{H}_\Lambda |\Psi'_\alpha} \! +\! c\, \alpha^{-\frac{7}{2}} V_\Omega\left(\mathcal{P}\right)\!   \big\langle\!  \Psi'_\alpha\big|\sqrt{c_+\! +\! \alpha^{-2} }\,\big|\Psi'_\alpha  \big\rangle
\end{align}
with $\Omega:=\{\rho:\int \mathrm{d}\rho\leq c_+\}$, where we have used $\chi\left(\mathcal{N}\leq c_+\right)\Psi'_\alpha=\Psi'_\alpha$ and $\Lambda\leq \alpha$. Since $\frac{\mathrm{d}}{\mathrm{d}x}\chi^{\frac{\delta}{3}}\left(\frac{2\delta}{3}\leq x \right)$ and $\frac{\mathrm{d}}{\mathrm{d}x}\chi^{\frac{\delta}{3}}\left(x \leq \frac{2\delta}{3}\right)$ are bounded by some $\delta$-dependent constant $D$, we have for all $\rho\in \Omega$ and $\rho':=\rho+\alpha^{-2}\delta_z$ with $z\in \mathbb{R}^3$, and $R>2\epsilon$, the estimate
\begin{align*}
\left|F_R\left(\rho'\right)-F_R\left(\rho\right)\right|&\leq D\left|K_R(\rho')-K_R(\rho)\right|=2D\alpha^{-2}\int \chi^{\epsilon}\left(R-\epsilon\leq |y-z|\right)\mathrm{d}\rho(y)\\
&\leq 2D\alpha^{-2} c_+,
\end{align*}
and the same result holds for $G_R$. Therefore we have by Eq.~(\ref{Equation-IMS for F_L}) and Lemma \ref{Lemma-Admissible sequence}
\begin{align}
\label{Equation-IMSapplication}
\braket{\widehat{F}_R\Psi'_\alpha|\mathbb{H}_\Lambda |\widehat{F}_R\Psi'_\alpha}+\braket{\widehat{G}_R\Psi'_\alpha|\mathbb{H}_\Lambda |\widehat{G}_R\Psi'_\alpha}\leq \braket{\Psi'_\alpha|\mathbb{H}_\Lambda |\Psi'_\alpha}+C_1\alpha^{-\frac{7}{2}}\leq \widetilde{E}_\alpha+C_2 \alpha^{-\frac{7}{2}}
\end{align}
for suitable constants $C_1,C_2>0$. Since $\|\widehat{F}_R\Psi'_\alpha\|^2+\|\widehat{G}_R\Psi'_\alpha\|^2=1$, this means that we either have $\braket{\Psi''_\alpha|\mathbb{H}_\Lambda|\Psi''_\alpha}\leq \widetilde{E}_\alpha+C_2\alpha^{-\frac{7}{2}}$ or $\braket{\widetilde{\Psi}_\alpha|\mathbb{H}_\Lambda|\widetilde{\Psi}_\alpha}\leq \widetilde{E}_\alpha+C_2\alpha^{-\frac{7}{2}}$, where $\widetilde{\Psi}_\alpha:=\|\widehat{G}_R\Psi'_\alpha\|^{-1}\widehat{G}_R\Psi'_\alpha$. In the following we are going to rule out the second case for $R$ and $\alpha$ large enough, to be precise we are going to verify $\braket{\widetilde{\Psi}_\alpha|\mathbb{H}_\Lambda|\widetilde{\Psi}_\alpha}> \widetilde{E}_\alpha+d\alpha^{-\frac{4}{29}}$ for any $d>0$ and large enough $R$ and $\alpha$ by contradiction. In order to do this, let us assume $\braket{\widetilde{\Psi}_\alpha|\mathbb{H}_\Lambda |\widetilde{\Psi}_\alpha}\leq \widetilde{E}_\alpha+d\alpha^{-\frac{4}{29}}$. Since $\widetilde{E}_\alpha\leq E_\alpha+C\alpha^{-\frac{4}{29}}\leq e^\mathrm{Pek}+C\alpha^{-\frac{4}{29}}$ by assumption for a suitable constant $C$, $\widetilde{\Psi}_\alpha$ satisfies the assumptions of Theorem \ref{Theorem-Coherent States} with $\delta e:=(d+C)\alpha^{-\frac{4}{29}}$. Hence there exists a measure $\mu$ such that Eq.~(\ref{Equation-Coherent state formula for F}) holds. By the support properties of $G_R$ we obtain
\begin{align}
\label{Equation-Contradiction}
\frac{\delta}{3}\leq \big\langle \widetilde{\Psi}_\alpha\big|\widehat{K}_R\big|\widetilde{\Psi}_\alpha\big\rangle &=\int K_R\left(\left|\varphi_x^\mathrm{Pek}\right|^2\right)\, \mathrm{d}\mu + O_{\alpha\rightarrow \infty}\left(\alpha^{-\frac{2}{29}}\right)\\
& = K_R\left(\left|\varphi^\mathrm{Pek}\right|^2\right) + O_{\alpha\rightarrow \infty}\left(\alpha^{-\frac{2}{29}}\right).
\end{align}
Since $\lim_{R\rightarrow \infty}K_R\left(\left|\varphi^\mathrm{Pek}\right|^2\right)=0$, Eq.~(\ref{Equation-Contradiction}) is a contradiction for large $R$ and $\alpha$, and consequently we have $\braket{\widetilde{\Psi}_\alpha|\mathbb{H}_\Lambda|\widetilde{\Psi}_\alpha}> \widetilde{E}_\alpha+d\alpha^{-\frac{4}{29}}$ for such $R$ and $\alpha$. In combination with Eq.~(\ref{Equation-IMSapplication}) this furthermore yields
\begin{align*}
Z_{R,\alpha}^2 E_\alpha\! +\! (1\! -\! Z_{R,\alpha}^2)\! \left(E_\alpha\! +\! d\alpha^{-\frac{4}{29}}\right) \leq Z_{R,\alpha}^2 E_\alpha\! +\! (1\! -\! Z_{R,\alpha}^2)\! \left(\widetilde{E}_\alpha\! +\! d\alpha^{-\frac{4}{29}}\right)\leq \widetilde{E}_\alpha\! +\! C_2 \alpha^{-\frac{7}{2}},
\end{align*}
and therefore $1-Z_{R,\alpha}^2\leq \frac{\alpha^{\frac{4}{29}}}{d}\left(\widetilde{E}_\alpha-E_\alpha+C_2\alpha^{-\frac{7}{2}}\right)\leq \frac{1}{d}+\frac{C_2}{d}\alpha^{\frac{4}{29}-\frac{7}{2}}$. Since this holds for any $d>0$ and $\alpha$ large enough, we conclude that $Z_{R,\alpha}\underset{\alpha \rightarrow \infty}{\longrightarrow}1$.
\end{proof}

The previous localizations in the Lemmas \ref{Lemma-Admissible sequence} and \ref{Lemma-Control of variance} will allow us to control the energy cost of the main localization in the proof of Lemma \ref{Lemma-Existence of a condensate}. Before we come to Lemma \ref{Lemma-Existence of a condensate} we need to define the regularized median $m_q$ in Definition \ref{Definition-Regularized Median} and verify Lemma \ref{Lemma-Variation of the regularized median}, which provides an upper bound on the variation $V_\Omega\left(\mathcal{P}\right)$ for partitions $\mathcal{P}=\{F_j:j\in J\}$ of the form $F_j(\rho)=f_j\left(m_q(\rho)\right)$. The following auxiliary Lemmas \ref{Lemma-Regular Omega} , \ref{Lemma-Symmetric difference} and \ref{Lemma-Variation of the quantiles} will be useful in proving Lemma \ref{Lemma-Variation of the regularized median}.

\begin{lem}
\label{Lemma-Regular Omega}
Let us define the set $\Omega_\mathrm{reg}$ as the set of all $\rho\in \mathcal{M}\left(\mathbb{R}^3\right)$ satisfying 
\begin{align*}
\int_{x_i=t}\mathrm{d}\rho(x)\leq \alpha^{-2}
\end{align*}
for all $t\in \mathbb{R}$ and $i\in \{1,2,3\}$. Then $\widehat{\mathds{1}_{\Omega_\mathrm{reg}}}\Psi=\Psi$ for all $\Psi\in \mathcal{F}\left(L^2\! \left(\mathbb{R}^3\right)\right)$.
\end{lem}
\begin{proof}
For given $x=(x^1,\dots ,x^n)\in \mathbb{R}^{3\times n}$, let us define the measure $\rho_x:=\alpha^{-2}\sum_{k=1}^n \delta_{x^k}$. Note that $\rho_x\notin \Omega_\mathrm{reg}$ if and only if there exists an $i\in \{1,2,3\}$ such that $x_i^k=x_i^{k'}$ for indices $k\neq k'$. Clearly the set of all such $x\in \mathbb{R}^{3\times n}$ has Lebesgue measure zero. Hence the multiplication operator by the function $(x^1,\dots ,x^n)\mapsto \mathds{1}_{\Omega_\mathrm{reg}}\left(\rho_x\right)$ is equal to the identity on $L^2_{\mathrm{sym}}\! \left(\mathbb{R}^{3\times n}\right)$, which concludes the proof according to Definition \ref{Definition-Hat operators}.
\end{proof}

\begin{lem}
\label{Lemma-Symmetric difference}
Let $\nu,\nu'$ be finite measures on $\mathbb{R}$ such that $\nu\left(\{t\}\right)\leq \epsilon$ and $\nu'\left(\{t\}\right)\leq \epsilon$ for all $t\in \mathbb{R}$, and let $x^\kappa(\nu)$ be the $\kappa$-quantile of the measure $\nu$ with $0\leq \kappa\leq 1$, to be precise $x^{\kappa}(\nu)$ is the supremum over all numbers $t\in \mathbb{R}$ satisfying $\int_{-\infty}^t\mathrm{d}\nu\leq \kappa \int \mathrm{d}\nu$, where we use the convention that the boundaries are included in the domain of integration $\int_{a}^b f\mathrm{d}\nu:=\int_{[a,b]}f\mathrm{d}\nu$. Then
\begin{align*}
\Big|\int_{-\infty}^{x^\kappa\left(\nu'\right)} \mathrm{d}\nu-\int_{-\infty}^{x^\kappa\left(\nu\right)} \mathrm{d}\nu\Big|\leq 2\|\nu'-\nu\|_{\mathrm{TV}}+\epsilon,
\end{align*}
where $\|\nu'-\nu\|_{\mathrm{TV}}:=\underset{\|f\|_\infty=1}{\sup}\left|\int f\, \mathrm{d}\nu'-\int f\, \mathrm{d}\nu\right|$.
\end{lem}
\begin{proof}
We estimate
\begin{align*}
&\int_{-\infty}^{x^\kappa\left(\nu'\right)} \mathrm{d}\nu-\int_{-\infty}^{x^\kappa\left(\nu\right)} \mathrm{d}\nu\leq \int_{-\infty}^{x^\kappa\left(\nu'\right)}\mathrm{d}\nu-\kappa \int \mathrm{d}\nu \leq \int_{-\infty}^{x^\kappa\left(\nu'\right)}\mathrm{d}\nu'+\|\nu'-\nu\|_{\mathrm{TV}}-\kappa \int \mathrm{d}\nu\\
&\ \leq \kappa \int \mathrm{d}\nu'+\epsilon+\|\nu'-\nu\|_{\mathrm{TV}}-\kappa \int \mathrm{d}\nu\leq 2\|\nu'-\nu\|_{\mathrm{TV}}+\epsilon,
\end{align*}
where we have used $\int_{-\infty}^{x^\kappa\left(\nu\right)} \mathrm{d}\nu\geq \kappa \int \mathrm{d}\nu$ and $\int_{-\infty}^{x^\kappa\left(\nu'\right)}\mathrm{d}\nu'\leq \kappa \int \mathrm{d}\nu'+\epsilon$. The bound from below can be obtained by interchanging the role of $\nu$ and $\nu'$.
\end{proof}

\begin{defi}
\label{Definition-Regularized Median}
Let $x^\kappa\! \left(\nu\right)$ be the $\kappa$-quantile of a measure $\nu$ on $\mathbb{R}$ defined in Lemma \ref{Lemma-Symmetric difference} and let us denote $K_q(\nu):=[x^{\frac{1}{2}-q}(\nu),x^{\frac{1}{2}+q}(\nu)]$ for $0< q< \frac{1}{2}$. Then we define 
\begin{align}
\label{Equation-Definition_Marginal_Median}
m_q(\nu):=\frac{1}{\int_{K_q(\nu)} \mathrm{d}\nu}\int_{K_q(\nu)}t\, \mathrm{d}\nu(t)\in \mathbb{R}
\end{align}
for $\nu\neq 0$ and $m_q(0):=0$. Furthermore we define for a measure $\rho$ on $\mathbb{R}^3$ the regularized median as $m_q(\rho):=\Big(m_q(\rho_1),m_q(\rho_2),m_q(\rho_3)\Big)\in \mathbb{R}^3$, where $\rho_1,\rho_2$ and $\rho_3$ are the marginal measures of $\rho$ defined by $\rho_i\left(A\right):=\rho\left(\left[x_i\in A\right]\right)$.
\end{defi}
Note that $x^\kappa(\nu)$ is the largest value, such that both $\int_{-\infty}^{x^\kappa(\nu)}\mathrm{d}\nu\geq \kappa \int \mathrm{d}\nu$ and $\int_{x^\kappa(\nu)}^\infty \mathrm{d}\nu\geq (1-\kappa)\int \mathrm{d}\nu$ hold. As an immediate consequence, we obtain that the expression in Eq.~(\ref{Equation-Definition_Marginal_Median}) is well-defined for $\nu\neq 0$ and $0<q<\frac{1}{2}$, since 
\begin{align}
\label{Equation-Mass_Lower_bound}
\int_{K_q(\nu)} \mathrm{d}\nu=\int_{-\infty}^{x^{\frac{1}{2}+q}(\nu)}\mathrm{d}\nu+\int_{x^{\frac{1}{2}-q}(\nu)}^{\infty}\mathrm{d}\nu-\int \mathrm{d}\nu\geq 2q\int \mathrm{d}\nu>0.
\end{align}

\begin{lem}
\label{Lemma-Variation of the quantiles}
Given constants $R,c>0$ and $0<\delta<\frac{c^2}{2}$, let $\rho$ satisfy $c\leq\int \mathrm{d}\rho$ and $\underset{|x-y|\geq R}{\int \int} \mathrm{d}\rho(x)\mathrm{d}\rho(y)\leq \delta$ and let $q$ be a constant satisfying $0<q\leq \frac{1}{2}-\frac{\delta}{c^2}$. Then we have for all $i\in \{1,2,3\}$ that $x^{\frac{1}{2}}(\rho_i)-R\leq x^{\frac{1}{2}-q}(\rho_i)\leq x^{\frac{1}{2}+q}(\rho_i)\leq x^{\frac{1}{2}}(\rho_i)+R$.
\end{lem}
\begin{proof}
Since $x^\kappa$ is translation covariant, i.e. $x^\kappa\left(\nu(\cdot-t)\right)=x^\kappa(\nu)+t$, we can assume w.l.o.g. that $x^{\frac{1}{2}}(\rho_i)=0$ for $i\in \{1,2,3\}$. Then
\begin{align*}
\delta\geq\!  \! \underset{|x-y|\geq R}{\int \int}\! \!  \mathrm{d}\rho(x)\mathrm{d}\rho(y)\geq 2\int_{x_i\geq 0}\! \! \mathrm{d}\rho(x)\int_{y_i\leq -R}\! \! \mathrm{d}\rho(y)\geq \int\! \mathrm{d}\rho\int_{y_i\leq -R}\! \! \mathrm{d}\rho(y)\geq c\int_{y_i\leq -R}\! \! \mathrm{d}\rho(y),
\end{align*}
where we have used that $x^{\frac{1}{2}}(\rho_i)=0$ and $\int\mathrm{d}\rho\geq c$ in the last two inequalities. Hence
\begin{align*}
\int_{y_i\leq -R}\mathrm{d}\rho(y)\leq \frac{\delta}{c}\leq \frac{\delta}{c^2}\int \mathrm{d}\rho\leq \kappa \int \mathrm{d}\rho
\end{align*}
for all $\kappa\geq  \frac{\delta}{c^2}$ and consequently we have $-R\leq x^\kappa(\rho_i)$ for all such $\kappa$ by the definition of $x^\kappa(\rho_i)$. Similarly we obtain $x^\kappa(\rho_i)\leq R$ for all $\kappa$ satisfying $\kappa\leq 1-\frac{\delta}{c^2}$. Therefore $|x^{\frac{1}{2}\pm q}(\rho_i)|\leq R$ for $q\leq \frac{1}{2}-\frac{\delta}{c^2}$.
\end{proof}

\begin{lem}
\label{Lemma-Variation of the regularized median}
Given constants $R,c>0$ and $0<\delta<\frac{c^2}{2}$, let $\Omega$ be the set of $\rho\in \Omega_\mathrm{reg}$ satisfying $c\leq\int \mathrm{d}\rho$ and $\underset{|x-y|\geq R}{\int \int} \mathrm{d}\rho(x)\mathrm{d}\rho(y)\leq \delta$. Then 
\begin{align*}
\left|m_q\left(\rho+\alpha^{-2}\delta_x\right)-m_q\left(\rho\right)\right|\lesssim \frac{R}{c\alpha^2 q}
\end{align*}
for all $\rho\in \Omega$, $x\in \mathbb{R}^3$ and $0<q<\frac{1}{2}-\frac{\delta}{c^2}$, where $m_q$ is defined in Definition \ref{Definition-Regularized Median}.
\end{lem}
\begin{proof}
Since $m_q$ acts translation covariant on any $\rho\neq 0$, i.e. $m_q\left(\rho(\cdot-y)\right)=m_q(\rho)+y$, we can assume w.l.o.g. that $x^{\frac{1}{2}}(\rho_i)=0$ for $i\in \{1,2,3\}$. By Lemma \ref{Lemma-Variation of the quantiles} we therefore obtain $|x^{\frac{1}{2}\pm q}(\rho_i)|\leq R$ for $\rho\in \Omega$ and $0<q\leq \frac{1}{2}-\frac{\delta}{c^2}$. Note that the marginal measures $\rho_i$ and $\rho'_i$, where $\rho':=\rho+\alpha^{-2}\delta_x$, satisfy $\rho_i\left(\{y\}\right)\leq \alpha^{-2}$ and $\rho'_i\left(\{y\}\right)\leq 2\alpha^{-2}$ by our assumption $\rho\in \Omega_\mathrm{reg}$. Therefore $x^{\kappa_*}(\rho_i)\leq x^\kappa(\rho'_i)\leq x^{\kappa^*}(\rho_i)$ for $\rho\in \Omega$ and $\kappa>0$, with $\kappa_*:=\kappa-2\frac{1}{c}\alpha^{-2}$ and $\kappa^*:=\kappa+3\frac{1}{c}\alpha^{-2}$. In particular, this implies $|x^{\frac{1}{2}\pm q}(\rho'_i)|\leq R$ for $0<q<1/2-\delta/c^2$ and $\alpha$ large  enough. In the following it will be convenient to write the difference $m_q\left(\rho'_i\right)\!-\!m_q(\rho_i)$ as
\begin{align}
\label{Equation-DifferenceOfMedians}
\left(\frac{1}{\int_{K_q\left(\rho'_i\right)}\mathrm{d}\rho'_i}\! -\! \frac{1}{\int_{K_q\left(\rho_i\right)}\mathrm{d}\rho_i}\right)\int_{K_q(\rho'_i)} t\, \mathrm{d}\rho'_i(t)+\frac{1}{\int_{K_q\left(\rho_i\right)}\mathrm{d}\rho_i}\left(\int_{K_q(\rho'_i)}t\, \mathrm{d}\rho'_i(t)\!- \! \!\int_{K_q(\rho_i)} t\, \mathrm{d}\rho_i(t)\! \right).
\end{align}
Making use of $\int_{K_q\left(\rho_i\right)}\mathrm{d}\rho_i\geq 2qc$, see Eq.~(\ref{Equation-Mass_Lower_bound}), and $K_q(\rho'_i)\subset [-R,R]$ for all $\rho\in \Omega$, we can estimate the individual terms in Eq.~(\ref{Equation-DifferenceOfMedians}) by
\begin{align*}
&\left|\left(\frac{1}{\int_{K_q\left(\rho'_i\right)}\mathrm{d}\rho'_i}\! -\! \frac{1}{\int_{K_q\left(\rho_i\right)}\mathrm{d}\rho_i}\right)\int_{K_q(\rho'_i)} t\, \mathrm{d}\rho'_i(t)\right|\leq R\frac{\left|\int_{K_q\left(\rho'_i\right)}\mathrm{d}\rho'_i-\int_{K_q\left(\rho_i\right)}\mathrm{d}\rho_i\right|}{2qc},\\
&\left|\frac{1}{\int_{K_q\left(\rho_i\right)}\mathrm{d}\rho_i}\left(\int_{K_q(\rho'_i)}t\, \mathrm{d}\rho'_i(t)\!- \! \!\int_{K_q(\rho_i)} t\, \mathrm{d}\rho_i(t)\! \right)\right|\leq \frac{\left|\int_{K_q(\rho'_i)}t\, \mathrm{d}\rho'_i(t)\!- \! \!\int_{K_q(\rho_i)} t\, \mathrm{d}\rho_i(t)\right|}{2qc}.
\end{align*}
Note that $K_q(\rho_i)$ is contained in $[-R,R]$ as well and consequently $t$ is bounded by $R$ on the subset $K_q(\rho_i)\cup K_q(\rho'_i)$. In order to verify the statement of the Lemma, it is therefore sufficient to prove that $\left|\int_{K_q(\rho'_i)}f(t)\, \mathrm{d}\rho'_i(t)\!- \! \!\int_{K_q(\rho_i)} f(t)\, \mathrm{d}\rho_i(t)\right|\lesssim \alpha^{-2}\|f\|_\infty $ for an arbitrary measurable and bounded $f:\mathbb{R}\rightarrow \mathbb{R}$. We estimate
\begin{align*}
&\left|\int_{K_q(\rho'_i)}f(t)\, \mathrm{d}\rho'_i(t)\!- \! \!\int_{K_q(\rho_i)} f(t)\, \mathrm{d}\rho_i(t)\right|\leq \left|\int_{K_q(\rho'_i)}f(t)\, \mathrm{d}\rho'_i(t)\!- \! \!\int_{K_q(\rho'_i)} f(t)\, \mathrm{d}\rho_i(t)\right|\\
&\ \  +\left|\int_{K_q(\rho'_i)}f(t)\, \mathrm{d}\rho_i(t)\!- \! \!\int_{K_q(\rho_i)} f(t)\, \mathrm{d}\rho_i(t)\right|\leq \|f\|_\infty \left(\|\rho'_i-\rho_i\|_{\mathrm{TV}}+\int_{K_q(\rho'_i)\Delta K_q(\rho_i)}\mathrm{d}\rho_i\right),
\end{align*}
where $A\Delta B:=\left(A\cup B\right)\setminus \left(A\cap B\right)$ is the symmetric difference. Note that $\|\rho'_i-\rho_i\|_{\mathrm{TV}}=\alpha^{-2}$. Furthermore we can estimate the expression $\int_{K_q(\rho'_i)\Delta K_q(\rho_i)} \mathrm{d}\rho_i$ by
\begin{align*}
\left|\int_{-\infty}^{x^{\frac{1}{2}-q}(\rho'_i)}\mathrm{d}\rho_i-\int_{-\infty}^{x^{\frac{1}{2}-q}(\rho_i)}\mathrm{d}\rho_i\right|+\left|\int_{-\infty}^{x^{\frac{1}{2}+q}(\rho'_i)}\mathrm{d}\rho_i-\int_{-\infty}^{x^{\frac{1}{2}+q}(\rho_i)}\mathrm{d}\rho_i\right|.
\end{align*}
Since the distributions $\rho_i$ and $\rho_i'$ satisfy the assumptions of Lemma \ref{Lemma-Symmetric difference} with $\epsilon:=2\alpha^{-2}$, we conclude that every term in the sum above is bounded by $2\|\rho'-\rho\|_{\mathrm{TV}}+\epsilon=4\alpha^{-2}$.
\end{proof}

Before we state the central Lemma \ref{Lemma-Existence of a condensate}, let us verify in the subsequent Lemma \ref{Lemma-Condensation} that low energy states with a localized median necessarily satisfy (complete) condensation with respect to a minimizer of the Pekar functional.

\begin{lem}
\label{Lemma-Condensation}
Given a constant $C>0$, there exists a constant $T>0$, such that
\begin{align*}
 \Big\langle \Psi\Big|W_{\varphi^\mathrm{Pek}}^{-1}\, \mathcal{N}W_{\varphi^\mathrm{Pek}}\Big| \Psi\Big\rangle\leq  T\left(\alpha^{-\frac{2}{29}}+q+\epsilon\right)
\end{align*}
for all states $\Psi$ satisfying $\braket{\Psi|\mathbb{H}_K|\Psi}\leq e^\mathrm{Pek}+\alpha^{-\frac{4}{29}}$ with $K\geq \alpha^{\frac{8}{29}}$ and $\widehat{\mathds{1}_{\Omega^*}}\Psi=\Psi$, where $\Omega^*$ is the set of all $\rho$ satisfying $\int \mathrm{d}\rho\leq C$ and $|m_q(\rho)|\leq \epsilon$ with $q,\epsilon>0$.
\end{lem}
\begin{proof}
Let us begin by defining the functions
\begin{align}
\label{Equation-Definition H_i}
P_i^\epsilon(\rho):=\left(\frac{1}{2}\int \mathrm{d}\rho\right)^2-\int_{x_i\leq \epsilon}\mathrm{d}\rho(x)\int_{y_i\geq -\epsilon}\mathrm{d}\rho(y).
\end{align}
Observe that $|m_{q}(\rho)|\leq \epsilon$ implies $-\epsilon\leq x^{\frac{1}{2}+q}(\rho_i)$ and $x^{\frac{1}{2}-q}(\rho_i)\leq \epsilon$ for all such $\rho$ which additionally satisfy $\rho\neq 0$, see Definition \ref{Definition-Regularized Median}. Therefore $P^\epsilon_i(\rho)\leq \left(\int \mathrm{d}\rho\right)^2 \big(\frac{1}{4}-\left(\frac{1}{2}-q\right)^2\big)\lesssim q$ for all $\rho\in {\Omega^*}$, and consequently the measure $\mu$ from Theorem \ref{Theorem-Coherent States} corresponding to the state $\Psi$ satisfies $\int P^\epsilon_i\left(\left|\varphi^\mathrm{Pek}_x\right|^2\right)\, \mathrm{d}\mu(x)\leq \big\langle \Psi\big|\widehat{P}^\epsilon_i\, \big|\Psi\big\rangle+D \alpha^{-\frac{2}{29}}\lesssim q+\alpha^{-\frac{2}{29}}$ for a suitable $D>0$, where we have used Eq.~(\ref{Equation-Coherent state formula for F}) in the first inequality. Furthermore we know that $\|\varphi^\mathrm{Pek}_x-\varphi^\mathrm{Pek}\|^2\lesssim  \sum_{i=1}^3 P^\epsilon_i\left(\left|\varphi^\mathrm{Pek}_x\right|^2\right)+\epsilon$ by Lemma \ref{Lemma-Results on the Pekar minimizer}, hence
\begin{align*}
\int \|\varphi^\mathrm{Pek}_x\! -\! \varphi^\mathrm{Pek}\|^2\, \mathrm{d}\mu(x)\lesssim \sum_{i=1}^3\int P^\epsilon_i\left(\left|\varphi^\mathrm{Pek}_x\right|^2\right)\, \mathrm{d}\mu(x)\! +\! \epsilon\lesssim q\! + \! \alpha^{-\frac{2}{29}}\! +\! \epsilon.
\end{align*}
Therefore Eq.~(\ref{Equation-Translated F version}) immediately concludes the proof of Eq.~(\ref{Equation-Weak Codensation}).
\end{proof}

\begin{lem}
\label{Lemma-Existence of a condensate}
Given $0<\sigma\leq \frac{1}{4}$, let $\Lambda$ and $L$ be as in Theorem \ref{Theorem-Strong cut-off}. Then there exist states $\Psi_\alpha'''$ satisfying $\braket{\Psi'''_\alpha|\mathbb{H}_\Lambda|\Psi'''_\alpha}-E_\alpha\lesssim \alpha^{-2(1+\sigma)}$, $\mathrm{supp}\left(\Psi_\alpha'''\right)\subset B_{4L}(0)$ and
\begin{align}
\label{Equation-Weak Codensation}
 \Big\langle \Psi'''_\alpha\Big|W_{\varphi^\mathrm{Pek}}^{-1}\, \mathcal{N}W_{\varphi^\mathrm{Pek}}\Big| \Psi'''_\alpha\Big\rangle\lesssim \alpha^{-\frac{2}{29}},
\end{align}
where $W_{\varphi^\mathrm{Pek}}$ is the Weyl operator corresponding to the Pekar minimizer $\varphi^\mathrm{Pek}$.
\end{lem}
\begin{proof}
It is clearly sufficient to consider only the case $\alpha \geq \alpha_0$ for a suitable (large) $\alpha_0$, 
since we can always re-define $\Psi'''_\alpha:=\Psi$ for $\alpha<\alpha_0$ where $\Psi$ is an arbitrary state satisfying $\mathrm{supp}\left(\Psi\right)\subset B_{4L}(0)$. In the following let us use the concrete choice $\Psi_\alpha^\bullet:=\Psi_\alpha^\diamond$ for the sequence in Eq.~(\ref{Equation-Energy_of_given_sequence}), where $\Psi_\alpha^\diamond$ is defined in in Theorem \ref{Theorem-Strong cut-off}, which is a valid choice since it satisfies the assumptions $\mathrm{supp}\left(\Psi_\alpha^\diamond\right)\subset B_L(0)$ and $\widetilde{E}_\alpha-E_\alpha\lesssim \alpha^{-2(1+\sigma)}\leq \alpha^{-\frac{4}{29}}$. Furthermore let $\{\chi_z:z\in \mathbb{Z}^3\}$ be a smooth (quadratic) partition of unity on $\mathbb{R}^3$, i.e. $0\leq \chi_z\leq 1$ and $\sum_{z\in \mathbb{Z}^3}\chi_z^2=1$, with $\chi_z(x)=\chi_0(x-z)$ and $\mathrm{supp} \left(\chi_0\right)\subset B_1(0)$. Then we define for $z\in \mathbb{Z}^3$ and $u,v\geq \frac{2}{29}$ with $u+v\leq \frac{1}{4}$ the function $F_z(\rho):=\chi_z\big(\alpha^{u}\, m_{\alpha^{-{v}}}(\rho)\big)$, as well as the states
\begin{align}
\label{Equation:Third Localization}
\Psi_{\alpha,z}:=Z_{\alpha,z}^{-1}\widehat{F}_z \Psi''_\alpha
\end{align}
with $Z_{\alpha,z}:=\|\widehat{F}_z \Psi''_\alpha\|$ and $\Psi''_\alpha$ as in Lemma \ref{Lemma-Control of variance} for $\epsilon=0$ and $0<\delta<\frac{c^2}{2}$, where $c:=c_-$ is as in Lemma \ref{Lemma-Admissible sequence}. Applying Lemma \ref{Lemma-IMS Formula} with respect to $\mathcal{P}:=\{F_z:z\in \mathbb{Z}^3\}$, where the functions $F_z$ are defined above Eq.~(\ref{Equation:Third Localization}) and $\Omega$ is defined as the set of all $\rho\in \Omega_\mathrm{reg}$ satisfying $c_-\leq\int \mathrm{d}\rho\leq c_+$ and $\underset{|x-y|\geq R}{\int \int} \mathrm{d}\rho(x)\mathrm{d}\rho(y)\leq \delta$, yields
\begin{align}
\label{Equation-Median Split}
\sum_{z\in \mathbb{Z}^3}Z^2_{\alpha,z}\braket{\Psi_{\alpha,z}|\mathbb{H}_\Lambda|\Psi_{\alpha,z}}\leq \braket{\Psi''_\alpha|\mathbb{H}_\Lambda|\Psi''_\alpha}+c\alpha^{-\frac{7}{2}}V_\Omega\left(\mathcal{P}\right)\sqrt{c_+ +\alpha^{-2}},
\end{align}
where we used Lemma \ref{Lemma-Regular Omega}, $\Lambda\leq \alpha$ and $\widehat{\mathds{1}_\Omega}\Psi''_\alpha=\Psi''_\alpha$ by the definition of $\Psi''_\alpha$ in Eq.~(\ref{Equation:Second Localization}). Since the support of $\chi_z$ only overlaps with the support of finitely many other $\chi_{z'}$, we obtain for $v>0$ and $\alpha$ large enough
\begin{align*}
V_\Omega\left(\mathcal{P}\right)&\lesssim \alpha^{4} \mathrm{sup}_{\rho\in \Omega,y\in \mathbb{R}^3}\, \mathrm{sup}_{z\in \mathbb{Z}^3}\left|\chi_z\big(\alpha^{u} m_{\alpha^{-{v}}}(\rho+\alpha^{-2}\delta_y)\big)-\chi_z\big(\alpha^{u} m_{\alpha^{-{v}}}(\rho)\big)\right|^2\\
&\lesssim \alpha^{2u+4}\mathrm{sup}_{\rho\in \Omega,y\in \mathbb{R}^3}\left|m_{\alpha^{-{v}}}(\rho+\alpha^{-2}\delta_y) -m_{\alpha^{-{v}}}(\rho)\right|^2\lesssim \alpha^{2(u+v)},
\end{align*}
where we have used $\mathrm{sup}_{z\in \mathbb{Z}^3}\left|\chi_z\left(y\right)-\chi_z\left(x\right)\right|\leq \big\|\nabla \chi_0\big\|_\infty |y-x|$ in the first inequality and Lemma \ref{Lemma-Variation of the regularized median} in the second one. Combining this with Eq.~(\ref{Equation-Median Split}) and the fact that $u+v\leq \frac{1}{4 }$ yields
\begin{align}
\label{Equation-Median Split Combined}
\sum_{z\in \mathbb{Z}^3}Z^2_{\alpha,z}\braket{\Psi_{\alpha,z}|\mathbb{H}_\Lambda|\Psi_{\alpha,z}}-\braket{\Psi''_\alpha|\mathbb{H}_\Lambda|\Psi''_\alpha}\lesssim \alpha^{-3}.
\end{align}
Since $\sum_{z\in \mathbb{Z}^3}Z_{\alpha,z}^2=1$, this in particular means that there exists a $z_\alpha\in \mathbb{Z}^3$ such that $\braket{\Psi_{\alpha,z_\alpha}|\mathbb{H}_\Lambda|\Psi_{\alpha,z_\alpha}}- E_\alpha\lesssim \alpha^{-2(1+\sigma)}$, and by the translation invariance of $\mathbb{H}_\Lambda$ we obtain $\braket{\Psi'''_{\alpha}|\mathbb{H}_\Lambda|\Psi'''_{\alpha}}- E_\alpha\lesssim \alpha^{-2(1+\sigma)}$ where $\Psi_\alpha'''=\mathcal{T}_{-\alpha^{-u}z_\alpha}\Psi_{\alpha,z_\alpha}$. Using the fact that $\mathds{1}_{\Omega^*}\Psi_\alpha'''=\Psi_\alpha'''$, where ${\Omega^*}$ is the set of all $\rho$ satisfying $\int \mathrm{d}\rho\leq c_+$ and $|m_{\alpha^{-v}}(\rho)|\leq \alpha^{-u}$, together with Lemma \ref{Lemma-Condensation}, immediately concludes the proof of Eq.~(\ref{Equation-Weak Codensation}).

Finally let us verify that $\mathrm{supp}\left(\Psi'''_\alpha\right)\subset B_{4L}(0)$. By the definition of $\Psi_\alpha'''=\mathcal{T}_{-\alpha^{-u}z_\alpha}\Psi_{\alpha,z_\alpha}$, and the fact that $\mathrm{supp}\left(\Psi_{\alpha,z_\alpha}\right)\subset B_{L}(0)$, it is clear that $\mathrm{supp}\left(\Psi'''_\alpha\right)\subset B_{L}(-w_\alpha)$ with $w_\alpha:=\alpha^{-u}z_\alpha$. In the following we show that $|w_\alpha|\leq 3L$ by contradiction for $\alpha$ large enough, and therefore $\mathrm{supp}\left(\Psi'''_\alpha\right)\subset B_{L+|w_\alpha|}(0)\subset B_{4L}(0)$. Assuming $|w_\alpha|>3L$, we obtain $\mathrm{supp}\left(\Psi'''_\alpha\right)\subset \mathbb{R}^3\setminus B_{2L}(0)$ and Corollary \ref{Corollary-Outside_Mass} consequently yields $\braket{\Psi_\alpha'''|\mathbb{H}_\Lambda|\Psi_\alpha'''}\geq E_\alpha+\braket{\Psi_\alpha'''|\mathcal{N}_{B_L(0)}|\Psi_\alpha'''}-\sqrt{\frac{D}{L}}$, where $\mathcal{N}_{B_L(0)}$ denotes the number operator in the ball $B_L(0)$ (as defined in Cor.~\ref{Corollary-Outside_Mass}). Defining $\varphi_L(x):=\chi\left(|x|\leq L\right)\varphi^\mathrm{Pek}(x)$, we further have
\begin{align*}
 \braket{\Psi'''_\alpha|\mathcal{N}_{B_L(0)}|\Psi_\alpha'''}&=\Big\langle \Psi_\alpha'''\Big|W_{\varphi^\mathrm{Pek}}^{-1}\left(\mathcal{N}_{B_L(0)}+a(\varphi_L)+a^\dagger(\varphi_L)+\|\varphi_L\|^2\right)W_{\varphi^\mathrm{Pek}}\Big| \Psi'''_\alpha\Big\rangle\\
 &\geq -\Big\langle \Psi'''_\alpha\Big|W_{\varphi^\mathrm{Pek}}^{-1}\mathcal{N} W_{\varphi^\mathrm{Pek}}\Big| \Psi'''_\alpha\Big\rangle+\frac{1}{2}\|\varphi_L\|^2\geq -D'\alpha^{-\frac{2}{29}}+\frac{1}{2}\|\varphi_L\|^2
\end{align*}
for a suitable constant $D'$, where we have used the operator inequality $\mathcal{N}_{B_L(0)}+a(\varphi_L)+a^\dagger(\varphi_L)+\|\varphi_L\|^2\geq -\mathcal{N}+\frac{1}{2}\|\varphi_L\|^2$ as well as Eq.~(\ref{Equation-Weak Codensation}). Therefore we obtain
\begin{align*}
\braket{\Psi_\alpha'''|\mathbb{H}_\Lambda|\Psi_\alpha'''}-E_\alpha\geq \frac{1}{2}\|\varphi_L\|^2-D'\alpha^{-\frac{2}{29}}-\sqrt{\frac{D}{L}}\underset{\alpha\rightarrow \infty}{\longrightarrow}\frac{1}{2}\|\varphi^\mathrm{Pek}\|^2>0,
\end{align*}
where we have used that $L=\alpha^{1+\sigma}\underset{\alpha\rightarrow \infty}{\longrightarrow}\infty$. This, however, is a contradiction to $\braket{\Psi'''_\alpha|\mathbb{H}_\Lambda|\Psi'''_\alpha}-E_\alpha\lesssim \alpha^{-2(1+\sigma)}$.
\end{proof}

Following the method in \cite{LNSS}, we are going to lift the weak condensation derived in Lemma \ref{Lemma-Existence of a condensate} to a strong one in the subsequent Theorem \ref{Theorem-Existence of a strong condensate}, which represents the main result of this section.

\begin{theorem}
\label{Theorem-Existence of a strong condensate}
Given $0<\sigma\leq \frac{1}{4}$ and $h<\frac{2}{29}$, let $\Lambda$ and $L$ be as in Theorem \ref{Theorem-Strong cut-off}. Then there exist states $\Psi_\alpha$ with $\braket{\Psi_\alpha|\mathbb{H}_\Lambda|\Psi_\alpha}- E_\alpha\lesssim \alpha^{-2(1+\sigma)}$ and $\mathrm{supp}\left(\Psi_\alpha\right)\subset B_{4L}(0)$, satisfying
\begin{align}
\label{Equation-StrongCondensation}
\chi\left( W_{\varphi^\mathrm{Pek}}^{-1}\, \mathcal{N}W_{\varphi^\mathrm{Pek}}\leq \alpha^{-h}\right)\Psi_\alpha=\Psi_\alpha
\end{align}
for large enough $\alpha$. 
\end{theorem}
\begin{proof}
Using the states $\Psi_\alpha'''$ from Lemma \ref{Lemma-Existence of a condensate}, we define for $0<\epsilon<\frac{1}{2}$
\begin{align*}
\Psi_\alpha:=Z_\alpha^{-1}\chi^{\epsilon}\left(\alpha^{h}W_{\varphi^\mathrm{Pek}}^{-1}\mathcal{N}W_{\varphi^\mathrm{Pek}} \leq \frac{1}{2}\right)\Psi'''_\alpha
\end{align*}
 where $Z_\alpha$ is a normalizing constant. Clearly the states $\Psi_\alpha$ satisfy the strong condensation property $\chi\left(W_{\varphi^\mathrm{Pek}}^{-1}\mathcal{N}W_{\varphi^\mathrm{Pek}} \leq \alpha^{-h}\right)\Psi_\alpha=\Psi_\alpha$. In order to control the energy cost of the localization with respect to the operator $W_{\varphi^\mathrm{Pek}}^{-1}\mathcal{N}W_{\varphi^\mathrm{Pek}}$, note that the partition $\mathcal{P}':=\{F',G'\}$ with $F'(\rho):=\chi^{\epsilon}\left(\alpha^{h}\int \mathrm{d}\rho \leq \frac{1}{2}\right)$ and $G'(\rho):=\chi^{\epsilon}\left(\frac{1}{2}\leq \alpha^{h}\int \mathrm{d}\rho\right)$ satisfies
\begin{align*}
\kappa:=V_{\mathcal{M}\left(\mathbb{R}^3\right)}\left(\mathcal{P}'\right)\lesssim \alpha^{4}  \mathrm{sup}_{\rho,x\in \mathbb{R}^3}\left|\alpha^h\int\mathrm{d}\! \left(\rho+\alpha^{-2}\delta_x\right)-\alpha^h\int \mathrm{d}\rho\right|^2=\alpha^{2h},
\end{align*} 
where we used $\left|\chi^{\epsilon}\left(y \leq \frac{1}{2}\right)-\chi^{\epsilon}\left(x \leq \frac{1}{2}\right)\right|\leq \big\|\frac{\mathrm{d}}{\mathrm{d}x} \chi^{\epsilon}\Big(\cdot \leq \frac{1}{2}\Big)\big\|_\infty |y-x|$ and the corresponding estimate for $\chi^{\epsilon}\left(\frac{1}{2}\leq \cdot \right)$. Therefore we obtain by Lemma \ref{Lemma-IMS Formula}, using $\Lambda\leq \alpha$,
\begin{align}
\nonumber
&Z_\alpha^{2}\! \braket{\Psi_\alpha|\mathbb{H}_\Lambda|\Psi_\alpha}\! +\! (1\! -\! Z_\alpha^2)\! \braket{\widetilde{\Psi}_\alpha|\mathbb{H}_\Lambda|\widetilde{\Psi}_\alpha}\! \leq \! \braket{\Psi'''_\alpha|\mathbb{H}_\Lambda|\Psi'''_\alpha}\! +\! c'\, \alpha^{-\frac{7}{2}}   \kappa\Big\langle  \Psi''' \Big| \sqrt{\mathcal{N}\! +\! 1}\Big| \Psi''' \Big\rangle\\
\label{Equation-Proof_Strong_Condensation}
&\ \ \ \ \leq E_\alpha+O_{\alpha\rightarrow \infty}\left(\alpha^{-2(1+\sigma)}\right)+O_{\alpha\rightarrow \infty}\left(\alpha^{2h-\frac{7}{2}}\right)=E_\alpha+O_{\alpha\rightarrow \infty}\left(\alpha^{-2(1+\sigma)}\right),
\end{align}
with $\widetilde{\Psi}_\alpha:=\sqrt{1-Z_\alpha^2}^{-1}\chi^{\epsilon}\left( \frac{1}{2}\leq \alpha^{h}W_{\varphi^\mathrm{Pek}}^{-1}\mathcal{N}W_{\varphi^\mathrm{Pek}} \right)\Psi'''_\alpha$. Making use of the trivial lower bound $E_\alpha\leq \braket{\widetilde{\Psi}_\alpha|\mathbb{H}_\Lambda|\widetilde{\Psi}_\alpha}$, Eq.~(\ref{Equation-Proof_Strong_Condensation}) implies $\braket{\Psi_\alpha|\mathbb{H}_\Lambda|\Psi_\alpha}\leq E_\alpha+Z_\alpha^{-2}O_{\alpha\rightarrow \infty}\left(\alpha^{-2(1+\sigma)}\right)$, which concludes the proof since
\begin{align*}
1-Z_\alpha^2&=\Big\langle \Psi'''_\alpha\Big|\chi^{\epsilon}\left(\frac{1}{2}\leq  \alpha^{h}W_{\varphi^\mathrm{Pek}}^{-1}\mathcal{N}W_{\varphi^\mathrm{Pek}} \right)^2\Big| \Psi'''_\alpha\Big\rangle\\
&\leq \frac{1}{\frac{1}{2}-\epsilon}\alpha^{h}\Big\langle \Psi'''_\alpha\Big| W_{\varphi^\mathrm{Pek}}^{-1}\mathcal{N}W_{\varphi^\mathrm{Pek}} \Big|\Psi'''_\alpha \Big\rangle\lesssim \frac{1}{\frac{1}{2}-\epsilon}\alpha^{h}\alpha^{-\frac{2}{29}}\underset{\alpha\rightarrow \infty}{\longrightarrow}0.
\end{align*}
\end{proof}

\section{Large Deviation Estimates for Strong Condensates}
\label{Section-Large Deviation Principle for strong Condensates}

In this Section we will derive a large deviation principle for states with suitably small particle number (compared to $\alpha^2$), which can be interpreted as  complete condensation with respect to the vacuum. We will show that such states are, up to an error which is exponentially small in $\alpha^2$, contained in the spectral subspace $\left|a(f)+a^\dagger(f)\right|\leq \epsilon$, see Eq.~(\ref{Equation-Large Deviations t}). Note that taking the point of condensation to be the vacuum is not a real restriction, since this is the case after applying a suitable Weyl transformation. Before we can formulate the main result of this section in Proposition \ref{Proposition-Large Deviations}, we need to introduce some notation.\\

For $0<\sigma<\frac{1}{4}$ let us define $\Lambda:=\alpha^{\frac{4}{5}(1+\sigma)}$, $\ell:=\alpha^{-4(1+\sigma)}$ and
\begin{align}
\label{Equation-Definition of strong projection}
\Pi :=\Pi^0_{\Lambda,\ell},
\end{align}
see Definition \ref{Definition-Pi}, and let us identify $\mathcal{F}\left(\Pi  L^2\! \left(\mathbb{R}^3\right)\right)$ with $L^2\! \left(\mathbb{R}^{N}\right)$ using the representation of real functions $\varphi=\sum_{n=1}^{N} \lambda_n \varphi_n \in \Pi  L^2\! \left(\mathbb{R}^3\right)$ by points $\lambda=(\lambda_1,\dots ,\lambda_{{N}})\in \mathbb{R}^{N}$, where $N:=\mathrm{dim}\Pi  L^2\! \left(\mathbb{R}^3\right)$ and $\{\varphi_1,\dots ,\varphi_{N}\}$ is a real orthonormal basis of $\Pi  L^2\! \left(\mathbb{R}^3\right)$. We choose this identification such that the annihilation operators $a_n:=a\left(\varphi_n\right)$ read
\begin{align}
\label{Equation-CreationAndAnnihilation}
a_n=\lambda_n+\frac{1}{2\alpha^2}\partial_{\lambda_n},
\end{align}
where $\lambda_n$ is the multiplication operator by the function $\lambda\mapsto \lambda_n$ on $L^2\! \left(\mathbb{R}^{N}\right)$. From the construction one readily checks that $N \lesssim (\Lambda / \ell)^3 \leq \alpha^p$ for suitable $p>0$.

In the following we will verify a large deviation principle for the density function $\rho(\lambda):=\gamma(\lambda,\lambda)$ corresponding to a density matrix $\gamma$ on $\mathcal{F}\left(\Pi  L^2\! \left(\mathbb{R}^3\right)\right)$ that satisfies the strong condensation condition
\begin{align}
\label{Equation-Strong Condensation}
\chi\left(\sum_{n=1}^{N} a_n^\dagger a_n\leq \alpha^{-h}\right)\gamma=\gamma
\end{align}
for some $h>0$. This result is comparable to \cite[Lemma C.2]{BS}. For this purpose,
we define a convenient norm $|\cdot |_\diamond$ on $\mathbb{R}^{N}$ in the subsequent Definition.

\begin{defi}
Let $|\lambda|:=\sqrt{\sum_{n=1}^{{N}} \lambda_n^2}$ denote the standard norm on $\mathbb{R}^{N}$ and let us define the norm $|\cdot|_\diamond$ on $\mathbb{R}^{N}$, using the identification $\varphi=\sum_{n=1}^{{N}} \lambda_n \varphi_n$, as
\begin{align}
\label{Equation-Definition_Diamond_Norm_in_Coordinates}
|\lambda|_\diamond:=2\sup_{x\in \mathbb{R}^3}\sqrt{\int_{B_1(x)}\left|\left(\left(-\Delta\right)^{-\frac{1}{2}}\varphi\right)(y)\right|^2\mathrm{d}y}.
\end{align}
\end{defi}
The norm $|\cdot|_\diamond$ will again appear naturally in Section \ref{Section-Properties of the Pekar Functional} where we investigate properties of the Pekar functional $\mathcal{F}^\mathrm{Pek}$ (see Eq.~(\ref{Equation-L2 diamond}) and the subsequent comment).

\begin{prop}
\label{Proposition-Large Deviations}
Let $0<s<\min\big\{\frac{h}{2},\frac{1}{5}(1-4\sigma)\big\}$ and $D>0$. Then there exist constants $\beta,\alpha_0>0$, such that we have for all $\alpha\geq \alpha_0$, $\epsilon\geq D\alpha^{-s}$ and $\gamma$ satisfying Eq.~(\ref{Equation-Strong Condensation})
\begin{align}
\label{Equation-Large Deviations}
\int_{|\lambda|_\diamond\geq \epsilon}\left(1+|\lambda|^2\right) \rho(\lambda)\mathrm{d}\lambda\leq e^{-\beta \epsilon^2\alpha^2},
\end{align}
where $\rho(\lambda):=\gamma(\lambda,\lambda)$ is the density function corresponding to the state $\gamma$. Furthermore for all $\zeta\in \mathbb{R}^{N}$ and $\beta<\frac{1}{|\zeta|^2}$, there exists a constant $\alpha(\beta,|\xi|)$ such that 
\begin{align}
\label{Equation-Large Deviations t}
\int_{|\braket{\zeta|\lambda}|\geq \epsilon}\left(1+|\lambda|^2\right) \rho(\lambda)\mathrm{d}\lambda\leq e^{-\beta \epsilon^2\alpha^2}
\end{align}
for all $\alpha\geq \alpha(\beta,|\xi|)$ and $\epsilon\geq D\alpha^{-s}$.
\end{prop}

The restriction to the finite dimensional space $\Pi L^2\! \left(\mathbb{R}^3\right)$ will be essential in the proof of Proposition \ref{Proposition-Large Deviations}, to be precise we will make use of the fact that $N\lesssim \alpha^p$ for a suitable $p>0$, which in particular implies that $N\lesssim e^{\alpha^t}$, uniformly in $\alpha$, for any $t>0$. 
Before we prove Proposition \ref{Proposition-Large Deviations}, we first need  auxiliary results concerning the $|\cdot|_\diamond$ norm.

\begin{defi}
\label{Definition-A_n}
For $x\in \mathbb{R}^3$ and $r>0$, let us define $T_x\lambda:=-2\chi\left(|\cdot -x|\leq 1\right)\left(-\Delta\right)^{-\frac{1}{2}}\varphi$ and $T_{\geq r} \lambda:=-2\chi\left(|\cdot |\geq r\right)\left(-\Delta\right)^{-\frac{1}{2}}\varphi$ with the above identification $\varphi=\sum_{n=1}^{{N}} \lambda_n \varphi_n$. Furthermore let us define the operators $A_x:=\sqrt{T_x^\dagger T_x}$ and $A_{\geq r}:=\sqrt{T_{\geq r}^\dagger T_{\geq r}}$, as well as the constant $\beta_0:=\inf_{x\in \mathbb{R}^3}\|A_x\|^{-2}$.
\end{defi}
Using the operators $A_x$ we can write $|\lambda|_{\diamond}=\sup_{x\in \mathbb{R}^3}|A_x \lambda|$, which is bounded by
\begin{align}
\label{Equation-Diamond estimate}
|\lambda|_{\diamond}\leq 65\max\bigg\{\sup_{z\in \mathbb{Z}^3:|z|\leq r+1}|A_z \lambda|,\, |A_{\geq r} \lambda|\bigg\}
\end{align}
for any $r>0$. In order to see this, note that for any $y\in \mathbb{R}^3$ there exists a $z\in \mathbb{Z}^3$ with $|y-z|< 1$. In case $y\in B_r(0)\cap B_1(x)$, where $x\in \mathbb{R}^3$, we  see that $z$ satisfies $|z|\leq r+1$ and $|x-z|<2$. Denoting the set of such $z$ by $M(x,r)\subset \mathbb{Z}^3$, we obtain $B_1(x)\subset \bigcup_{z\in M(x,r)}B_1(z)\cup \left(\mathbb{R}^3\setminus B_r(0)\right)$. Consequently
\begin{align*}
|\lambda|_{\diamond}\leq \sup_x \sum_{z\in M(x,r)}|A_z \lambda|+|A_{\geq r} \lambda|\leq \sup_x\Big( |M(x,r)|+1\Big)\max_x\Big\{\sup_{z\in M(x,r)}|A_z \lambda|,\, |A_{\geq r} \lambda|\Big\}.
\end{align*}
This concludes the proof of Eq.~(\ref{Equation-Diamond estimate}), since there are at most $64$ elements $z\in \mathbb{Z}^3$ satisfying $|x-z|<2$.

\begin{lem}
\label{Lemma-Trace estimate}
The constant $\beta_0$ from Definition \ref{Definition-A_n} is positive, uniformly in $\alpha$, and $\|A_x\|_{\mathrm{HS}}\lesssim \Lambda$ uniformly in $x\in \mathbb{R}^3$, where $\Lambda$ is defined above Eq.~(\ref{Equation-Definition of strong projection}). Furthermore there exists a constant $v>0$ such that $\|A_{\geq r}\|_{\mathrm{HS}}\lesssim \frac{\alpha^v}{\sqrt{r}}$ for all $\alpha\geq 1$ and $r>0$.
\end{lem}
\begin{proof}
Note that the space $\Pi  L^2\! \left(\mathbb{R}^3\right)$ is contained in the spectral subspace $-\Delta\leq \Lambda^2$, hence $\Pi \leq \left(1+ \Lambda^2\right)\left(1-\Delta\right)^{-1}$, and therefore
\begin{align*}
\|A_x\|^2_{\mathrm{HS}}&\! =\! 4\! \left\|\chi\left(|\cdot \! -x|\leq 1\right)\left(\! -\! \Delta\right)^{-\frac{1}{2}}\Pi \right\|^2_{\mathrm{HS}}\! \leq\!  4\! \left(\! 1\! +\!  \Lambda^2\right)\! \left\|\chi\left(|\cdot \!  -x|\leq 1\right)\left(-\Delta\right)^{-\frac{1}{2}}\left(1-\Delta\right)^{-\frac{1}{2}}\right\|^2_{\mathrm{HS}}\\
&=4\left(1+ \Lambda^2\right)\left\|\chi\left(|\cdot|\leq 1\right)\left(-\Delta\right)^{-\frac{1}{2}}\left(1-\Delta\right)^{-\frac{1}{2}}\right\|^2_{\mathrm{HS}}.
\end{align*}
Applying Eq.~(\ref{Equation-Second HS norm}) with $\psi=\chi\left(|.|\leq 1\right)$ yields that $\chi\left(|\cdot|\leq 1\right)\left(-\Delta\right)^{-\frac{1}{2}}\left(1-\Delta\right)^{-\frac{1}{2}}$ is Hilbert-Schmidt, hence $\|A_x\|_{\mathrm{HS}}\lesssim \Lambda$. In order to prove the uniform lower bound $\beta_0>0$, it is enough to verify the boundedness of $\chi\left(|.|\leq 1\right)f(-\Delta)$, where $f(t):=\frac{\chi(|t|\leq 1)}{\sqrt{t}}$. An explicit computation in Fourier space yields for $\varphi\in L^2 \!  \left(\mathbb{R}^3\right)$ 
\begin{align*}
&\braket{\varphi|f(-\Delta)\chi\left(|.|\leq 1\right)f(-\Delta)|\varphi}=\int_{|k|\leq 1}\int_{|k|\leq 1}\frac{\widehat{\chi\left(|\cdot|\leq 1\right)}(k-k')\overline{\widehat{\varphi}(k)}\widehat{\varphi}(k')}{|k|\, |k'|}\, \mathrm{dk} \mathrm{dk}\\
&\ \ \ \ \ \ \leq \left\|\widehat{\chi\left(|\cdot|\leq 1\right)}\right\|_\infty\left|\int_{|k|\leq 1}\frac{|\widehat{\varphi}(k)|}{|k|}\, \mathrm{dk}\right|^2\lesssim \|\varphi\|^2.
\end{align*}
Finally we are going to verify $\|A_{\geq r}\|_{\mathrm{HS}}\lesssim \frac{\alpha^v}{\sqrt{r}}$, using that 
\begin{align*}
\|A_{\geq r}\|_{\mathrm{HS}}= 2 \sqrt{\sum_{n=1}^{{N}} \left\|\chi\left(|\cdot |\geq r\right)\left(-\Delta\right)^{-\frac{1}{2}}\varphi_n\right\|^2}\lesssim \sqrt{{N}}\frac{\alpha^v}{\sqrt{r}}
\end{align*}
for a suitable constant $v>0$ by Corollary \ref{Corollary-BasisElements}, where $N$ is the dimension of $\Pi  L^2\! \left(\mathbb{R}^3\right)$. This concludes the proof, since ${N}\lesssim \alpha^p$ for some $p>0$.
\end{proof}

\begin{proof}[Proof of Proposition \ref{Proposition-Large Deviations}]
Making use of Eq.~(\ref{Equation-Diamond estimate}) and defining $\epsilon_*:=\frac{\epsilon}{65}$, we obtain
\begin{align*}
\int_{|\lambda|_\diamond\geq \epsilon}\! \! \left(1\! +\! |\lambda|^2\right)\!  \rho(\lambda)\mathrm{d}\lambda\!  \leq \! \!\! \!  \sum_{|z|\leq r+1}\! \int_{|A_z\lambda|\geq \epsilon_*}\! \! \! \left(1\! +\! |\lambda|^2\right)\!  \rho(\lambda)\mathrm{d}\lambda\,+\!\int_{|A_{\geq r}\lambda|\geq \epsilon_*}\! \! \! \! \left(1\! +\! |\lambda|^2\right) \! \rho(\lambda)\mathrm{d}\lambda,
\end{align*}
where the sum runs over $z\in \mathbb{Z}^3$ with $|z|\leq r+1$. In the following we are going to verify that every contribution of the form $\int_{|A_x\lambda|\geq \epsilon_*}\! \! \! \left(1\! +\! |\lambda|^2\right)\!  \rho(\lambda)\mathrm{d}\lambda$ is exponentially small uniformly in $x\in \mathbb{R}^3$. As a consequence of Eq.~(\ref{Equation-Strong Condensation}), we have for $t\geq 0$ the estimate
\begin{align*}
\gamma\leq \chi\left(\sum_{n=1}^{{N}} a_n^\dagger a_n\leq \alpha^{-h}\right)\leq e^{t\left(\alpha^{-h}-\sum_{n=1}^{{N}} a_n^\dagger a_n\right)}.
\end{align*}
By our assumption on $s$, there exists a $h'$ such that $2s<h'<h$. Consequently we obtain for $t:=\alpha^{2+(h-h')}$, using Mehler's kernel,
\begin{align}
\label{Equation-Estimate of the density}
\rho(\lambda)\! =\! \gamma(\lambda,\lambda)\! \leq \! e^{\alpha^{2-h'}}e^{-t\sum_{n=1}^{{N}} a_n^\dagger a_n}(\lambda,\lambda)\! = \! e^{\alpha^{2-h'}}\! \left(\frac{1}{1-e^{-\alpha^{h-h'}}}\right)^{N}\! \! \left(\frac{ \alpha^2 w_\alpha}{\pi}\right)^{\frac{{N}}{2}}\! \! \! \!   e^{-\alpha^2 w_\alpha |\lambda|^2},
\end{align}
with $w_\alpha:=\mathrm{coth}\left(\alpha^{h-h'}\right)-\mathrm{cosech}\left(\alpha^{h-h'}\right)$. Since $N e^{-\alpha^{h-h'}}\underset{\alpha\rightarrow \infty}\longrightarrow 0$, it is clear that $\left(\frac{1}{1-e^{-\alpha^{h-h'}}}\right)^{N}$ is bounded uniformly in $\alpha$. 
Since $w_\alpha\geq 0$ is strictly increasing in $\alpha$, we can choose $0<\beta'<\beta_0\inf_{\alpha\geq 1}w_\alpha$, where $\beta_0$ is the constant from Definition \ref{Definition-A_n}. Consequently $\|\frac{\beta'}{w_\alpha}|A_x|^2\|<1$ uniformly in $x\in \mathbb{R}^3$ and $\alpha\geq 1$,   and in particular  $\left(1-\frac{\beta'}{w_\alpha}|A_x|^2\right)^{-1}$ is a bounded operator. 
Hence we obtain for $x\in \mathbb{R}^3$ 
\begin{align*}
&\int_{|A_x\lambda|\geq \epsilon_*}\left(1+|\lambda|^2\right) \rho(\lambda)\mathrm{d}\lambda\lesssim e^{\alpha^{2-h'}}\left(\frac{ \alpha^2 w_\alpha}{\pi}\right)^{\frac{{N}}{2}}\int_{|A_x\lambda|\geq \epsilon_*}\left(1+|\lambda|^2\right)e^{-\alpha^2 w_\alpha |\lambda|^2}\mathrm{d}\lambda\\
&\ \ \ \leq e^{\alpha^{2-h'}}\left(\frac{ \alpha^2 w_\alpha}{\pi}\right)^{\frac{{N}}{2}}\int_{\mathbb{R}^{N}} \left(1+|\lambda|^2\right)e^{-\alpha^2 \left(w_\alpha|\lambda|^2+\beta'\epsilon_*^2 -\beta' |A_x \lambda|^2\right)}\mathrm{d}\lambda\\
&\ \ \  =e^{\alpha^{2-h'}}\frac{w_\alpha+\alpha^{-2}\mathrm{Tr}\left(1-\frac{\beta'}{w_\alpha}|A_x|^2\right)^{-1}}{w_\alpha\det \sqrt{1-\frac{\beta'}{w_\alpha}|A_x|^2}}\ e^{-\beta'\epsilon_*^2 \alpha^2}.
\end{align*}
Furthermore, for a suitable, $x$-independent, constant $\mu$
\begin{align}
\nonumber &e^{\alpha^{2-h'}}\frac{w_\alpha\! +\! \alpha^{-2}\mathrm{Tr}\left(1\! -\! \frac{\beta'}{w_\alpha}|A_x|^2\right)^{-1}}{w_\alpha\det \sqrt{1\! -\! \frac{\beta'}{w_\alpha}|A_x|^2}}\! \lesssim \! e^{\alpha^{2-h'}}\! \! \frac{\alpha^p}{\det \! \sqrt{1\! -\! \frac{\beta'}{w_\alpha}|A_x|^2}}\! \\
\label{Equation-Low Exponent}
& = e^{\alpha^{2-h'}+p\ln {\alpha}-\frac{1}{2}\mathrm{Tr}\ln\left(1-\frac{\beta'}{w_\alpha}|A_x|^2\right)}\leq e^{\alpha^{2-h'}+p\ln {\alpha}+\mu\|A_x\|_{\mathrm{HS}}^2}\leq e^{\alpha^{2-h'}+p\ln \alpha+\mu C\Lambda^2},
\end{align}
where we have used the rough estimate $w_\alpha + \alpha^{-2}\mathrm{Tr}\left(1 - \frac{\beta'}{w_\alpha}|A_x|^2\right)^{-1}\lesssim 1+\alpha^{-2}N\lesssim \alpha^p$ for a suitable exponent $p>0$ in the first inequality and Lemma \ref{Lemma-Trace estimate} in the last inequality. Note that the exponent in Eq.~(\ref{Equation-Low Exponent}) is of order $\alpha^{\mathrm{max}\big\{2-h',\frac{8}{5}(1+\sigma)\big\}}\ll \epsilon_*^2 \alpha^2$ since $\Lambda^2=\alpha^{\frac{8}{5}(1+\sigma)}$ and $\epsilon\geq D\alpha^{-s}$ with $s<\mathrm{min}\big\{\frac{h'}{2},\frac{1}{5}(1-4\sigma)\big\}$.

 Defining $r:=\alpha^{2q}$ with $q>v$, where $v$ is the constant from Lemma \ref{Lemma-Trace estimate} and making use of the fact that the number of $z\in \mathbb{Z}^3$ with $|z|\leq r+1$ is of order $r^3=\alpha^{6q}$, we obtain
\begin{align*}
\sum_{|z|\leq r+1}\int_{|A_{z}\lambda|\geq \epsilon_*}\left(1+|\lambda|^2\right) \rho(\lambda)\mathrm{d}\lambda\, \mathrm{d}x\lesssim \alpha^{6q}e^{\alpha^{2-h'}+p\ln \alpha+\mu C\Lambda^2-\beta'\epsilon_*^2\alpha^2}\leq e^{-\beta\epsilon_*^2\alpha^2}
\end{align*}
for $\beta<\beta'$ and $\alpha$ large enough. We have $\|A_{\geq r}\|_\mathrm{HS} \underset{\alpha \rightarrow \infty}{\longrightarrow}0$ by Lemma \ref{Lemma-Trace estimate} and our choice $r=\alpha^{2q}$ with $q>v$. Using Eq.~(\ref{Equation-Estimate of the density}), and an argument similar to the one in Eq.~(\ref{Equation-Low Exponent}), we can therefore estimate $\int_{|A_{\geq r}\lambda|\geq \epsilon_*}\left(1+|\lambda|^2\right) \rho(\lambda)\mathrm{d}\lambda$ by
\begin{align*}
\int_{|A_{\geq r}\lambda|\geq \epsilon_*}&\left(1+|\lambda|^2\right) \rho(\lambda)\mathrm{d}\lambda\lesssim e^{\alpha^{2-h'}}\left(\frac{ \alpha^2 w_\alpha}{\pi}\right)^{\frac{{N}}{2}}\int_{|A_{\geq r}\lambda|\geq \epsilon_*}\left(1+|\lambda|^2\right)e^{-\alpha^2 w_\alpha |\lambda|^2}\mathrm{d}\lambda\\
&\lesssim e^{\alpha^{2-h'}}\frac{\alpha^p}{\det \sqrt{1-\frac{\beta'}{w_\alpha}|A_{\geq r}|^2}}\ e^{-\beta'\epsilon_*^2 \alpha^2}\lesssim e^{\alpha^{2-h'}+p\ln \alpha+\mu \|A_{\geq r}\|_{\mathrm{HS}}^2-\beta'\epsilon_*^2 \alpha^2}.
\end{align*}
Again we observe that the exponent $\alpha^{2-h'}+p\ln \alpha+\mu \|A_{\geq r}\|_{\mathrm{HS}}^2$ is small compared to $\epsilon_*^2 \alpha^2$, which concludes the proof of Eq.~(\ref{Equation-Large Deviations}). 

The proof of Eq.~(\ref{Equation-Large Deviations t}) can be carried out analogously with the help of the operator $A_\zeta\lambda:=\braket{\zeta|\lambda}\frac{\zeta}{\|\zeta\|}$ using the fact that $\|A_\zeta\|_{\mathrm{HS}}=\|A_\zeta\|=|\zeta|$ and the assumption $\beta<\frac{1}{|\xi|^2}$. More precisely we obtain for $\beta<\beta'<\frac{1}{|\xi|^2}$ 
\begin{align*}
\int_{|\braket{\zeta|\lambda}|\geq \epsilon}&\left(1+|\lambda|^2\right) \rho(\lambda)\mathrm{d}\lambda\lesssim e^{\alpha^{2-h'}}\left(\frac{ \alpha^2 w_\alpha}{\pi}\right)^{\frac{{N}}{2}}\int_{|A_\zeta\lambda|\geq \epsilon}\left(1+|\lambda|^2\right)e^{-\alpha^2 w_\alpha |\lambda|^2}\mathrm{d}\lambda\\
&\lesssim e^{\alpha^{2-h'}}\frac{\alpha^p}{\det \sqrt{1-\frac{\beta'}{w_\alpha}|A_\zeta|^2}}\ e^{-\beta'\epsilon^2 \alpha^2}\lesssim e^{\alpha^{2-h'}+p\ln \alpha+\mu \|A_\zeta\|_{\mathrm{HS}}^2-\beta'\epsilon^2 \alpha^2}\leq e^{-\beta \epsilon^2\alpha^2}.
\end{align*}
\end{proof}

\section{Properties of the Pekar Functional}
\label{Section-Properties of the Pekar Functional}
In this section we are going to discuss essential properties of the Pekar functional $\mathcal{F}^\mathrm{Pek}$, and we are going to verify an asymptomatically sharp quadratic approximation for $\mathcal{F}^\mathrm{Pek}\! \left(\varphi\right)$, which is valid for all field configurations $\varphi$ close to a minimizer $\varphi^\mathrm{Pek}$. It has been proven in \cite{FS} that a suitable quadratic approximation of $\mathcal{F}^\mathrm{Pek}$ holds for all configurations $\varphi$ satisfying $\|V_{\varphi-\varphi^\mathrm{Pek}}\|\ll 1$, where
\begin{align}
\label{Equation-Potential}
V_\varphi:=-2\left(-\Delta\right)^{-\frac{1}{2}}\mathfrak{Re}\, \varphi.
\end{align}
In the following we are showing that this result is still valid, in case we substitute the $L^2$-norm with the weaker $\|\cdot \|_\diamond$ norm, which is a hybrid between the $L^2$ and the $L^\infty$ norm defined as
\begin{align}
\label{Equation-L2 diamond}
\|V \|_\diamond:=\sup_{x\in \mathbb{R}^3}\sqrt{\int_{B_{1}(x)}|V(y)|^2\ \mathrm{d}y},
\end{align}
where $B_{1}(x)$ is the unit ball centered at $x\in \mathbb{R}^3$. This will be the content of Lemma \ref{Lemma-Basic Estimate} and Theorem \ref{Theorem-Lower Bound on the Pekar Functional}, respectively. We  have $\|V_\varphi\|_\diamond=|\lambda|_\diamond$ for $\varphi=\sum_{n=1}^N \lambda_n \varphi_n$, where $|\cdot |_\diamond$ is the norm defined in Eq.~(\ref{Equation-Definition_Diamond_Norm_in_Coordinates}). Before we come to the proof of Lemma \ref{Lemma-Basic Estimate}, we first need the subsequent auxiliary Lemma \ref{Lemma-Relative Operator Bound}.

\begin{lem}
\label{Lemma-Relative Operator Bound}
There exists a constant $C>0$ such that the operator inequality 
\begin{align}
\label{Equation-Relative Operator Bound}
V^2\leq C\|V\|_\diamond^2 \left(1-\Delta\right)^2
\end{align}
holds for all (measurable) $V:\mathbb{R}^3\longrightarrow \mathbb{R}$, where $V^2$ is interpreted as a multiplication operator.
\end{lem}
\begin{proof}
As a first step, we are going to verify that Eq.~(\ref{Equation-Relative Operator Bound}) holds in case we use the $L^2$ norm $\|V\|$ instead of $\|V\|_\diamond$. This follows from $V^2\leq \left\|V(1-\Delta)^{-1}\right\|_{\mathrm{HS}}^2 \left(1-\Delta\right)^2$, where $\|\cdot \|_{\mathrm{HS}}$ is the Hilbert-Schmidt norm, and
\begin{align*}
\left\|V(1-\Delta)^{-1}\right\|_{\mathrm{HS}}^2=\int \int V(x)^2 K(y-x)^2\, \mathrm{d}x\mathrm{d}y=\int K(y)^2\mathrm{d}y\, \|V\|^2
\end{align*}
with $K(y-x)$ being the kernel of the operator $(1-\Delta)^{-1}$. Note that $C':=\int K(y)^2\mathrm{d}y$ is finite, which concludes the first step. In order to obtain the analogue statement for $\|V\|_\diamond$, let $\chi$ be a smooth, non-negative, function with $\mathrm{supp} \left(\chi\right)\subset B_1(0)$ and $\int_{\mathbb{R}^3} \chi(y)^2\, \mathrm{d}y=1$. Defining $\chi_y(x):=\chi(x-y)$ for $y\in \mathbb{R}^3$ and using the previously derived inequality $V^2\leq C'\|V\|^2 \left(1-\Delta\right)^2$,
which holds for any $V\in L^2\! \left(\mathbb{R}^3\right)$, we obtain
\begin{align*}
V^2&=\! \int\! \chi_y V^2\chi_y\, \mathrm{d}y=\! \int\! \chi_y \left(\mathds{1}_{B_1(y)}V\right)^2\chi_y\, \mathrm{d}y\leq  C'\!\int \! \|\mathds{1}_{B_1(y)}V\|^2\, \chi_y \left(1-\Delta\right)^2\chi_y\, \mathrm{d}y\\
&\leq C'\|V\|^2_\diamond \int\chi_y \left(1-\Delta\right)^2\chi_y\, \mathrm{d}y=C'\|V\|^2_\diamond \int\big|(1-\Delta)\chi_y\big|^2\, \mathrm{d}y,
\end{align*}
where $\left|A\right|^2=A^\dagger A$. Furthermore $(1-\Delta)\chi_y=\chi_y(1-\Delta)-2\left(\nabla \chi_y\right)\nabla-\left(\Delta \chi_y\right)$, which yields together with a Cauchy--Schwarz inequality the estimate
\begin{align*}
&\int\big|(1-\Delta)\chi_y\big|^2 \mathrm{d}y\leq 3\int \Big((1-\Delta)\chi_y^2(1-\Delta)-4\nabla \left|\nabla \chi_y\right|^2
\nabla+|\Delta \chi_y|^2\Big)\mathrm{d}y\\
&\ \ \ \ = 3(1-\Delta)^2-12\, \nabla \! \left(\int \left|\nabla \chi_y^2\right|
 \mathrm{d}y\right)\! \nabla +3\int |\Delta \chi_y|^2\mathrm{d}y\lesssim \left(1-\Delta\right)^2,
\end{align*}
where we have used that $ \int \left|\nabla \chi(y)\right|^2\, \mathrm{d}y$ and $\int \left|\Delta\chi(y)\right|^2\, \mathrm{d}y$ are finite.
\end{proof}

In the following we are going to use that we can write the Pekar energy as 
\begin{equation}\label{def:FP}
\mathcal{F}^\mathrm{Pek}\! \left(\varphi\right)=\|\varphi\|^2+\inf \sigma\left(-\Delta+V_\varphi\right),
\end{equation}
where $V_\varphi$ is defined in Eq.~(\ref{Equation-Potential}). As an immediate consequence of Eq.~(\ref{Equation-Relative Operator Bound}) we have $\pm V\leq \sqrt{C}\|V\|_\diamond \left(1-\Delta\right)$ and consequently there exists a $\delta_0>0$ and a contour $\mathcal{C}\subset \mathbb{C}$, such that $\mathcal{C}$ separates the ground state energy $\inf \sigma\left(-\Delta+V\right)$ from the excitation spectrum of $H_V:=-\Delta+V$ for all $V$ with $\|V-V_{\varphi^\mathrm{Pek}}\|_\diamond<\delta_0$, see also \cite{FS}. This allows us to further identify $\mathcal{F}^\mathrm{Pek}\! \left(\varphi\right)$ as
\begin{align}
\label{Equation-Trace formula}
\mathcal{F}^\mathrm{Pek}\! \left(\varphi\right)=\|\varphi\|^2+\mathrm{Tr}\int_\mathcal{C} \frac{z}{z-H_{V_\varphi}}\frac{\mathrm{d}z}{2\pi i}
\end{align}
for all $\varphi$ satisfying $\|V_{\varphi-\varphi^\mathrm{Pek}}\|_\diamond<\delta_0$. Following the strategy in \cite{FS}, we will use Eq.~(\ref{Equation-Trace formula}) to compare $\mathcal{F}^\mathrm{Pek}\! \left(\varphi\right)$ with $e^\mathrm{Pek}=\mathcal{F}^\mathrm{Pek}\! \left(\varphi^\mathrm{Pek}\right)$. Before we do this let us introduce the operators
\begin{align}
\label{Equation-K}
K^\mathrm{Pek}&:=1-H^\mathrm{Pek}=4\left(-\Delta\right)^{-\frac{1}{2}}\psi^\mathrm{Pek}\frac{1-\ket{\psi^\mathrm{Pek}}\bra{\psi^\mathrm{Pek}}}{H_{V^\mathrm{Pek}}-\mu^\mathrm{Pek}}\psi^\mathrm{Pek}\left(-\Delta\right)^{-\frac{1}{2}},\\
\label{Equation-L}
L^\mathrm{Pek}&:=4\left(-\Delta\right)^{-\frac{1}{2}}\psi^\mathrm{Pek}\left(1-\Delta\right)^{-1}\psi^\mathrm{Pek}\left(-\Delta\right)^{-\frac{1}{2}},
\end{align}
where $H^\mathrm{Pek}$ is defined in Eq.~(\ref{Equation-Definition_Hessian}), $\mu^\mathrm{Pek}:=e^\mathrm{Pek}-\|\varphi^\mathrm{Pek}\|^2$ and $\psi^\mathrm{Pek}$ is the, non-negative, ground state of the operator $H_{V^\mathrm{Pek}}$ with $V^\mathrm{Pek}:=V_{\varphi^\mathrm{Pek}}$, which we interpret as a multiplication operator in Eqs.~(\ref{Equation-K}) and ~(\ref{Equation-L}). The following Lemma \ref{Lemma-Basic Estimate} can be proved in the same way as \cite[Proposition 3.3]{FS}, using Lemma~\ref{Lemma-Relative Operator Bound}.

\begin{lem}
\label{Lemma-Basic Estimate}
There exist  constants $c,\delta_0>0$ such that for all $\varphi$ with $\|V_{\varphi-\varphi^\mathrm{Pek}}\|_\diamond<\delta_0$
\begin{align}
\label{Equation-Basic Estimate}
\left|\mathcal{F}^\mathrm{Pek}\! \left(\varphi\right)\! -\! e^\mathrm{Pek}-\braket{\varphi\! -\! \varphi^\mathrm{Pek}|1\! -\! K^\mathrm{Pek}|\varphi-\varphi^\mathrm{Pek}}\right|\! \leq\!  c\|V_{\varphi-\varphi^\mathrm{Pek}}\|_\diamond\braket{\varphi\! -\! \varphi^\mathrm{Pek}|L^\mathrm{Pek}|\varphi\! -\! \varphi^\mathrm{Pek}}.
\end{align}
\end{lem}
\begin{proof}
By taking $\delta_0$ small enough, we can assume for all $V$ with $\|V_{\varphi-\varphi^\mathrm{Pek}}\|_\diamond<\delta_0$  that
\begin{align}
\label{Equation-Well defined inverse}
\sup_{z\in \mathcal{C}}\left\|V_{\varphi-\varphi^\mathrm{Pek}}\frac{1}{z-H_{V^\mathrm{Pek}}}\right\|_{\mathrm{op}}<1,
\end{align}
where $\|\cdot \|_\mathrm{op}$ denotes the operator norm. This immediately follows from 
\begin{align*}
\left\|V_{\varphi-\varphi^\mathrm{Pek}}\frac{1}{H_{V^\mathrm{Pek}}-z}\right\|_\mathrm{op}^2\lesssim \left\|\left(V_{\varphi}-V^\mathrm{Pek}\right)\left(1-\Delta\right)^{-1}\right\|_\mathrm{op}^2\leq C\|V_{\varphi-\varphi^\mathrm{Pek}}\|^2_\diamond,
\end{align*}
where we used Eq.~(\ref{Equation-Relative Operator Bound}) and the fact that the spectrum of $H_{V^\mathrm{Pek}}$ has a positive distance to the contour $\mathcal{C}$, allowing us to bound the operator norm  of $\left(1-\Delta\right)\frac{1}{H_{V^\mathrm{Pek}}-z}$ uniformly in $z\in \mathcal{C}$. Given Eq.~(\ref{Equation-Well defined inverse}), it has been verified in the proof of \cite[Proposition 3.3]{FS} that
\begin{align*}
&\left|\|\varphi\|^2+\mathrm{Tr}\int_\mathcal{C} \frac{z}{z-H_{V_\varphi}}\frac{\mathrm{d}z}{2\pi i}-e^\mathrm{Pek}-\braket{\varphi-\varphi^\mathrm{Pek}|1-K^\mathrm{Pek}|\varphi-\varphi^\mathrm{Pek}}\right|\\
&\ \ \ \ \ \ \ \ \ \ \ \ \ \ \ \ \ \ \ \ \lesssim \epsilon\braket{\varphi-\varphi^\mathrm{Pek}|L^\mathrm{Pek}|\varphi-\varphi^\mathrm{Pek}}
\end{align*}
for $\epsilon:=\sup_{z\in \mathcal{C}}\Big\{\left\|\frac{A}{1-A}\right\|_\mathrm{op}+\left\|\frac{B}{1-B}\right\|_\mathrm{op}+\left\|\left(1-\Delta\right)^{\frac{1}{2}}\frac{1}{z-H_{V^\mathrm{Pek}}}\frac{A}{1-A}\left(1-\Delta\right)^{\frac{1}{2}}\right\|_\mathrm{op}\Big\}$, where we denote $A:=\left(V_{\varphi-\varphi^\mathrm{Pek}}\right)\frac{1}{z-H_{V^\mathrm{Pek}}}$ and $B:=\left(1-\ket{\psi^\mathrm{Pek}}\bra{\psi^\mathrm{Pek}}\right)A^\dagger$. In the following we want to verify that $\epsilon\lesssim \|V_{\varphi-\varphi^\mathrm{Pek}}\|_\diamond$, which concludes the proof by Eq.~(\ref{Equation-Trace formula}). Since $(1-\Delta)\frac{1}{z-H_{V^\mathrm{Pek}}}$ is uniformly bounded in $z$, $\left\|\frac{A}{1-A}\right\|_\mathrm{op}\leq  \frac{\|A\|_\mathrm{op}}{1-\|A\|_\mathrm{op}}\lesssim \|\left(V_{\varphi-\varphi^\mathrm{Pek}}\right)\left(1-\Delta\right)^{-1}\|_\mathrm{op}\lesssim \|V_{\varphi-\varphi^\mathrm{Pek}}\|_\diamond$ by Eq.~(\ref{Equation-Relative Operator Bound}). Similarly $\left\|\frac{B}{1-B}\right\|_\mathrm{op}\lesssim \|V_{\varphi-\varphi^\mathrm{Pek}}\|_\diamond$. Regarding the final term in the definition of $\epsilon$, note that $(1-\Delta)^{\frac{1}{2}}\frac{1}{z-H_{V^\mathrm{Pek}}}(1-\Delta)^{\frac{1}{2}}$ is uniformly bounded in $z$, and therefore
\begin{align*}
&\left\|\left(1\! -\! \Delta\right)^{\frac{1}{2}}\frac{1}{z\! -\! H_{V^\mathrm{Pek}}}\frac{A}{1\! -\! A}\left(1\! -\! \Delta\right)^{\frac{1}{2}}\right\|_\mathrm{op}\! \lesssim \! \left\|\left(1\! -\! \Delta\right)^{-\frac{1}{2}}\frac{A}{1\! -\! A}\left(1\! -\! \Delta\right)^{\frac{1}{2}}\right\|_\mathrm{op}\! =\! \left\|\frac{A'}{1\! -\! A'}\right\|_\mathrm{op}\! ,
\end{align*}
with $A'\! :=\! \left(1-\Delta\right)^{-\frac{1}{2}}A\left(1-\Delta\right)^{\frac{1}{2}}$. Furthermore $\left\|\frac{A'}{1\! -\! A'}\right\|_\mathrm{op}\! \! \! \leq \! \frac{\|A'\|_\mathrm{op}}{1\! -\! \|A'\|_\mathrm{op}}$ and
\begin{align*}
\|A'\|_\mathrm{op}\lesssim \left\|\left(1-\Delta\right)^{-\frac{1}{2}}\left(V_{\varphi-\varphi^\mathrm{Pek}}\right)\left(1-\Delta\right)^{-\frac{1}{2}}\right\|_\mathrm{op}\leq \|\left(V_{\varphi-\varphi^\mathrm{Pek}}\right)\left(1-\Delta\right)^{-1}\|_\mathrm{op}\lesssim \|V_{\varphi-\varphi^\mathrm{Pek}}\|_\diamond.
\end{align*}
\end{proof}

Lemma \ref{Lemma-Basic Estimate} gives a lower bound on $\mathcal{F}^\mathrm{Pek}\! \left(\varphi^\mathrm{Pek}+\xi\right)-e^\mathrm{Pek}$ in terms of a quadratic function $\xi\mapsto \braket{\xi|1-\left(K^\mathrm{Pek}+\epsilon L^\mathrm{Pek}\right)|\xi}$ for $\xi$ satisfying $\|V_\xi\|_\diamond< \min\{\frac{\epsilon}{c},\delta_0\}$. Due to the translation invariance of $\mathcal{F}^\mathrm{Pek}$, this lower bound is however insufficient, since we have for all $\xi\in \mathrm{span}\{\partial_{y_1}\varphi^\mathrm{Pek},\partial_{y_2}\varphi^\mathrm{Pek},\partial_{y_3}\varphi^\mathrm{Pek}\}\setminus \{0\}$ 
\begin{align}
\label{Equation-NegativityOfQuadraticLowerBound}
\braket{\xi|1-\left(K^\mathrm{Pek}+\epsilon L^\mathrm{Pek}\right)|\xi}=\mathrm{Hess}|_{\varphi^\mathrm{Pek}}\mathcal{F}^\mathrm{Pek}[\xi]-\epsilon \braket{\xi|L^\mathrm{Pek}|\xi}=-\epsilon \braket{\xi|L^\mathrm{Pek}|\xi}<0,
\end{align}
i.e. the quadratic lower bound is not even non-negative. In order to improve this lower bound, we will introduce a suitable coordinate transformation $\tau$ in Definition \ref{Definition-tau}. Before we can formulate Definition \ref{Definition-tau} we need some auxiliary preparations.

In the following let $\Pi$ be the projection defined in Eq.~(\ref{Equation-Definition of strong projection}) and let us define the real orthonormal system
\begin{align}
\label{Equation-Orthonormal Basis}
\varphi_n:=\frac{\Pi\partial_{y_n}\varphi^\mathrm{Pek}}{\|\Pi\partial_{y_n}\varphi^\mathrm{Pek}\|}
\end{align}
for $n\in \{1,2,3\}$, which we complete to a real orthonormal basis $\{\varphi_1,..,\varphi_{N}\}$ of $\Pi L^2\! \left(\mathbb{R}^3\right)$. Furthermore let us write $\varphi^\mathrm{Pek}_{x}(y):=\varphi^\mathrm{Pek}(y-x)$ for the translations of $\varphi^\mathrm{Pek}$ and let us define the map $\omega:\mathbb{R}^3\longrightarrow \mathbb{R}^3$ as
\begin{align}
\label{Equation-lambda}
\omega\left(x\right):=\Big(\braket{\varphi_n|\varphi^\mathrm{Pek}_{x}}\Big)_{n=1}^3\in \mathbb{R}^3.
\end{align}
Since $\varphi^\mathrm{Pek} \in H^1\! \left(\mathbb{R}^3\right)$, $\omega$ is differentiable. Moreover, since $\varphi^\mathrm{Pek}$ is invariant under the action of $O\left(3\right)$ and since the operator $\Pi$ commutes with the reflections $y_i\rightarrow -y_i$ and permutations $y_i\leftrightarrow y_j$, it is clear that $\omega(0)=0$. By the same argument we see that $D|_0 \omega$ has full rank and therefore there exists a local inverse $t\mapsto x_t$ for $|t|< \delta_*$ and a suitable constant $\delta_*>0$.

\begin{defi}
\label{Definition-tau}
We define the map $\tau:\Pi L^2\! \left(\mathbb{R}^3\right)\longrightarrow \Pi L^2\! \left(\mathbb{R}^3\right)$ as
\begin{align*}
\tau\left(\varphi\right):=\varphi-f(t^\varphi),
\end{align*}
where $t^\varphi:=\big(\braket{\varphi_1|\varphi},\braket{\varphi_2|\varphi},\braket{\varphi_3|\varphi}\big)\in \mathbb{R}^3$ and $f(t)$ is defined as
\begin{align*}
f(t):=\chi\big(|t|< \delta_*\big)\left(\Pi \varphi^\mathrm{Pek}_{x_t}-\sum_{n=1}^3 t_n \varphi_n\right).
\end{align*}
\end{defi}
The map $\tau$ is constructed in a way such that it \lq\lq flattens\rq\rq\ the manifold of Pekar minimizers $\{\varphi^\mathrm{Pek}_x:x\in \mathbb{R}^3\}$. More precisely, we have that $\tau\left(\Pi\varphi^\mathrm{Pek}_x\right)$ is for all small enough $x\in \mathbb{R}^3$ an element of the linear space spanned by $\{\varphi_1,\varphi_2,\varphi_3\}$. 
A similar construction appears in \cite{BS} and, in a somewhat different way, in  \cite{FeS}.\\

Recall the operators $K^\mathrm{Pek}$ and $L^\mathrm{Pek}$ from Eqs.~(\ref{Equation-K}) and~(\ref{Equation-L}), and let $T_x$ be the translation operator defined by $(T_x \varphi)(y):=\varphi(y-x)$. Then we define the operators $K^\mathrm{Pek}_x:=T_x\, K^\mathrm{Pek}\, T_{-x}$ and $L^\mathrm{Pek}_x:=T_x\, L^\mathrm{Pek}\, T_{-x}$, as well as for $|t|<\epsilon$ with $\epsilon<\delta_*$
\begin{align}
\label{Equation-Definition of J}
J_{t,\epsilon}:=\pi \Big(1-(1+\epsilon)\left(K^\mathrm{Pek}_{x_t}+\epsilon L^\mathrm{Pek}_{x_t}\right)\Big)\pi,
\end{align}
where $\pi:L^2\! \left(\mathbb{R}^3\right)\longrightarrow L^2\! \left(\mathbb{R}^3\right)$ is the orthogonal projection onto the subspace spanned by $\{\varphi_4,\dots, \varphi_N\}$. Furthermore we define $J_{t,\epsilon}:=\pi$ for $|t|\geq \epsilon$. In contrast to the operator $1-\left(K^\mathrm{Pek}+\epsilon L^\mathrm{Pek}\right)$ from Eq.~(\ref{Equation-NegativityOfQuadraticLowerBound}), the operator $J_{t,\epsilon}$ is non-negative for $\epsilon$ small enough, as will be shown in  Lemma \ref{Lemma-Convergence of the trace}. With the operator $J_{t,\epsilon}$ and the transformation $\tau$ at hand we can formulate a strong lower bound for $\mathcal{F}^\mathrm{Pek}\! \left(\varphi\right)-e^\mathrm{Pek}$ in the subsequent Theorem \ref{Theorem-Lower Bound on the Pekar Functional}, where we use the shorthand notation $J_{t,\epsilon} \big[ \varphi \big]:=\langle \varphi|J_{t,\epsilon}|\varphi\rangle$. 

\begin{theorem}
\label{Theorem-Lower Bound on the Pekar Functional}
There exist constants $C>0$, $0<\epsilon_0\leq \delta_*$ and $0<D\leq 1$ such that
\begin{align}
\label{Equation-Pekar Functional}
\mathcal{F}^\mathrm{Pek}\! \left(\varphi\right)\geq e^\mathrm{Pek}+J_{t^\varphi,\epsilon} \big[ \tau\! \left(\varphi \right)\big]-\frac{C}{\epsilon}\left\|\left(1-\Pi\right)\varphi^\mathrm{Pek}_{x_{t^\varphi}}\right\|^2
\end{align}
for all $0<\epsilon< \epsilon_0$ and $\varphi\in \Pi  L^2\! \left(\mathbb{R}^3\right)$ satisfying $\left\|V_{\varphi-\varphi^\mathrm{Pek}}\right\|_\diamond<\epsilon D$ and $|t^\varphi|<\epsilon D$, where $J_{t,\epsilon}$ is defined in Eq.~(\ref{Equation-Definition of J}).
\end{theorem}
\begin{proof}
In the following we use the abbreviation $t:=t^\varphi$. Since $\big\|V_{\varphi^\mathrm{Pek}-\varphi^\mathrm{Pek}_{x}}\big\|_\diamond\lesssim |x|$ and $|x_t|\lesssim |t|$ for $|t|\leq \frac{\delta_*}{2}$, we have for all $\varphi$ satisfying $\left\|V_{\varphi-\varphi^\mathrm{Pek}}\right\|_\diamond<D\epsilon$ and $|t|<\min\{D\epsilon,\frac{\delta_*}{2}\}$
\begin{align*}
\big\|V_{T_{-x_t}\varphi-\varphi^\mathrm{Pek}}\big\|_\diamond= \big\|V_{\varphi-\varphi^\mathrm{Pek}_{x_t}}\big\|_\diamond\leq \big\|V_{\varphi-\varphi^\mathrm{Pek}}\big\|_\diamond+\big\|V_{\varphi^\mathrm{Pek}-\varphi^\mathrm{Pek}_{x_t}}\big\|_\diamond\lesssim  \big\|V_{\varphi-\varphi^\mathrm{Pek}}\big\|_\diamond+|t|\lesssim D\epsilon.
\end{align*}
By taking $D$ small enough we obtain $\big\|V_{T_{-x_t}\varphi-\varphi^\mathrm{Pek}}\big\|_\diamond\leq \frac{\epsilon}{c}$ where $c$ is the constant from Lemma \ref{Lemma-Basic Estimate}.  Let us define $\epsilon_0:=\min\left\{c \delta_0,\frac{\delta_*}{2D},\delta_*\right\}$. Using the translation-invariance of $\mathcal{F}^\mathrm{Pek}$ and applying Lemma \ref{Lemma-Basic Estimate} yields 
\begin{align}
\nonumber
\mathcal{F}^\mathrm{Pek}\! \left(\varphi\right)\! -\! e^\mathrm{Pek}\! &=\! \mathcal{F}^\mathrm{Pek}\! \left(T_{-x_t}\varphi\right)\! -\! e^\mathrm{Pek}\! \geq\!  \braket{T_{-x_t}\varphi\! -\! \varphi^\mathrm{Pek}|1\! -\! \left(\! K^\mathrm{Pek}\! +\! \epsilon L^\mathrm{Pek}\! \right)\! |T_{-x_t}\varphi\! -\! \varphi^\mathrm{Pek}}\\
\nonumber
& =\braket{\varphi-\varphi^\mathrm{Pek}_{x_t}|1 - \left(K^\mathrm{Pek}_{x_t}+\epsilon L^\mathrm{Pek}_{x_t}\right)|\varphi-\varphi^\mathrm{Pek}_{x_t}}\\
\nonumber
& \geq \|\varphi-\Pi\varphi^\mathrm{Pek}_{x_t}\|^2 - \braket{\varphi-\varphi^\mathrm{Pek}_{x_t}|K^\mathrm{Pek}_{x_t} + \epsilon L^\mathrm{Pek}_{x_t}|\varphi-\varphi^\mathrm{Pek}_{x_t}}\\
\nonumber
&  \geq \|\varphi-\Pi\varphi^\mathrm{Pek}_{x_t}\|^2-\left(1+\epsilon\right)\braket{\varphi-\Pi\varphi^\mathrm{Pek}_{x_t}|K^\mathrm{Pek}_{x_t}+\epsilon  L^\mathrm{Pek}_{x_t}|\varphi-\Pi\varphi^\mathrm{Pek}_{x_t}}\\
\label{Equation-Quadratic lower bound}
&\ \ \ \ \ \ -\left(1+\epsilon^{-1}\right)\braket{ \left(1-\Pi\right)\varphi^\mathrm{Pek}_{x_t}|K^\mathrm{Pek}_{x_t}+\epsilon  L^\mathrm{Pek}_{x_t}|\left(1-\Pi\right)\varphi^\mathrm{Pek}_{x_t}},
\end{align}
where we have used the positivity of $K^\mathrm{Pek}_x$ and $L^\mathrm{Pek}_x$, and the Cauchy--Schwarz inequality in the last estimate. Note that by construction of $x_t$ as the local inverse of the function $\omega$ from Eq.~(\ref{Equation-lambda}), we have $\braket{\varphi_n|\varphi-\Pi\varphi^\mathrm{Pek}_{x_t}}=0$ for $n\in \{1,2,3\}$ and therefore
\begin{align*}
\varphi-\Pi\varphi^\mathrm{Pek}_{x_t}=\pi\left(\varphi-\Pi\varphi^\mathrm{Pek}_{x_t}\right)=\pi\left(\varphi-f(t)\right)=\pi\left(\tau\left(\varphi\right)\right)
\end{align*}
with $\pi$ being defined below Eq.~(\ref{Equation-Definition of J}), where we used $|t|<\delta_*$. This concludes the proof with $C:=\left(1+\epsilon_0\right)\left(\|K\|_\mathrm{op}+\epsilon_0 \|L\|_\mathrm{op}\right)$.
\end{proof}

\section{Proof of Theorem \ref{Theorem: Main}}
\label{Section-Proof of the lower Bound}
In the following we will combine the results of the previous sections in order to prove the lower bound on the ground state energy $E_\alpha$ in Theorem \ref{Theorem: Main}. We start by verifying the subsequent Lemma \ref{Lemma-CombiningResults}, which provides a lower bound on $E_\alpha$ in terms of an operator that is, up to a coordinate transformation $\tau$ and a non-negative term, a harmonic oscillator. 

 Let us again use the identification $\mathcal{F}\left(\Pi  L^2\! \left(\mathbb{R}^3\right)\right)\cong L^2\! \left(\mathbb{R}^{N}\right)$ utilizing the representation of real functions $\varphi=\sum_{n=1}^{N} \lambda_n \varphi_n\in \Pi  L^2\! \left(\mathbb{R}^3\right)$ by points $\lambda=(\lambda_1,\dots ,\lambda_{{N}})\in \mathbb{R}^{N}$, such that the annihilation operators $a_n:=a\left(\varphi_n\right)$ are given by $a_n=\lambda_n+\frac{1}{2\alpha^2}\partial_{\lambda_n}$, where $\lambda_n$ is the multiplication operator by the function $\lambda\mapsto \lambda_n$ on $L^2\! \left(\mathbb{R}^{N}\right)$, see also Eq.~(\ref{Equation-CreationAndAnnihilation}), where $\Pi $ is the projection from Eq.~(\ref{Equation-Definition of strong projection}) and $\{\varphi_1,\dots ,\varphi_{N}\}$ is the orthonormal basis of $\Pi  L^2\! \left(\mathbb{R}^3\right)$ constructed around Eq.~(\ref{Equation-Orthonormal Basis}). Let us also use for functions $\varphi\mapsto g(\varphi)$ depending on elements $\varphi\in \Pi  L^2\! \left(\mathbb{R}^3\right)$ the convenient notation $g(\lambda):=g\left(\sum_{n=1}^{N} \lambda_n \varphi_n\right)$, where $\lambda\in \mathbb{R}^N$.
 
\begin{lem}
\label{Lemma-CombiningResults}
Let $C>0$ and $0<\sigma\leq \frac{1}{4}$, and assume $s,h$ and $\sigma$ satisfy $2s<h$ and $\sigma<\frac{1-5s}{4}$. Furthermore let us define $\Lambda:=\alpha^{\frac{4}{5}(1+\sigma)}$ and $L:=\alpha^{1+\sigma}$. Then we obtain for any state $\Psi$ satisfying $\braket{\Psi|\mathbb{H}_\Lambda|\Psi}\leq C$, $\mathrm{supp}\left(\Psi\right)\subset B_{4L}(0)$ and 
\begin{align}
\label{Equation-AssumptionCondensation}
\chi\left(W^{-1}_{\varphi^\mathrm{Pek}}\mathcal{N}W_{\varphi^\mathrm{Pek}}\leq \alpha^{-h}\right)\Psi=\Psi,
\end{align}
that
\begin{align}
\nonumber
\braket{\Psi|\mathbb{H}_\Lambda|\Psi}\geq e^\mathrm{Pek}&+\Big\langle\Psi\Big| -\frac{1}{4\alpha^4}\sum_{n=1}^{N}\partial_{\lambda_n}^2+J_{t^\lambda,\alpha^{-s}}\! \big[ \tau\! \left(\lambda \right)\big]+\mathcal{N}-\sum_{n=1}^N a_n^\dagger a_n\Big|\Psi\Big\rangle \! -\frac{{N}}{2\alpha^2}\\
\label{Equation-QuasiQuadraticLowerBOund}
&\ \ \ \ \ \ +O\left(\alpha^{s-\frac{12}{5}(1+\sigma)}+\alpha^{-2(1+\sigma)}\right),
\end{align}
where $t^\varphi$ and $\tau(\varphi)$ are defined in Lemma \ref{Definition-tau} and $J_{t,\epsilon}$ is defined in Eq.~(\ref{Equation-Definition of J}). Furthermore, there exists a $\beta>0$, such that $\braket{\Psi|1-\mathbb{B}|\Psi}\leq e^{-\beta\alpha^{2(1-s)}}$, where $\mathbb{B}$ is the multiplication operator by the function $\lambda\mapsto \chi\left(|t^\lambda|<\alpha^{-s}\right)$.
\end{lem}
\begin{proof}
%
 Applying Eq.~(\ref{Equation-General estimate for cut-off's}) with $\Lambda$ and $\ell$ as in the definition of $\Pi $, see Eq.~(\ref{Equation-Definition of strong projection}), and $K:=\Lambda$, and utilizing Eq.~(\ref{Equation-Boundedness from below}), we obtain for a suitable $C'$
\begin{align}
\label{Equation-Introduction of Discretization}
\braket{\Psi|\mathbb{H}_\Lambda|\Psi}\geq \braket{\Psi|\mathbb{H}^0_{\Lambda,\ell}|\Psi}-C' \alpha^{-2(1+\sigma)}.
\end{align}
Making use of $\sum_{n=1}^{N} a_n^\dagger a_n=\sum_{n=1}^{N}\left(-\frac{1}{4\alpha^4}\partial_{\lambda_n}^2+\lambda_n^2\right)-\frac{{N}}{2\alpha^2}$ and $a_n+a_n^\dagger=2\lambda_n$, we further have the identity
\begin{align*}
\mathbb{H}^0_{\Lambda,\ell}& =-\Delta_x-2\sum_{n=1}^N \braket{\varphi_n|w_x}\lambda_n+\sum_{n=1}^{N}\left(-\frac{1}{4\alpha^4}\partial_{\lambda_n}^2+\lambda_n^2\right)-\frac{{N}}{2\alpha^2}+\mathcal{N}-\sum_{n=1}^N a_n^\dagger a_n\\
& = -\Delta_x+V_\lambda(x)+\sum_{n=1}^{N}\left(-\frac{1}{4\alpha^4}\partial_{\lambda_n}^2+\lambda_n^2\right)-\frac{{N}}{2\alpha^2}+\mathcal{N}-\sum_{n=1}^N a_n^\dagger a_n,
\end{align*}
with $V_\varphi$ defined in Eq.~(\ref{Equation-Potential}). 
Clearly $-\Delta_x+V_\lambda\geq \inf \sigma\left(-\Delta_x+V_\lambda\right)=\mathcal{F}^\mathrm{Pek}\! \left(\lambda\right)-\sum_{n=1}^{N} \lambda_n^2$, which yields the inequality $\mathbb{H}^0_{\Lambda,\ell}\geq \mathbb{K}+\mathcal{N}-\sum_{n=1}^N a_n^\dagger a_n$ with
\begin{align}
\label{Equation-Hamiltonian without electron part}
\mathbb{K}:=-\frac{1}{4\alpha^4}\sum_{n=1}^{N}\partial_{\lambda_n}^2+\mathcal{F}^\mathrm{Pek}\! \left(\lambda\right)-\frac{{N}}{2\alpha^2}.
\end{align}
Combining Eqs.~(\ref{Equation-Introduction of Discretization}) and~(\ref{Equation-Hamiltonian without electron part}), we obtain
\begin{align}
\label{Equation-K estimate}
\Big\langle \Psi\Big|\mathbb{H}_\Lambda-\mathcal{N}+\sum_{n=1}^N a_n^\dagger a_n\Big|\Psi\Big\rangle+C'\alpha^{-2(1+\sigma)}\geq \braket{\Psi|\mathbb{K}|\Psi}=\braket{\mathbb{K}}_{\gamma},
\end{align}
where $\gamma$ is the reduced density matrix on the Hilbert space $\mathcal{F}\left(\Pi  L^2\! \left(\mathbb{R}^3\right)\right)\cong L^2\! \left(\mathbb{R}^{N}\right)$ corresponding to the state $\Psi$, i.e. we trace out the electron component as well as all the modes in the orthogonal complement of $\Pi  L^2\! \left(\mathbb{R}^3\right)$, 
\begin{align*}
\gamma:=\mathrm{Tr}_{L^2\! \left(\mathbb{R}^3\right)\otimes \mathcal{F}\left(L^2\! \left(\mathbb{R}^3\right)\right)\rightarrow \mathcal{F}\left(\Pi  L^2\! \left(\mathbb{R}^3\right)\right)}\left[\, \ket{\Psi}\bra{\Psi}\, \right].
\end{align*}
Note that we have the inequality $W_{\Pi\varphi^\mathrm{Pek}}^{-1}\left(\sum_{n=1}^{N} a_n^\dagger a_n\right)W_{\Pi\varphi^\mathrm{Pek}}\leq W_{\varphi^\mathrm{Pek}}^{-1}\, \mathcal{N}W_{\varphi^\mathrm{Pek}}$. The operators on the left and right hand side commute, and consequently (\ref{Equation-AssumptionCondensation}) implies that  $\chi\left( W_{\Pi\varphi^\mathrm{Pek}}^{-1}\left(\sum_{n=1}^{N} a_n^\dagger a_n\right)W_{\Pi\varphi^\mathrm{Pek}}\leq \alpha^{-h}\right)\Psi=\Psi$. This in particular means that the transformed reduced density matrix $\widetilde{\gamma}:=W_{\Pi\varphi^\mathrm{Pek}} \gamma W_{\Pi\varphi^\mathrm{Pek}}^{-1}$ satisfies
\begin{align}
\label{Equation-Strong condensation of a concrete sequence_coordinate version}
\chi\left( \sum_{n=1}^{N} a_n^\dagger a_n\leq \alpha^{-h}\right)\widetilde{\gamma}=\widetilde{\gamma}.
\end{align}
Using the identification $\varphi=\sum_{n=1}^{N} \lambda_n \varphi_n$ as before, the Weyl operator $W_{\Pi \varphi^\mathrm{Pek}}$ acts as $\left(W_{\Pi \varphi^\mathrm{Pek}}\Psi\right)\left(\lambda\right)=\Psi\left(\lambda+\lambda^\mathrm{Pek}\right)$ with $\lambda^\mathrm{Pek}:=\big(\braket{\varphi_1|\varphi^\mathrm{Pek}},\dots ,\braket{\varphi_{N}|\varphi^\mathrm{Pek}}\big)$. Due to Eq.~(\ref{Equation-Strong condensation of a concrete sequence_coordinate version}), and the fact that $2s<h$ and $\sigma < \frac{1-5s}{4}$, the assumptions of Proposition \ref{Proposition-Large Deviations} are satisfied, and therefore we obtain for any $D>0$ the existence of a constant $\beta>0$ such that for $\alpha$ large enough
\begin{align}
\label{Equation-Small probability 1}
&\int_{|\lambda-\lambda^\mathrm{Pek}|_\diamond\geq \alpha^{-s}D}\left(1+|\lambda-\lambda^\mathrm{Pek}|^2\right) \rho(\lambda)\mathrm{d}\lambda=\int_{|\lambda|_\diamond\geq \alpha^{-s}D}\left(1+|\lambda|^2\right) \widetilde{\rho}(\lambda)\mathrm{d}\lambda\leq e^{-\beta \alpha^{2(1-s)}},\\
& \nonumber \int_{\left|t^\lambda\right|\geq \alpha^{-s}D}\left(1+|\lambda-\lambda^\mathrm{Pek}|^2\right) \rho(\lambda)\mathrm{d}\lambda \leq \sum_{n=1}^3\int_{|\lambda_n|\geq \frac{\alpha^{-s}}{\sqrt{3}}D}\left(1+|\lambda-\lambda^\mathrm{Pek}|^2\right) \rho(\lambda)\mathrm{d}\lambda\\
\label{Equation-Small probability 2}
&\ \ \ \ \ = \sum_{n=1}^3\int_{|\lambda_n|\geq \frac{\alpha^{-s}}{\sqrt{3}}D}\left(1+|\lambda|^2\right) \widetilde{\rho}(\lambda)\mathrm{d}\lambda\leq e^{-\beta \alpha^{2(1-s)}},
\end{align}
where $\rho$ and $\widetilde{\rho}$ are the density functions corresponding to $\gamma$ and $\widetilde{\gamma}$, respectively, and where we have used $t^\lambda=(\lambda_1,\lambda_2,\lambda_3)\in \mathbb{R}^3$. For the concrete choice $D:=1$, Eq.~(\ref{Equation-Small probability 2}) immediately yields the claim $\braket{\Psi|1-\mathbb{B}|\Psi}=\int_{\left|t^\lambda\right|\geq \alpha^{-s}} \rho(\lambda)\mathrm{d}\lambda\leq e^{-\beta\alpha^{2(1-s)}}$.

In order to verify Eq.~(\ref{Equation-QuasiQuadraticLowerBOund}), we need to find a sufficient lower bound for the expectation value $\braket{\mathbb{K}}_{\gamma}$, where $\mathbb{K}$ is the operator from Eq.~(\ref{Equation-Hamiltonian without electron part}). Recall the definition of the transformation $\tau:\Pi L^2\! \left(\mathbb{R}^3\right)\longrightarrow \Pi L^2\! \left(\mathbb{R}^3\right)$ from Definition \ref{Definition-tau} and the operator $J_{t,\epsilon}$ from Eq.~(\ref{Equation-Definition of J}). As a first step we will provide a lower bound on $\braket{\mathcal{F}^\mathrm{Pek}\!\left(\lambda\right)}_{\gamma}$, using Eq.~(\ref{Equation-Pekar Functional}) and the fact that $\sup_{|t|\leq t_0}\left\|\left(1-\Pi \right)\varphi^\mathrm{Pek}_{x_{t}}\right\|^2\lesssim \alpha^{-\frac{12}{5}(1+\sigma)}$ for $t_0$ small enough, which follows from Lemma \ref{Lemma-Cut-off residuum estimate} together with $x_t\underset{t\rightarrow 0}{\longrightarrow} 0$. We define the operator $\mathbb{A}:=\chi\left(|\lambda-\lambda^\mathrm{Pek}|_\diamond<\alpha^{-s}D\right)\chi\left(\left|t^\lambda\right|<\alpha^{-s}D\right)$, where $D$ is as in Theorem \ref{Theorem-Lower Bound on the Pekar Functional}, and estimate
\begin{align}
\nonumber 
&\braket{\mathcal{F}^\mathrm{Pek}\! \left(\lambda\right)}_{\gamma}=\braket{\mathcal{F}^\mathrm{Pek}\! \left(\lambda\right)\mathbb{A}}_{\gamma}+\braket{\mathcal{F}^\mathrm{Pek}\! \left(\lambda\right)\left(1-\mathbb{A}\right)}_{\gamma}\\
\nonumber 
&\ \geq \Big\langle\big(e^\mathrm{Pek}+J_{t^\lambda,\alpha^{-s}}\! \big[ \tau\! \left(\lambda \right)\big]\big)\mathbb{A}\Big\rangle_{\gamma}+\braket{\mathcal{F}^\mathrm{Pek}\! \left(\lambda\right)\left(1-\mathbb{A}\right)}_{\gamma}+O\left(\alpha^{s-\frac{12}{5}(1+\sigma)}\right)\\
\label{Equation-Potential lower bound}
&\ =\Big\langle e^\mathrm{Pek}\!+\!J_{t^\lambda,\alpha^{-s}}\! \big[ \tau\! \left(\lambda \right)\big]\Big\rangle_{\gamma}\! \! \!+\!\Big\langle X \Big\rangle_{\gamma}+\!O\left(\! \alpha^{s-\frac{12}{5}(1+\sigma)}\! \right)
\end{align}
with $X:=\left(\mathcal{F}^\mathrm{Pek}\! \left(\lambda\right)-e^\mathrm{Pek}-J_{t^\lambda,\alpha^{-s}}\! \big[ \tau\! \left(\lambda \right)\big]\right)\left(1-\mathbb{A}\right)$. Using Eqs.~(\ref{Equation-Small probability 1}) and~(\ref{Equation-Small probability 2}) as well as the fact that $1-\mathbb{A}\leq \chi\left(|\lambda-\lambda^\mathrm{Pek}|_\diamond\geq D\alpha^{-s}\right)+\chi\left(\left|t^\lambda\right|\geq D\alpha^{-s}\right)$, we obtain $\Big\langle X \Big\rangle_{\gamma}\lesssim e^{-\beta\alpha^{2(1-s)}}$, where we have used that $\mathcal{F}^\mathrm{Pek}\! \left(\lambda\right)$ and $J_{t^\lambda,\alpha^{-s}}\! \big[ \tau\! \left(\lambda \right)\big]$ are bounded by $C(1+|\lambda|^2)$ for suitable $C>0$. By Eq.~(\ref{Equation-Potential lower bound}) we therefore have the estimate $\braket{\mathcal{F}^\mathrm{Pek}\! \left(\lambda\right)}_{\gamma}\geq \Big\langle e^\mathrm{Pek}\!+\!J_{t^\lambda,\alpha^{-s}}\! \big[ \tau\! \left(\lambda\right)\big]\Big\rangle_{\gamma}\! \! \!+\!O\left(\! \alpha^{s-\frac{12}{5}(1+\sigma)}\! \right)$, and consequently
\begin{align}
\label{Equation-H_*}
\braket{\mathbb{K}}_{\gamma}\geq e^\mathrm{Pek}\! \! +\Big\langle \!  -\frac{1}{4\alpha^4}\sum_{n=1}^{N}\partial_{\lambda_n}^2+J_{t^\lambda,\alpha^{-s}}\! \big[ \tau\! \left(\lambda \right)\big]\Big\rangle_{\gamma}\! \! \! \! -\frac{{N}}{2\alpha^2}+O\left(\alpha^{s-\frac{12}{5}(1+\sigma)}\right).
\end{align}
Since $\Big\langle \!  -\frac{1}{4\alpha^4}\sum_{n=4}^{N}\partial_{\lambda_n}^2+J_{t^\lambda,\alpha^{-s}}\! \big[ \tau\! \left(\lambda \right)\big]\Big\rangle_{\gamma}=\Big\langle \Psi\Big|\!  -\frac{1}{4\alpha^4}\sum_{n=4}^{N}\partial_{\lambda_n}^2+J_{t^\lambda,\alpha^{-s}}\! \big[ \tau\! \left(\lambda \right)\big]\Big|\Psi \Big\rangle$, this concludes the proof together with Eq.~(\ref{Equation-K estimate}).
\end{proof}

In the following, let $\Psi_\alpha$ be the sequence of states constructed in Theorem \ref{Theorem-Existence of a strong condensate}, satisfying $\braket{\Psi_\alpha|\mathbb{H}_\Lambda|\Psi_\alpha}-E_\alpha\lesssim \alpha^{-2(1+\sigma)}$, $\mathrm{supp}\left(\Psi_\alpha\right)\subset B_{4L}(0)$ with $L=\alpha^{1+\sigma}$ and strong condensation with respect to $\varphi^\mathrm{Pek}$, i.e. $\chi\left( W_{\varphi^\mathrm{Pek}}^{-1}\, \mathcal{N}W_{\varphi^\mathrm{Pek}}\leq \alpha^{-h}\right)\Psi_\alpha=\Psi_\alpha$, and furthermore let $s<\frac{1}{29}$ be a given constant and let us choose $\sigma$ and $h$ such that $2s<h<\frac{2}{29}$ and $\frac{s}{2}\leq \sigma <\frac{1-5s}{4}$. Note that $h<\frac{2}{29}$ makes sure that the assumption of Theorem \ref{Theorem-Existence of a strong condensate} is satisfied, while $2s<h$ and $\sigma <\frac{1-5s}{4}$ are necessary in order to apply Lemma \ref{Lemma-CombiningResults}. The final assumption $\frac{s}{2}\leq \sigma$ will be useful later in Eq.~(\ref{Equation-FinalEstimate}) in order to make sure that $\alpha^{-2(1+\sigma)}\leq \alpha^{-(2+s)}$. Making use of $-\frac{1}{4\alpha^4}\sum_{n=1}^3 \partial_{\lambda_n}^2\geq 0$ and $\mathcal{N}\geq \sum_{n=1}^N a_n^\dagger a_n$, we obtain by Lemma \ref{Lemma-CombiningResults} that
\begin{align}
\label{Equation-EssentialEnergyEstimate}
E_\alpha \geq  e^\mathrm{Pek} \!+ \!\Big\langle\Psi_\alpha\Big| \! -\frac{1}{4\alpha^4}\sum_{n=4}^{N}\partial_{\lambda_n}^2 \!+ \!J_{t^\lambda,\alpha^{-s}}\! \big[ \tau\! \left(\lambda \right)\big]\Big|\Psi_\alpha\Big\rangle \! - \!\frac{{N}}{2\alpha^2} \!+ \!O\left(\alpha^{-2(1+\sigma)}\right)
\end{align}
for a suitable $C'$, where we have used $\alpha^{s-\frac{12}{5}(1+\sigma)}\leq \alpha^{-2(1+\sigma)}$ and $E_\alpha-\braket{\Psi_\alpha|\mathbb{H}_\Lambda|\Psi_\alpha}\gtrsim -\alpha^{-2(1+\sigma)}$. In order to further estimate the expectation value in Eq.~(\ref{Equation-EssentialEnergyEstimate}), let us define the unitary transformation $\left(\mathcal{U}\Psi\right)(\lambda):=\Psi\left(\tau'\left(\lambda\right)\right)$ with $\tau'\left(\lambda\right):=\Big(\big\langle \varphi_n|\tau\left(\lambda\right)\big\rangle \Big)_{n=1}^{N}\in \mathbb{R}^N$. Since $\tau'$ acts as a shift operator on each of the planes $X_t:=\{\lambda:(\lambda_1,\lambda_2,\lambda_3)=t\}$ for $t\in \mathbb{R}^3$, it is clear that $\mathrm{det}\, \mathrm{D}|_\lambda \tau'=1$, which in particular means that the operator $\mathcal{U}$ is indeed unitary, and we have $\partial_{\lambda_n}=\mathcal{U}^{-1}\partial_{\lambda_n}\mathcal{U}$ for $n\geq 4$. Furthermore we define the operator 
\begin{align*}
\mathbb{Q}_{t,\epsilon}:= -\frac{1}{4\alpha^4}\sum_{n=4}^{N}\partial_{\lambda_n}^2+\sum_{n,m=1}^{N}\left(J_{t,\epsilon}\right)_{n,m}\, \lambda_n\lambda_m
\end{align*}
with $\left(J_{t,\epsilon}\right)_{n,m}:=\braket{\varphi_n|J_{t,\epsilon}|\varphi_m}$. Note that $\left(J_{t,\epsilon}\right)_{n,m}=\left(J_{t,\epsilon}\right)_{m,n}=0$ in case $n\in \{1,2,3\}$, i.e. the operator $\mathbb{Q}_{t,\epsilon}$ depends only on the variables $\lambda_n$ for $n\geq 4$ and not on $t^\lambda=(\lambda_1,\lambda_2,\lambda_3)$, hence it acts on the Fock space $\mathcal{F}\big(\mathrm{span}\{\varphi_4,\dots ,\varphi_{N}\}\big)\cong L^2\! \left(\mathbb{R}^{N-3}\right)$ only. Utilizing the fact that $\mathcal{U}^{-1}J_{t^\lambda,\alpha^{-s}}\! \big[ \tau\! \left(\lambda \right)\big]\, \mathcal{U}=J_{t^\lambda,\alpha^{-s}}\! \big[ \lambda\big]=\sum_{n,m=1}^{N}\left(J_{t^\lambda,\alpha^{-s}}\right)_{n,m}\, \lambda_n\lambda_m$, where we used that $\mathcal{U}^{-1}t^\lambda\mathcal{U}=t^\lambda$, we obtain
\begin{align*}
&\mathcal{U}^{-1}\! \! \left( \! -\frac{1}{4\alpha^4}\sum_{n=4}^{N}\partial_{\lambda_n}^2\!\! +\!J_{t^\lambda,\alpha^{-s}}\! \big[ \tau\! \left(\lambda \right)\big] \right) \!  \mathcal{U}=\mathbb{Q}_{t^\lambda,\alpha^{-s}}\geq \mathbb{Q}_{t^\lambda,\alpha^{-s}}\mathbb{B}\geq \inf_{|t|<\alpha^{-s}}\inf \sigma\left(\mathbb{Q}_{t,\alpha^{-s}}\right)\mathbb{B},
\end{align*}
where $\mathbb{B}$ is as in Lemma \ref{Lemma-CombiningResults}. Here we have used $\mathbb{Q}_{t^\lambda,\alpha^{-s}}\geq 0$, which follows from Lemma \ref{Lemma-Convergence of the trace}, as well as the fact that $1-\mathbb{B}$ is non-negative and commutes with $\mathbb{Q}_{t^\lambda,\alpha^{-s}}$. Applying this inequality with respect to the state $\widetilde{\Psi}_\alpha:=\mathcal{U}^{-1}\Psi_\alpha$ yields
\begin{align}
\nonumber
&\Big\langle \Psi_\alpha\Big|  -\frac{1}{4\alpha^4}\sum_{n=4}^{N}\partial_{\lambda_n}^2\! +\! J_{t^\lambda,\alpha^{-s}}\! \big[ \tau\! \left(\lambda \right)\big]\Big|\Psi_\alpha \Big\rangle \geq \inf_{|t|<\alpha^{-s}}\inf \sigma\left(\mathbb{Q}_{t,\alpha^{-s}}\right)\big\langle \widetilde{\Psi}_\alpha\big| \mathbb{B} \big| \widetilde{\Psi}_\alpha \big\rangle \\
\label{Equation-Applying U}
&\ \ \ \geq \inf_{\left|t\right|<\alpha^{-s}}\! \! \! \inf \! \sigma\left(\mathbb{Q}_{t,\alpha^{-s}}\right)\! -\! \frac{{N}}{2\alpha^2}\big\langle \widetilde{\Psi}_\alpha\big| 1-\mathbb{B}\big| \widetilde{\Psi}_\alpha \big\rangle
\end{align}
where we have used $J_{t,\epsilon}\leq 1$, and therefore $\inf \sigma\left(\mathbb{Q}_{t,\epsilon}\right)\leq \frac{{N}}{2\alpha^2}$.  By Lemma \ref{Lemma-CombiningResults}, we know that $\big\langle \widetilde{\Psi}_\alpha\big| 1-\mathbb{B}\big| \widetilde{\Psi}_\alpha \big\rangle=\big\langle \Psi_\alpha\big| 1-\mathbb{B}\big| \Psi_\alpha \big\rangle\leq e^{-\beta \alpha^{2-2s}}$. Combining Eqs.~(\ref{Equation-EssentialEnergyEstimate}) and~(\ref{Equation-Applying U}), and making use of the fact that ${N}\lesssim \alpha^p$ for some $p>0$, yields
\begin{align}
\label{Equation-Pre final}
E_\alpha\geq e^\mathrm{Pek}+\inf_{\left|t\right|<\alpha^{-s}}\inf \sigma\left(\mathbb{Q}_{t,\alpha^{-s}}\right)-\frac{{N}}{2\alpha^2}+O\left(\alpha^{-2(1+\sigma)}\right).
\end{align}
Since the operator $\mathbb{Q}_{t,\alpha^{-s}}$ is quadratic in $\partial_{\lambda_n}$ and $\lambda_n$, we have an explicit formula for its ground state energy, given by
\begin{equation}
\inf \sigma\left(\mathbb{Q}_{t,\alpha^{-s}}\right)-\frac{N}{2\alpha^2}=-\frac{\mathrm{Tr}_{\Pi L^2\! \left(\mathbb{R}^3\right)}\! \! \left[1 -\sqrt{  J_{t,\alpha^{-s}} }\right]}{2\alpha^2},
\end{equation}
where we used the fact that $J_{t,\alpha^{-s}}\geq 0$  for $\alpha$ large enough, as shown in Lemma \ref{Lemma-Convergence of the trace}. Using Eq.~(\ref{Equation-Convergence of the trace}), we can approximate this quantity by
\begin{align*}
\sup_{|t|<\alpha^{-s}}\left|\mathrm{Tr}_{\Pi L^2\!  \left(\mathbb{R}^3\right)}\! \!\left[1 -\sqrt{  J_{t,\alpha^{-s}} }\, \right]-\mathrm{Tr}\left[1-\sqrt{H^\mathrm{Pek}}\, \right]\right|\lesssim \alpha^{-s}+\alpha^{-\frac{1}{5}},
\end{align*}
where $H^\mathrm{Pek}$ is defined in Eq.~(\ref{Equation-Definition_Hessian}). Consequently  Eq.~(\ref{Equation-Pre final}) yields
\begin{align}
\label{Equation-FinalEstimate}
E_\alpha\! -\! e^\mathrm{Pek}\! +\! \frac{1}{2\alpha^2}\mathrm{Tr}\left[1\! -\! \sqrt{H^\mathrm{Pek}}\, \right]\gtrsim \! - \alpha^{-2(1+\sigma)}\!  -\! \alpha^{-(2+s)}\! -\! \alpha^{-\left(2+\frac{1}{5}\right)},
\end{align}
which concludes the proof, since all the terms on the right side are of order $\alpha^{-(2+s)}$.

\section{Approximation by Coherent States}
\label{Convergence to Coherent States}

This section is devoted to the proof of Theorem \ref{Theorem-Coherent States}, which states that asymptotically the phonon part of any low energy state is a convex combination of the coherent states $\Omega_{\varphi^\mathrm{Pek}_x}$ with $x\in \mathbb{R}^3$, where the convex combination is taken on the level of density matrices. As a central tool we will verify in Lemma \ref{Lemma-Coherent State Methode} an asymptotic formula for the expectation value $\big\langle\Psi\big|\widehat{F}\big|\Psi\big\rangle$ in terms of the lower symbol $\mathbb{P}_y$ corresponding to the state $\Psi$, see Eq.~(\ref{Equation-Definition of the y-dependent measure}). Furthermore we will make use of the inequality
\begin{align}
\label{Equation-CoercivityPekarFunctional}
\inf_{x\in \mathbb{R}^3}\|\varphi-\varphi^\mathrm{Pek}_x\|^2\lesssim \mathcal{F}^\mathrm{Pek}\! \left(\varphi\right)- e^\mathrm{Pek}
\end{align}
derived in \cite[Lemma~7]{FeRS}, which implies that the only coherent states $\Omega_\varphi$ with a low energy have their point of condensation $\varphi$ close to the manifold of Pekar minimizers $\{\varphi^\mathrm{Pek}_x:x\in \mathbb{R}^3\}$. We start with the subsequent Lemma \ref{Lemma-Fock space expression for F hat}, which provides an asymptotic formula for $\widehat{F}$ operators in terms of creation and annihilation operators.

\begin{lem}
\label{Lemma-Fock space expression for F hat}
Let $m\in \mathbb{N}$ and $C>0$ be given constants, $\{g_n:n\in \mathbb{N}\}$ an orthonormal basis of $L^2\!\left(\mathbb{R}^3\right)$ and let us denote $a_n:=a(g_n)$. Then there exists a constant $T>0$ such that for all functions $F$ of the form
\begin{align}
\label{Equation-F function}
F\left(\rho\right)=\int\dots \int f(x_1,\dots , x_m)\, \mathrm{d}\rho(x_1)\dots \mathrm{d}\rho(x_m),
\end{align}
with $f:\mathbb{R}^{3\times m}\longrightarrow \mathbb{R}$ bounded, and states $\Psi$ satisfying $\chi\left(\mathcal{N}\leq C\right)\Psi=\Psi$, we can approximate the operator $\widehat{F}$ from Definition \ref{Definition-Hat operators} by
\begin{align}
\label{Equation-Rigorous Definition}
\left|\big\langle\Psi\big|\widehat{F}\big|\Psi\big\rangle-\sum_{I,J\in \mathbb{N}^m} f_{I,J}\big\langle\Psi\big|a_{I_1}^\dagger\dots a_{I_m}^\dagger a_{J_1}\dots a_{J_m}\big|\Psi\big\rangle\right|\leq  T\|f\|_\infty\alpha^{-2},
\end{align}
where we interpret $f$ as a multiplication operator on $L^2\! \left(\mathbb{R}^3\right)^{\otimes^m}\cong L^2\! \left(\mathbb{R}^{3\times m}\right)$ and denote the matrix elements $f_{I,J}:=\braket{g_{I_1}\otimes \dots \otimes g_{I_m}|f|g_{J_1}\otimes \dots \otimes g_{J_m}}$.
\end{lem}
\begin{proof}
By the assumption $\chi\left(\mathcal{N}\leq C\right)\Psi=\Psi$, we can represent the state $\Psi$ as $\Psi=\bigoplus_{n\leq C\alpha^2} \Psi_n$ where $\Psi_n(y,x^1,\dots ,x^{n})$ is a function of the electron variable $y$ and the $n$ phonon coordinates $x^k\in \mathbb{R}^3$. As in the proof of Lemma \ref{Lemma-IMS Formula}, we will suppress the dependence on the electron variable $y$ in our notation. Using the definition of $\widehat{F}$ in Definition \ref{Definition-Hat operators}, as well as the notation $X=(x^1,\dots,x^n)$, we can write
\begin{align*}
\big\langle\Psi\big|\widehat{F}\big|\Psi\big\rangle&=\sum_{n\leq C\alpha^2}\int_{\mathbb{R}^{3n}}F\left(\alpha^{-2}\sum_{k=1}^n\delta_{x^k}\right)|\Psi_n(X)|^2\mathrm{d}X\\
&=\alpha^{-2m}\sum_{n\leq C\alpha^2}\sum_{k\in \{1,\dots,n\}^m}\int_{\mathbb{R}^{3n}}f(x^{k_1},\dots,x^{k_m})|\Psi_n(X)|^2\mathrm{d}X.
\end{align*}
Defining $\mathcal{K}$ as the set of all $k\in \{1,\dots,n\}^m$ satisfying $k_i\neq k_j$ for all $i\neq j$, we can further express the operator $\sum_{I,J\in \mathbb{N}^m} f_{I,J}\, a_{I_1}^\dagger\dots a_{I_m}^\dagger a_{J_1}\dots a_{J_m}$ as
\begin{align*}
\sum_{I,J\in \mathbb{N}^m}\! \! \!  f_{I,J}\big\langle\Psi\big|a_{I_1}^\dagger\dots a_{I_m}^\dagger a_{J_1}\dots a_{J_m}\big|\Psi\big\rangle= \alpha^{-2m}\! \! \sum_{n\leq C\alpha^2}\sum_{k\in \mathcal{K}}\int_{\mathbb{R}^{3n}}\! \! f(x^{k_1},\dots,x^{k_m})|\Psi_n(X)|^2\mathrm{d}X.
\end{align*}
Consequently we can identify the left hand side of Eq.~(\ref{Equation-Rigorous Definition}) as 
\begin{align*}
&\left|\alpha^{-2m}\sum_{n\leq C\alpha^2}\sum_{k\in \{1,\dots,n\}^m\setminus \mathcal{K}}\int_{\mathbb{R}^{3n}}f(x^{k_1},\dots,x^{k_m})|\Psi_n(X)|^2\mathrm{d}X\right| \\ &\leq 
\|f\|_\infty \sum_{n\leq C\alpha^2}\left(\sum_{k\in \{1,\dots,n\}^m\setminus \mathcal{K}}\alpha^{-2m}\right)\int_{\mathbb{R}^{3n}}|\Psi_n(X)|^2\mathrm{d}X.
\end{align*}
Since $\sum_{k\in \{1,\dots,n\}^m\setminus \mathcal{K}}\alpha^{-2m}=\left(n^m-\frac{n!}{(n-m)!}\right)\alpha^{-2m}\leq m 2^m n^{m-1}\alpha^{-2m}\lesssim \alpha^{-2}$ for $n\leq C\alpha^2$ and since $\sum_{n\leq C\alpha^2}\int_{\mathbb{R}^{3n}}|\Psi_n(X)|^2\mathrm{d}X=\|\Psi\|^2=1$, this concludes the proof.
\end{proof}

In the following we are going to define the lower symbol $\mathbb{P}_y$ corresponding to a state $\Psi\in L^2\Big(\mathbb{R}^3,\mathcal{F}\big(L^2\! \left(\mathbb{R}^3\right)\big)\Big)$. Since we consider the Fock space over the infinite dimensional Hilbert space $L^2\! \left(\mathbb{R}^3\right)$, we need to modify the usual definition of the lower symbol by introducing suitable localizations. For $0<s\leq \frac{4}{27}$ and $y\in \mathbb{R}^3$, let us define $\ell_*:=\alpha^{-\frac{5}{2}s}$ and $\Lambda_*:=\alpha^{2s}$, and the projection 
\begin{align}
\label{Equation-Definition of weak projection}
\Pi_{y}:=\Pi^y_{\Lambda_*,\ell_*},
\end{align}
see Definition \ref{Definition-Pi}. We have $N_*:=\mathrm{dim}\,\Pi_{y}L^2\! \left(\mathbb{R}^3\right)\lesssim \left({\Lambda_*}/{\ell_*}\right)^3\leq \alpha^2$ by our assumption $s\leq \frac{4}{27}$. Using the notation $\{e_{y,1},\dots ,e_{y,{N_*}}\}$ for the orthonormal basis of $\Pi_{y}L^2\! \left(\mathbb{R}^3\right)$ from Definition \ref{Definition-Pi}, we introduce for $\xi\in \mathbb{C}^{N_*}$ the coherent states $\Omega_{y,\xi}:=e^{\alpha^2 a^\dagger\left(\varphi_{y,\xi}\right)- \alpha^2 a\left(\varphi_{y,\xi}\right)}\Omega$, where $\Omega$ is the vacuum in $\mathcal{F}\left(\Pi_y L^2\! \left(\mathbb{R}^3\right)\right)$ and $\varphi_{y,\xi}:=\sum_{n=1}^{N_*} \xi_n e_{y,n}\in \Pi_{y}L^2\! \left(\mathbb{R}^3\right)$. Furthermore we define wave-functions $\Psi_y$ localized in the electron coordinates $x$ as 
\begin{align}
\label{Definition-Psi_y}
\Psi_{y}(x):=L_*^{-\frac{3}{2}}\chi\left(\frac{x-y}{L_*}\right)\Psi(x),
\end{align}
where $y\in \mathbb{R}^3$ and $L_*:=\alpha^{\frac{s}{2}}$, and $\chi$ is a smooth non-negative function with $\mathrm{supp} \left(\chi\right)\subset B_1(0)$ and $\int \chi(y)^2\, \mathrm{d}y=1$. For the following construction, note that we can identify $L^2\Big(\mathbb{R}^3,\mathcal{F}\big(L^2\! \left(\mathbb{R}^3\right)\big)\Big)\cong\mathcal{F}\left(\Pi_y L^2\! \left(\mathbb{R}^3\right)\right)\otimes L^2\! \Big(\mathbb{R}^3,\mathcal{F}\Big(\Pi_y L^2\! \left(\mathbb{R}^3\right)^\perp\Big)\Big)$. Let us now define measures $\mathbb{P}_y$ on $\mathbb{C}^{N_*}\cong \mathbb{R}^{2{N_*}}$ corresponding to the state $\Psi_y$ as 
\begin{align}
\label{Equation-Definition of the y-dependent measure}
\frac{\mathrm{d}\mathbb{P}_y}{\mathrm{d}\xi}:=\frac{1}{\pi^{N_*}}\left\|\Theta_{y,\xi}\Psi_{y}\right\|^2,
\end{align}
where $\Theta_{y,\xi}$ is the orthogonal projection onto the set spanned by elements of the form $\Omega_{y,\xi}\otimes \widetilde{\Psi}$ with $\widetilde{\Psi}\in L^2\Big(\mathbb{R}^3,\mathcal{F}\Big(\Pi_y L^2\! \left(\mathbb{R}^3\right)^\perp\Big)\Big)$. Note that the coherent states $\Omega_{y,\xi}$ provide a resolution of the identity $\frac{1}{\pi^{N_*}}\int_{\mathbb{C}^{N_*}}\ket{\Omega_{y,\xi}}\bra{\Omega_{y,\xi}}\mathrm{d}\xi=1_{\mathcal{F}\left(\Pi_y L^2\! \left(\mathbb{R}^3\right)\right)}$, see for example \cite{LT}, and consequently the projections $\Theta_{y,\xi}$ satisfy $\frac{1}{\pi^{N_*}}\int_{\mathbb{C}^{N_*}}\Theta_{y,\xi}\, \mathrm{d}\xi=1$. In particular we see that the total mass of the measure $\mathbb{P}_y$ is $\int \mathrm{d}\mathbb{P}_y=\|\Psi_y\|^2$ and therefore
\begin{align*}
\iint \mathrm{d}\mathbb{P}_y\mathrm{d}y=\int \|\Psi_y\|^2\mathrm{d}y=\|\Psi\|^2=1.
\end{align*}
 In the following Lemma \ref{Lemma-Coherent State Methode} and Corollary \ref{Corollary-Coherent State Methode} we will provide an asymptotic formula for the expectation value $\big\langle\Psi_y\big|\widehat{F}\big|\Psi_y\big\rangle$, respectively $\big\langle\Psi\big|\widehat{F}\big|\Psi\big\rangle$, in terms of the measures $\mathbb{P}_y$.

\begin{lem}
\label{Lemma-Coherent State Methode}
Given $m\in \mathbb{N}$, $C>0$ and $g\in L^2\! \left(\mathbb{R}^3\right)$, there exists a $T>0$ such that for all $F$ of the form~(\ref{Equation-F function}), $y\in \mathbb{R}^3$ and $\epsilon>0$, and states $\Psi$ satisfying $\chi\left(\mathcal{N}\leq C\right)\Psi=\Psi$
\begin{align}
\label{Equation-Coherent State Methode}
\frac{1}{T\|f\|_{\infty}}\! \left| \big\langle \Psi_y\big|\widehat{F}\,\big|\Psi_y \big\rangle \! -\! \int F\left(|\varphi_{y,\xi}|^2\right) \mathrm{d}\mathbb{P}_y\! \left(\xi\right) \right|\! \leq \!  \left(\! \frac{{N_*}}{\alpha^2}+\epsilon\! \right)\!  \|\Psi_y\|^2\! +\! \epsilon^{-1}\big\langle \Psi_y\big|\mathcal{N}^y_{>{N_*}}\big|\Psi_y\big\rangle,
\end{align}
with $\mathcal{N}^y_{>{N_*}}:=\mathcal{N}-\sum_{n=1}^{N_*} a_{y,n}^\dagger a_{y,n}$ and $a_{y,n}:=a\left(e_{y,n}\right)$, and furthermore
\begin{align}
\label{Equation-Coherent State Methode_Second Part}
\frac{1}{T}\left| \big\langle \Psi_y\big|W_{g}^{-1} \mathcal{N}W_{g}\big|\Psi_y \big\rangle\!  - \! \! \int \! \! \|\varphi_{y,\xi}\! -\! g\|^2 \mathrm{d}\mathbb{P}_y\! \left(\xi\right)\right|\! \leq \! \left(\! \frac{{N_*}}{\alpha^2}\! +\! \epsilon\! \right)\!  \|\Psi_y\|^2\! +\! \epsilon^{-1}\big\langle \Psi_y\big|\mathcal{N}^y_{>{N_*}}\big|\Psi_y\big\rangle,
\end{align}
where $W_{g}$ is the corresponding Weyl transformation.
\end{lem}
\begin{proof}
Let $\{g_n:n\in \mathbb{N}\}$ be a completion of $\{e_{y,1},\dots ,e_{y,{N_*}}\}$ to an orthonormal basis of $L^2\! \left(\mathbb{R}^3\right)$ and let us define $a_n:=a\left(g_n\right)$. We further introduce an operator $\widetilde{F}$ 
as
\begin{align}
\label{Equation-Tilde F}
\widetilde{F}:=\! \! \! \! \! \! \!  \sum_{I,J\in \{1,\dots ,{N_*}\}^m}\! \! \! \! \! \! \! \! f_{I,J}\, a_{I_1}^\dagger\! \dots a_{I_m}^\dagger a_{J_1}\! \dots  a_{J_m}\! =\! \! \! \! \!  \sum_{I,J\in \mathbb{{N}}^m}\!  \! \! \! \! \left(\Pi_y^{\otimes^m}f\Pi_y^{\otimes^m}\right)_{I,J} a_{I_1}^\dagger \! \dots a_{I_m}^\dagger a_{J_1}\! \dots  a_{J_m}.
\end{align}
In the following we want to verify that both $\|f\|_\infty^{-1} \left|\big\langle \Psi_y\big|\widehat{F}\,\big|\Psi_y \big\rangle -\big\langle \Psi_y\big|\widetilde{F}\,\big|\Psi_y \big\rangle\right|$ and $\|f\|_\infty^{-1} \left|\big\langle \Psi_y\big|\widetilde{F}\,\big|\Psi_y \big\rangle -\int F\left(|\varphi_{y,\xi}|^2\right)\, \mathrm{d}\mathbb{P}_y\left(\xi\right)\right|$ are, up to a multiplicative constant, bounded by the right hand side of Eq.~(\ref{Equation-Coherent State Methode}). Applying the Cauchy--Schwarz inequality, we obtain for all $\epsilon>0$ 
\begin{align*}
\pm \! \left(f\! -\!\Pi_y^{\otimes^m}f\, \Pi_y^{\otimes^m}\right)&=\pm f\left(1\! -\! \Pi_y^{\otimes^m}\right)\! \pm \! \left(1\! -\! \Pi_y^{\otimes^m}\right)f\, \Pi_y^{\otimes^m}\leq \epsilon\|f\|_{\infty}\!+\!\epsilon^{-1}\|f\|_{\infty}\left(1\!-\!\Pi_y^{\otimes^m}\right)\\
&\leq \epsilon\|f\|_{\infty}\!+\!\epsilon^{-1}\|f\|_{\infty}\big((1\!-\!\Pi_y)_1\!+\!\dots \!+\!(1\!-\!\Pi_y)_m\big),
\end{align*}
where $(1-\Pi_y)_j$ means that the operator $1-\Pi_y$ acts on the $j$-th factor in the tensor product. Consequently we have the operator inequality
\begin{align*}
\pm \left(\sum_{I,J\in \mathbb{N}^m} f_{I,J}\, a_{I_1}^\dagger\dots a_{I_m}^\dagger a_{J_1}\dots a_{J_m}-\widetilde{F}\right)\leq \epsilon \|f\|_{\infty}\mathcal{N}^m+\epsilon^{-1}\|f\|_{\infty}m\, \mathcal{N}^y_{>{N_*}}\mathcal{N}^{m-1}.
\end{align*}
Making use of Eq.~(\ref{Equation-Rigorous Definition}) and the fact that $\chi\left(\mathcal{N}\leq C\right)\Psi_y=\Psi_y$ further yields
\begin{align*}
\left|\braket{\Psi_y|\widehat{F}|\Psi_y}-\sum_{I,J\in \mathbb{N}^m} f_{I,J}\,\braket{\Psi_y| a_{I_1}^\dagger\dots a_{I_m}^\dagger a_{J_1}\dots a_{J_m}|\Psi_y}\right|\leq d\alpha^{-2}\|f\|_\infty \|\Psi_y\|^2
\end{align*}
for a suitable constant $d>0$. We have thus shown the bound
\begin{align}
\label{Equation-Difference Hat and Tilde}
\frac{1}{\|f\|_{\infty}}\left|\! \braket{\Psi_y|\widehat{F}\, |\Psi_y}\! -\! \braket{\Psi_y|\widetilde{F}\, |\Psi_y }\! \right|\! \leq \! \left(d\alpha^{-2}\! +\! \epsilon C^m\right) \! \|\Psi_y\|^2\! +\! \epsilon^{-1}m C^{m-1}\! \big\langle \Psi_y\big|\mathcal{N}^y_{>{N_*}}\big|\Psi_y\big\rangle
\end{align}
which is of the desired form.
%

In order to verify that $\frac{1}{\|f\|_{\infty}}\left|\big\langle \Psi_y\big|\widetilde{F}\,\big|\Psi_y \big\rangle -\int F\left(|\varphi_{y,\xi}|^2\right)\, \mathrm{d}\mathbb{P}_y\left(\xi\right)\right|$ is of the same order as the right hand side of Eq.~(\ref{Equation-Coherent State Methode}) as well, we will first compute $\widetilde{F}$ with reversed operator ordering, i.e. we compute
\begin{align}
\label{Equation-Reordering}
&\sum_{I,J\in \{1,\dots ,{N_*}\}^m} f_{I,J}\, a_{J_1}\dots  a_{J_m} a_{I_1}^\dagger\dots a_{I_m}^\dagger=    \sum_{I,J\in \{1,\dots ,{N_*}\}^m}\! \! \! \! \! \! \! \! f_{I,J}\, a_{I_1}^\dagger\! \dots a_{I_m}^\dagger a_{J_1}\! \dots  a_{J_m}\\
\nonumber
&\ \ \ \ \ \ \  \ \ \ \ \  \  \ \ \ \ \ \ \  +\sum_{n=1}^m \frac{1}{\alpha^{2n}n!}\sum_{\sigma,\tau\in \mathcal{S}^{m,n}} \! \! \left(\sum_{I',J'} f^{\sigma,\tau}_{I',J'}\! \prod_{k\notin  \{\sigma_1,\dots,\sigma_n\}}a_{I'_k}^\dagger \prod_{\ell \notin  \{\tau_1,\dots,\tau_n\}}a_{J'_\ell}\right)
\end{align}
where $\mathcal{S}^{m,n}$ is the set of all sequences $\sigma=(\sigma_1,\dots ,\sigma_n)$ without repetitions having values  $\sigma_k\in \{1,\dots,m\}$ and the coordinate matrices $f^{\sigma,\tau}$ are defined as
\begin{align*}
f^{\sigma,\tau}_{I',J'}:=\sum_{I,J\in \{1,\dots,N_*\}^m} f_{I,J}\, \delta_{I_{\sigma_1},J_{\tau_1}}\dots \delta_{I_{\sigma_n},J_{\tau_n}}\prod_{k\notin  \{\sigma_1,\dots,\sigma_n\}}\delta_{I_k,I'_k}\prod_{\ell \notin  \{\tau_1,\dots,\tau_n\}}\delta_{J_\ell,J'_\ell}
\end{align*}
for $I'\in \{1,\dots,N_*\}^{\{1,\dots,m\}\setminus \{\sigma_1,\dots,\sigma_n\}}$ and $J'\in \{1,\dots,N_*\}^{\{1,\dots,m\}\setminus\{\tau_1,\dots,\tau_n\}}$. One can verify Eq.~(\ref{Equation-Reordering}) either by iteratively applying the (rescaled) canonical commutation relations $[a_i,a_j^\dagger]=\alpha^{-2}\delta_{i,j}$, or by using the fact that the operator $e^{\alpha^{-2}\nabla_{\bar{\xi}}\nabla_\xi}$, which is well defined on polynomials in $\xi$ and $\bar{\xi}$, transforms the upper symbol into the lower symbol (see e.g. \cite{perel}), and computing its action on $P(\xi):=\sum_{I,J\in \{1,\dots ,{N_*}\}^{m}}  f_{I,J} \overline{\xi_{I_1}\dots \xi_{I_m}}\xi_{J_1}\dots \xi_{J_m} $ as
\begin{align*}
e^{\alpha^{-2}\nabla_{\bar{\xi}}\nabla_\xi}\left(P\right)(\xi)=   P(\xi) +\sum_{n=1}^m \frac{1}{\alpha^{2n}n!}\sum_{\sigma,\tau\in \mathcal{S}^{m,n}} \! \! \left(\sum_{I',J'} f^{\sigma,\tau}_{I',J'}\! \prod_{k\notin  \{\sigma_1,\dots,\sigma_n\}}\overline{\xi_{I'_k}} \prod_{\ell \notin  \{\tau_1,\dots,\tau_n\}}\xi_{J'_\ell}\right).
\end{align*}
In order to identify the left hand side of Eq.~(\ref{Equation-Reordering}), we will make use of the resolution of identity $\frac{1}{\pi^{N_*}}\int_{\mathbb{C}^{N_*}}\Theta_{y,\xi}\, \mathrm{d}\xi=1$, where $\Theta_{y,\xi}$ is defined below Eq.~(\ref{Equation-Definition of the y-dependent measure}), which allows us to rewrite the anti-wick ordered term $a_{J_1}\dots  a_{J_m}a_{I_1}^\dagger\dots a_{I_m}^\dagger$ as
\begin{align*}
\frac{1}{\pi^{N_*}}\int_{\mathbb{C}^{N_*}}a_{J_1}\dots  a_{J_m}\Theta_{y,\xi}a_{I_1}^\dagger\dots a_{I_m}^\dagger\, \mathrm{d}\xi=\frac{1}{\pi^{N_*}}\int_{\mathbb{C}^{N_*}}\xi_{J_1}\dots \xi_{J_m} \overline{\xi_{I_1}\dots \xi_{I_m}}\Theta_{y,\xi}\, \mathrm{d}\xi.
\end{align*}
Here we have used that $a_i\Theta_{y,\xi}=\xi_i\Theta_{y,\xi}$ for all $i\in \{1,\dots,N_*\}$. By the definition of $\mathbb{P}_y$ in Eq.~(\ref{Equation-Definition of the y-dependent measure}) we can therefore rewrite the expectation value of the first term on the left hand side of Eq.~(\ref{Equation-Reordering}) with respect to the state $\Psi_y$ as
\begin{align}
\nonumber
&\sum_{I,J\in \{1,\dots ,{N_*}\}^{m}}\! \! \! \!  \! \! \! \! \!\!  \! \! \! f_{I,J}\,  \big\langle \Psi_y\big|  a_{J_1}\dots  a_{J_m} a_{I_1}^\dagger\dots a_{I_m}^\dagger\big| \Psi_y\big\rangle\!  = \! \!  \! \! \! \! \! \sum_{I,J\in \{1,\dots ,{N_*}\}^{m}}  \! \! \! \! \! \! \! \! \! \! \!  f_{I,J}\int   \xi_{J_1}\dots \xi_{J_m} \overline{\xi_{I_1}\dots \xi_{I_m}}\mathrm{d}\mathbb{P}_y(\xi)\\
\label{Equation-IdentifyingAntiWick}
&\ \ \ \ \ \ \ \ \ =\int \big\langle \varphi_{y,\xi}^{\otimes^m}\big|f\big|\varphi_{y,\xi}^{\otimes^m}\big\rangle\, \mathrm{d}\mathbb{P}_y(\xi)=\int F\left(\left|\varphi_{y,\xi}\right|^2\right)\, \mathrm{d}\mathbb{P}_y\left(\xi\right).
\end{align}
In order to control the terms in the second line of Eq.~(\ref{Equation-Reordering}),  we can estimate the norm $\left\|f^{\sigma,\tau}\right\|_{\mathrm{op}}\leq \|f\|_\infty N_*^{n}$ for all $\sigma,\tau\in \mathcal{S}^{m,n}$, which follow from
\begin{align*}
\braket{v|f^{\sigma,\tau}|w}&=\sum_{I,J\in \{1,\dots,N_*\}^m}f_{I,J}\delta_{I_{\sigma_1},J_{\tau_1}}\dots \delta_{I_{\sigma_n},J_{\tau_n}}\overline{v_{I'}}w_{J'}=\sum_{k\in \{1,\dots,N_*\}^n}\big\langle v^k\big|f\big|w^k\big\rangle\\
&\leq \|f\|_\infty \sum_{k\in \{1,\dots,N_*\}^n}\|v^k\| \|w^k\|\leq \|f\|_\infty N_*^n \|v\|\|w\|,
\end{align*} 
where $I'$ denotes the restriction of $I$ to $\{1,\dots,m\}\setminus \{\sigma_1,\dots,\sigma_n\}$ and $v^k$ is defined as $\left(v^k\right)_I:=\delta_{I_{\sigma_1},k_1}\dots \delta_{I_{\sigma_n},k_n}v_{I'}$, and $J'$ and $w^k$ are defined analogue. Hence we obtain 
\begin{align*}
\left|\frac{1}{\alpha^{2n}}\sum_{I',J'}f^{\sigma,\tau}_{I',J'}\big\langle  \Psi_y\big| \prod_{k\notin  \{\sigma_1,\dots,\sigma_n\}}a_{I'_k}^\dagger \prod_{\ell \notin  \{\tau_1,\dots,\tau_n\}}a_{J'_\ell}\big| \Psi_y\big\rangle\right| \leq \|f\|_\infty \left(\frac{N_*}{\alpha^2}\right)^n\braket{\Psi_y|\mathcal{{N}}^{m-n}|\Psi_y}
\end{align*}
for $n\geq 1$. Since $\chi\left(\mathcal{N}\leq C\right)\Psi_y=\Psi_y$ and $N_*\lesssim \alpha^2$, see the comment below Eq.~(\ref{Equation-Definition of weak projection}), this is a quantity of order $\|f\|_{\infty}\frac{N_*}{\alpha^2}\|\Psi_y\|^2$. Combing this estimate with Eq.~(\ref{Equation-Reordering}) and Eq.~(\ref{Equation-IdentifyingAntiWick}) yields that $\frac{1}{\|f\|_{\infty}}\left|\big\langle \Psi_y\big|\widetilde{F}\,\big|\Psi_y \big\rangle -\int F\left(|\varphi_{y,\xi}|^2\right)\, \mathrm{d}\mathbb{P}_y\left(\xi\right)\right|$ is, up to a multiplicative factor, bounded by the right hand side of Eq.~(\ref{Equation-Coherent State Methode}). Together with Eq.~(\ref {Equation-Difference Hat and Tilde}), this concludes the proof of Eq.~(\ref{Equation-Coherent State Methode}).

 In order to verify Eq.~(\ref{Equation-Coherent State Methode_Second Part}), let us  define $G(\rho):=\int \mathrm{d}\rho$. Note that $W_{g}^{-1} \mathcal{N}W_{g}=\mathcal{N}-a(g)-a^\dagger(g)+\|g\|^2=\widehat{G\, }-a(g)-a^\dagger(g)+\|g\|^2$. Furthermore we have $\big\langle \Psi_{y}\big|a(\Pi_{y}g)+a^\dagger(\Pi_{y}g)\,\big|\Psi_{y} \big\rangle =\int \left(\braket{g|\varphi_{y,\xi}}+\braket{\varphi_{y,\xi}|g}\right)\mathrm{d}\mathbb{P}_{y}\left(\xi\right)$, where we used that $a(g)+a^\dagger(g)$ is anti-Wick ordered, and
\begin{align*}
\left|\Big\langle\Psi_{y}\Big| a(g)+a^\dagger(g)- a(\Pi_{y} g)-a^\dagger(\Pi_{y} g)\,\Big|\Psi_{y} \Big\rangle \right|\leq \epsilon^{-1}\braket{\Psi_{y}|\mathcal{N}^y_{>{N_*}}|\Psi_{y}}+\epsilon\|g\|^2 \|\Psi_y\|^2.
\end{align*}
Hence, applying Eq.~(\ref{Equation-Coherent State Methode}) with respect to the function $G$ and using that $\int \! \! \|\varphi_{y,\xi} - g\|^2 \mathrm{d}\mathbb{P}_y=\int \left(G\left(|\varphi_{y,\xi}|^2\right)+\braket{g|\varphi_{y,\xi}}+\braket{\varphi_{y,\xi}|g}\right)\mathrm{d}\mathbb{P}_{y}\left(\xi\right)+\|g\|^2\|\Psi_y\|^2$ concludes the proof of Eq.~(\ref{Equation-Coherent State Methode_Second Part}).
\end{proof}

\begin{cor}
\label{Corollary-Coherent State Methode}
Given constants $m\in \mathbb{N},C>0$ and $g\in L^2\! \left(\mathbb{R}^3\right)$, there exists a constant $T>0$ such that for all $F$ of the form~(\ref{Equation-F function}) and states $\Psi$ satisfying $\chi\left(\mathcal{N}\leq C\right)\Psi=\Psi$ and $\braket{\Psi|\mathbb{H}_K|\Psi}\leq e^\mathrm{Pek}+\delta e$, with $\delta e\geq 0$ and $K\geq \Lambda_*=\alpha^{2s}$,
\begin{align}
\label{Equation-Corollary_Coherent State Methode}
\frac{1}{T\|f\|_{\infty}}\left|\big\langle\Psi\big|\widehat{F}\,\big|\Psi \big\rangle-\iint F\left(\left|\varphi_{y,\xi}\right|^2\right)\mathrm{d}\mathbb{P}_{y}\left(\xi\right)\mathrm{d}y\right|\leq \sqrt{\delta e}+  \alpha^{-\frac{s}{2}}+  \alpha^{\frac{27}{2}s-2},
\end{align}
and furthermore 
\begin{align}
\label{Equation-Corollary_Coherent State Methode_Second Part}
\frac{1}{T}\left|\Big\langle \Psi\Big|W_g^{-1}\mathcal{N}W_g\Big|\Psi\Big\rangle-\iint \|\varphi_{y,\xi}-g\|^2\mathrm{d}\mathbb{P}_{y}\left(\xi\right)\mathrm{d}y\right|\leq \sqrt{\delta e}+  \alpha^{-\frac{s}{2}}+  \alpha^{\frac{27}{2}s-2}.
\end{align}
\end{cor}
\begin{proof}
Using the fact that we have $\big\langle\Psi\big|\widehat{F}\,\big|\Psi \big\rangle=\int \big\langle\Psi_y\big|\widehat{F}\,\big|\Psi_y \big\rangle\mathrm{d}y$ and $\big\langle\Psi\big|W_g^{-1}\mathcal{N}W_g\big|\Psi \big\rangle=\int \big\langle\Psi_y\big|W_g^{-1}\mathcal{N}W_g\big|\Psi_y \big\rangle\mathrm{d}y$, and applying Eq.~(\ref{Equation-Coherent State Methode}), respectively Eq.~(\ref{Equation-Coherent State Methode_Second Part}), immediately yields that the left hand sides of Eqs.~(\ref{Equation-Corollary_Coherent State Methode}) and~(\ref{Equation-Corollary_Coherent State Methode_Second Part}) are bounded by
\begin{align}
\label{Equation: Integral Particle Number Bound}
\! \frac{{N_*}}{\alpha^2}+\epsilon+\epsilon^{-1}\! \!\int\! \braket{\Psi_{y}|\mathcal{N}^y_{>{N_*}}|\Psi_{y}}\mathrm{d}y
\end{align}
 for any $\epsilon>0$. In order to bound $\int\! \braket{\Psi_{y}|\mathcal{N}^y_{>{N_*}}|\Psi_{y}}\mathrm{d}y$ from above, let us first apply Eq.~(\ref{Equation-General estimate for cut-off's}) together with Eq.~(\ref{Equation-Boundedness from below}), which provides the auxiliary estimate
\begin{align*}
& \int|\braket{\Psi_{y}|\mathbb{H}^{y}_{\Lambda_*,\ell_*}|\Psi_{y}}-\braket{\Psi_{y}|\mathbb{H}_K|\Psi_{y}}|\mathrm{d}y\lesssim  \alpha^{-s}\int \braket{\Psi_{y}|-\Delta_x+\mathcal{N}+1|\Psi_{y}}\mathrm{d}y\\
&\ \ \ \ \ \ \ \ \ \ \ \leq \alpha^{-s}\int \braket{\Psi_{y}|2\mathbb{H}_K+d+1|\Psi_{y}}\mathrm{d}y.
\end{align*} 
Note that the assumptions of Eq.~(\ref{Equation-General estimate for cut-off's}) are indeed satisfied, since $K\geq  \Lambda_*$ and $\mathrm{supp}\left(\Psi_{y}\right)\subset B_{L_*}(y)$. In combination with the IMS identity $\int \braket{\Psi_{y}|\mathbb{H}_K|\Psi_{y}}\mathrm{d}y=\braket{\Psi|\mathbb{H}_K|\Psi}+L_*^{-2}\|\nabla \chi\|^2$, where $\chi$ is the function from Eq.~(\ref{Definition-Psi_y}), this furthermore yields
\begin{align}
\label{Equation-IntegralEnergyEstimate}
& \left|\int\braket{\Psi_{y}|\mathbb{H}^{y}_{\Lambda_*,\ell_*}|\Psi_{y}}\mathrm{d}y-\braket{\Psi|\mathbb{H}_K|\Psi}\right|\lesssim \alpha^{-s}\left(\braket{\Psi|\mathbb{H}_K|\Psi}+d+1\right),
\end{align}
where we have used $L_*^{-2}=\alpha^{-s}$.  Furthermore $\braket{\Psi|\mathbb{H}_K|\Psi}\leq e^\mathrm{Pek}+\delta e$ by assumption, and consequently $|\int\braket{\Psi_{y}|\mathbb{H}^{y}_{\Lambda_*,\ell_*}|\Psi_{y}}\mathrm{d}y-\braket{\Psi|\mathbb{H}_K|\Psi}|\leq D \alpha^{-s}(\delta e+1)$ for a suitable $D$. Consequently
\begin{align}
\nonumber
\braket{\Psi|\mathbb{H}_K|\Psi} &\geq \int\braket{\Psi_{y}|\mathbb{H}^{y}_{\Lambda_*,\ell_*}|\Psi_{y}}\mathrm{d}y-D \alpha^{-s}(\delta e+1)\\
\label{Equation-Excitation number control}
&  \geq E_\alpha+\int\braket{\Psi_{y}|\mathcal{N}^y_{>{N_*}}|\Psi_{y}}\mathrm{d}y-D \alpha^{-s}(\delta e+1).
\end{align}
where we have used that $\mathbb{H}^{y}_{\Lambda_*,\ell_*}\geq E_\alpha+\mathcal{N}^y_{>{N_*}}$ in the second inequality. Using Eq.~(\ref{Equation-Excitation number control}) as well as the fact that $E_\alpha-e^\mathrm{Pek}\gtrsim -\alpha^{-\frac{1}{5}}\geq  -\alpha^{-s}$, see \cite{LT}, we obtain the upper bound
\begin{align}
\label{Equation-Useful upper bound on excitation number}
\int \braket{\Psi_{y}|\mathcal{N}^y_{>{N_*}}|\Psi_{y}}\mathrm{d}y\lesssim \braket{ \Psi|\mathbb{H}_K|\Psi }- e^\mathrm{Pek}+\alpha^{-s}(\delta e+1)\lesssim \delta e+\alpha^{-s}.
\end{align}
Choosing $\epsilon:=\sqrt{\delta e+\alpha^{-s}}$ in Eq.~(\ref{Equation: Integral Particle Number Bound}) therefore concludes the proof together with the observation that $\frac{N_*}{ \alpha^2}\lesssim \alpha^{\frac{27}{2}s-2}$.
\end{proof}

In the following Lemma \ref{Lemma-Support property of the measure} we are investigating the support properties of the lower symbol $\mathbb{P}_y$. In particular we derive bounds on the associated moments and verify that $\varphi_{y,\xi}$ is typically close to the manifold of minimizers $\{\varphi_x^\mathrm{Pek}:x\in \mathbb{R}^3\}$. 

\begin{lem}
\label{Lemma-Support property of the measure}
Given constants $m\in \mathbb{N}$ and $C>0$, there exists a $T>0$, such that $\iint |\xi|^{2m}\mathrm{d}\mathbb{P}_{y}\left(\xi\right)\mathrm{d}y\leq T$ for all $\Psi$ satisfying $\chi\left(\mathcal{N}\leq C\right)\Psi=\Psi$, and furthermore we have for all $K\geq \Lambda_*$, where $\Lambda_*$ is as in the definition of $\Pi^{y}$ in Eq.~(\ref{Equation-Definition of weak projection}),
\begin{align}
\label{Equation-Support property of the measure}
&\frac{1}{T}\iint \inf_{x\in \mathbb{R}^3}\|\varphi_{y,\xi}-\varphi^\mathrm{Pek}_x\|^2\mathrm{d}\mathbb{P}_{y}\left(\xi\right) \mathrm{d}y\leq \braket{\Psi|\mathbb{H}_K|\Psi}-e^\mathrm{Pek}+\alpha^{-s}+\alpha^{\frac{27}{2}s-2}.
\end{align}
\end{lem}
\begin{proof}
For $m\in \mathbb{N}$, let us define the function $G(\rho):=\left(\int \mathrm{d}\rho(x)\right)^m=\int\dots \int  \mathrm{d}\rho(x_1)\dots \mathrm{d}\rho(x_m)$, which is clearly of the form given in Eq.~(\ref{Equation-F function}). Consequently by Lemma \ref{Lemma-Coherent State Methode}
\begin{align*}
&\int\! |\xi|^{2m}\mathrm{d}\mathbb{P}_{y}\left(\xi\right)=\! \int\!  G\left(|\varphi_{y,\xi}|^{2}\right)\mathrm{d}\mathbb{P}_{y}\left(\xi\right)\lesssim \big\langle\Psi_{y}\big|\widehat{G}\,\big|\Psi_{y} \big\rangle+\left(\frac{{N_*}}{\alpha^2}+1\right)\|\Psi_y\|^2+\braket{\Psi_{y}|\mathcal{N}^y_{>{N_*}}|\Psi_{y}}\\
&=\big\langle \Psi_{y}\big| \mathcal{N}^{\, 2m}\big|\Psi_{y}\big\rangle+\left(\frac{{N_*}}{\alpha^2}+1\right)\|\Psi_y\|^2+\braket{\Psi_{y}|\mathcal{N}^y_{>{N_*}}|\Psi_{y}}\leq \left(C^{2m}+\frac{{N_*}}{\alpha^2}+1+C\right)\|\Psi_y\|^2,
\end{align*}
which concludes the proof of the first part, since $N_*\lesssim \alpha^2$ and $\int \|\Psi_y\|^2\mathrm{d}y=\|\Psi\|^2=1$.

Regarding the proof of Eq.~(\ref{Equation-Support property of the measure}),  we have the simple bound
\begin{align}
\nonumber
\mathbb{H}^{y}_{\Lambda_*,\ell_*}&\!=\! -\!\Delta_x\!-\!a\left(\Pi_{y} w_x\right)\!-\!a^\dagger\left(\Pi_{y} w_x\right)\!+\!\mathcal{N}\!\geq\! -\!\Delta_x\!-\!a\left(\Pi_{y} w_x\right)\!-\!a^\dagger\left(\Pi_{y} w_x\right)\!+\!\sum_{n=1}^{N_*} a_{y,n}^\dagger a_{y,n}\\
\label{Equation-Anti Wick order}
&=-\Delta_x-a\left(\Pi_{y} w_{x}\right)-a^\dagger\left(\Pi_{y} w_{x}\right)+\sum_{n=1}^{N_*} a_{y,n} a_{y,n}^\dagger-\frac{{N_*}}{\alpha^2}.
\end{align}
Since all terms in Eq.~(\ref{Equation-Anti Wick order}) are represented in anti-Wick ordering, we can follow \cite{LT} and express, similar as in the proof of Lemma \ref{Lemma-Coherent State Methode}, their expectation value as
\begin{align}
\nonumber
&\Big \langle\! \Psi_{y}\! \Big|\!\! -\! \!\Delta_x\! -\! a\left(\Pi_{y} w_{x}\right)\! -\! a^\dagger\left(\Pi_{y} w_{x}\right)\! +\! \! \sum_{n=1}^{N_*} a_{y,n} a_{y,n}^\dagger\Big|\!  \Psi_{y}\! \Big \rangle\! \!=\!  \! \! \int\! \! \! \left( \! \braket{ \psi^\xi_{y}|\! -\! \! \Delta_x\! +\! V_{\varphi_{y,\xi}}|\psi^\xi_{y}}  \! +\! \|\varphi_{y,\xi}\|^2\! \right)\! \mathrm{d}\mathbb{P}_{y}\! \left(\xi\right)\\
\label{Equation-CoherentState_Energy}
&\  \ \  \ \  \ \ \geq \int \left(\inf \sigma\left(-\Delta_x+V_{\varphi_{y,\xi}}\right)+\|\varphi_{y,\xi}\|^2\right)\mathrm{d}\mathbb{P}_{y}\left(\xi\right)=\int \mathcal{F}^\mathrm{Pek}\! \left(\varphi_{y,\xi}\right)\mathrm{d}\mathbb{P}_{y}\left(\xi\right),
\end{align}
with $\psi^\xi_{y}:=\frac{\Theta_{y,\xi} \Psi_{y}}{\|\Theta_{y,\xi} \Psi_{y}\|}$ where $\Theta_{y,\xi}$ is defined below Eq.~(\ref{Equation-Definition of the y-dependent measure}), $\mathcal{F}^\mathrm{Pek}$ is the Pekar functional and $V_{\varphi}$ is defined in Eq.~(\ref{Equation-Potential}). Making use of Eq.~(\ref{Equation-CoercivityPekarFunctional}) 
we obtain together with Eqs.~(\ref{Equation-IntegralEnergyEstimate}),~\eqref{Equation-Anti Wick order} and~(\ref{Equation-CoherentState_Energy})
\begin{align*}
&\int \int \inf_{x\in \mathbb{R}^3}\|\varphi_{y,\xi}-\varphi^\mathrm{Pek}_x\|^2\mathrm{d}\mathbb{P}_{y}\left(\xi\right)\mathrm{d}y\lesssim \int \braket{ \Psi_{y}|\mathbb{H}^{y}_{\Lambda_*,\ell_*}| \Psi_{y}}\mathrm{d}y-e^\mathrm{Pek}+\frac{{N_*}}{\alpha^2}\\
&\ \ \lesssim \braket{ \Psi|\mathbb{H}_K| \Psi}-e^\mathrm{Pek}+\frac{{N_*}}{\alpha^2}+D\alpha^{-s}\left(\braket{\Psi|\mathbb{H}_K|\Psi}+d+1\right),
\end{align*}
for a suitable $D>0$. This concludes the proof, since we have ${N_*}\lesssim  \alpha^{\frac{27}{2}s}$.
\end{proof}

The bound in Eq.~(\ref{Equation-Support property of the measure}) suggests that  $\varphi_{y,\xi}$ is close to $\varphi^\mathrm{Pek}_{x^{y,\xi}}$ with a high probability, where $x^{y,\xi}$ is the minimizer of $x\mapsto \|\varphi_{y,\xi}-\varphi^\mathrm{Pek}_{x}\|$. Motivated by this observation we expect $\iint F\left(\left|\varphi_{y,\xi}\right|^2\right)\mathrm{d}\mathbb{P}_y\mathrm{d}y\approx \iint F\left(\left|\varphi^\mathrm{Pek}_{x^{y,\xi}}\right|^2\right)\mathrm{d}\mathbb{P}_y\mathrm{d}y$ for measures $\mathbb{P}_y$ for low energy states $\Psi$, and therefore it seems natural to  define the measure $\mu$ in Theorem \ref{Theorem-Coherent States} as $\int f\mathrm{d}\mu:=\iint f\left(x^{y,\xi}\right)\mathrm{d}\mathbb{P}_y\mathrm{d}y$, allowing us to identify $\iint F\left(\left|\varphi^\mathrm{Pek}_{x^{y,\xi}}\right|^2\right)\mathrm{d}\mathbb{P}_y\mathrm{d}y=\int F\left(\left|\varphi^\mathrm{Pek}_x\right|^2\right)\mathrm{d}\mu$. This expression is however ill-defined, since the infimum $\inf_{x\in \mathbb{R}^3}\|\varphi_{y,\xi}-\varphi^\mathrm{Pek}_{x}\|$ is not necessarily attained and it is not necessarily unique. In order to avoid these difficulties, we will slightly modify the definition of the measure $\mu$ in the proof of Lemma \ref{Lemma-Measurable Function}.

\begin{lem}
\label{Lemma-Measurable Function}
Given $m\in \mathbb{N},C>0$ and $g \in L^2\! \left(\mathbb{R}^3\right)$ we can find a constant $T>0$, such that for all states $\Psi$ satisfying $\chi\left(\mathcal{N}\leq C\right)\Psi=\Psi$ and $\braket{\Psi|\mathbb{H}_K|\Psi}\leq e^\mathrm{Pek}+\delta e$, with $\delta e\geq 0$ and $K\geq \Lambda_*$, there exists a probability measure $\mu$ on $\mathbb{R}^3$ with the property
\begin{align}
\label{Equation-Approximation of double integral}
\frac{1}{T\|f\|_{\infty}}\left|\iint \!  F\left(\left|\varphi_{y,\xi}\right|^2\right)\!  \mathrm{d}\mathbb{P}_{y}\!  \left(\xi\right)\mathrm{d}y-\int\!  \!  F\!  \left(\left|\varphi^\mathrm{Pek}_{x}\right|^2\right)\!  \mathrm{d}\mu\!  \left(x\right)\right|\leq  \sqrt{\delta e}+  \alpha^{-\frac{s}{2}}+  \alpha^{\frac{27}{4}s-1},
\end{align}
for all $F$ of the form~(\ref{Equation-F function}), and furthermore
\begin{align}
\label{Equation-Approximation of double integral_second part}
\frac{1}{T}\left|\iint  \|\varphi_{y,\xi}-g\|^2 \mathrm{d}\mathbb{P}_{y}\!  \left(\xi\right)\mathrm{d}y-\int\!  \!  \|\varphi^\mathrm{Pek}_x-g\|^2  \mathrm{d}\mu\!  \left(x\right)\right|\leq \sqrt{\delta e}+  \alpha^{-\frac{s}{2}}+  \alpha^{\frac{27}{4}s-1}.
\end{align}
\end{lem}
\begin{proof}
For $\epsilon>0$, let $\bigcup_{n=1}^\infty A_{\epsilon,n}=\mathbb{C}^{N_*}$ be a partition of $\mathbb{C}^{N_*}$ consisting of non-empty measurable sets $A_{\epsilon,n}$ having a diameter bounded by $d(A_{\epsilon,n})\leq \epsilon$. Furthermore choose $\xi_{\epsilon,n}\in A_{\epsilon,n}$ and $x_{\epsilon,n}\in \mathbb{R}^3$ satisfying $\|\varphi_{0,\xi_{\epsilon,n}}-\varphi^\mathrm{Pek}_{x_{\epsilon,n}}\|\leq \inf_{x\in \mathbb{R}^3}\|\varphi_{0,\xi_{\epsilon,n}}-\varphi^\mathrm{Pek}_{x}\|+\epsilon$. 
Then
\begin{align}
\nonumber
&\|\varphi_{y,\xi}\! -\! \varphi^\mathrm{Pek}_{y+x_{\epsilon,n}}\| = \|\varphi_{0,\xi}\! -\! \varphi^\mathrm{Pek}_{x_{\epsilon,n}}\| \leq  \|\varphi_{0,\xi_{\epsilon,n}}\! -\! \varphi^\mathrm{Pek}_{x_{\epsilon,n}}\|\! +\! \|\varphi_{0,\xi}\! -\! \varphi_{0,\xi_{\epsilon,n}}\| \leq  \|\varphi_{0,\xi_{\epsilon,n}}\! -\! \varphi^\mathrm{Pek}_{x_{\epsilon,n}}\|\! +\! \epsilon\\
\label{Equation-Infimum estimate}
&\ \ \ \ \ \ \ \leq \inf_{x\in \mathbb{R}^3}\! \|\varphi_{0,\xi_{\epsilon,n}}\! -\! \varphi^\mathrm{Pek}_{x}\|\! +\! 2\epsilon\leq \inf_{x\in \mathbb{R}^3}\! \|\varphi_{0,\xi}\! -\! \varphi^\mathrm{Pek}_{x}\|\! +\! 3\epsilon=\inf_{x\in \mathbb{R}^3}\! \|\varphi_{y,\xi}\! -\! \varphi^\mathrm{Pek}_{x}\|\! +\! 3\epsilon.
\end{align}
Let us now define the probability measure $\mu$ on $\mathbb{R}^3$ by specifying its action on functions $f\in C\left(\mathbb{R}^3\right)$ as
\begin{align*}
\int f\mathrm{d}\mu:=\sum_{n=1}^\infty \int \!  f\left(y+x_{\epsilon,n}\right)\! \mathbb{P}_{y}\! \left(A_{\epsilon,n}\right) \mathrm{d}y=\sum_{n=1}^\infty \int \!  \int_{A_{\epsilon,n}}\!  f\left(y+x_{\epsilon,n}\right) \mathrm{d}\mathbb{P}_{y} \mathrm{d}y.
\end{align*}
 Since $\int \!  F\left(\left|\varphi_{y,\xi}\right|^2\right)\!  \mathrm{d}\mathbb{P}_{y}\!  \left(\xi\right)=\sum_{n=1}^\infty \int_{A_{\epsilon,n}} \!  F\left(\left|\varphi_{y,\xi}\right|^2\right)\!  \mathrm{d}\mathbb{P}_{y}\!  \left(\xi\right)$, we can estimate the left hand side of Eq.~(\ref{Equation-Approximation of double integral}) with the aid of the triangle inequality by
\begin{align}
\label{Equation-Polished version of the double integral}
\sum_{n=1}^\infty \int \!  \int_{A_{\epsilon,n}}\!  \left|F\left(\left|\varphi_{y,\xi}\right|^2\right)-F\left(\left|\varphi^\mathrm{Pek}_{y+x_{\epsilon,n}}\right|^2\right)\right| \mathrm{d}\mathbb{P}_{y}\! \left(\xi\right) \mathrm{d}y.
\end{align}
From the concrete form of the function $F$ given in Eq.~(\ref{Equation-F function}), as well as the facts that $\|\varphi^\mathrm{Pek}_{y+x_{\epsilon,n}}\|=\|\varphi^\mathrm{Pek}_{0}\|$ is finite and $\|\varphi_{y,\xi}\|=|\xi|$, one readily concludes that
  $$\left|F\left(\left|\varphi_{y,\xi}\right|^2\right)-F\left(\left|\varphi^\mathrm{Pek}_{y+x_{\epsilon,n}}\right|^2\right)\right|\lesssim \|f\|_\infty\left\|\varphi_{y,\xi}-\varphi^\mathrm{Pek}_{y+x_{\epsilon,n}}\right\|\left(1+|\xi|\right)^{2m-1}.$$
   Using Eq.~(\ref{Equation-Infimum estimate}) we further obtain for any $\kappa>0$ and $\xi\in A_{\epsilon,n}$
\begin{align*}
&\left\|\varphi_{y,\xi}-\varphi^\mathrm{Pek}_{y+x_{\epsilon,n}}\right\|\left(1+|\xi|\right)^{2m-1}\leq \left(\inf_{x\in \mathbb{R}^3}\! \|\varphi_{y,\xi}\! -\! \varphi^\mathrm{Pek}_{x}\|+3\epsilon\right)\left(1+|\xi|\right)^{2m-1}\\
&\ \ \ \ \ \ \ \ \ \ \ \ \ \ \leq \kappa^{-1}\inf_{x\in \mathbb{R}^3}\! \|\varphi_{y,\xi}\! -\! \varphi^\mathrm{Pek}_{x}\|^2+\frac \kappa 4\left(1+|\xi|\right)^{4m-2}+3\epsilon\left(1+|\xi|\right)^{2m-1},
\end{align*}
and therefore the expression in Eq.~(\ref{Equation-Polished version of the double integral}) can be bounded from above by
\begin{align*}
&\|f\|_\infty\bigg(\kappa^{-1}\iint \inf_{x\in \mathbb{R}^3}\! \|\varphi_{y,\xi}\! -\! \varphi^\mathrm{Pek}_{x}\|^2 \mathrm{d}\mathbb{P}_{y}\! \left(\xi\right) \mathrm{d}y + \frac \kappa 4\! \iint \left(1+|\xi|\right)^{4m-2}\mathrm{d}\mathbb{P}_{y}\! \left(\xi\right) \mathrm{d}y\\
& \ \ \ \ \ \ \ \ \ + 3\epsilon\! \iint \left(1+|\xi|\right)^{2m-1}\mathrm{d}\mathbb{P}_{y}\! \left(\xi\right) \mathrm{d}y\bigg).
\end{align*}
By Lemma \ref{Lemma-Support property of the measure} this concludes the proof of \eqref{Equation-Approximation of double integral}  with $\epsilon:=\kappa:=\sqrt{\delta e+\alpha^{-s}+\alpha^{\frac{27}{2}s-2}}$. Eq.~(\ref{Equation-Approximation of double integral_second part}) can be proven analogously, using the estimate
\begin{align*}
\left| \|\varphi_{y,\xi}-g\|^2-\|\varphi^\mathrm{Pek}_{y+x_{\epsilon,n}}-g\|^2\right| \lesssim \left\|\varphi_{y,\xi}-\varphi^\mathrm{Pek}_{y+x_{\epsilon,n}}\right\|\left(1+|\xi|\right)
\end{align*}
for $\xi \in A_{\epsilon,n}$. 
\end{proof}

Combining Eq.~(\ref{Equation-Corollary_Coherent State Methode}), respectively Eq.~(\ref{Equation-Corollary_Coherent State Methode_Second Part}), with Eq.~(\ref{Equation-Approximation of double integral}), respectively Eq.~(\ref{Equation-Approximation of double integral_second part}), immediately yields that the left hand side of Eq.~(\ref{Equation-Coherent state formula for F}), respectively Eq.~(\ref{Equation-Translated F version}), is of the order $\sqrt{\delta e}+  \alpha^{-\frac{s}{2}}+  \alpha^{\frac{27}{4}s-1}$. Optimizing in the parameter $0<s\leq \frac{4}{27}$ concludes the proof of Theorem \ref{Theorem-Coherent States} with the concrete choice $s:=\frac{4}{29}$.

\appendix

\section{Properties of the Pekar Minimizer}
\label{Properties of the Pekar Minimizer}
In the following section we derive certain useful properties concerning the minimizer $\varphi^\mathrm{Pek}$ of the Pekar functional $\mathcal{F}^\mathrm{Pek}$ in \eqref{def:FP}. We start with Lemma \ref{Lemma-Cut-off residuum estimate}, where we quantify the error of applying the cut-off $\Pi $ to a minimizer, where $\Pi$ is the projection defined in Eq.~(\ref{Equation-Definition of strong projection}) for a given parameter $0<\sigma<\frac{1}{4}$. The subsequent Lemmas \ref{Lemma-Boundedness of the density} and \ref{Lemma-Results on the Pekar minimizer} then concern the concentration of the density $\left|\varphi^\mathrm{Pek}\right|^2$ around the origin.

\begin{lem}
\label{Lemma-Cut-off residuum estimate}
For all $r>0$ we have the estimates $\sup_{|x|\leq r}\left\|\left(1-\Pi \right)\varphi^\mathrm{Pek}_{x}\right\|\lesssim \alpha^{-\frac{6}{5}(1+\sigma)}$. Moreover,  $\left\|\left(1-\Pi \right)\partial_{x_n} \varphi^\mathrm{Pek}\right\|\lesssim \alpha^{-\frac{2}{5}(1+\sigma)}$ for $n\in \{1,2,3\}$.
\end{lem}
\begin{proof}
We can write $\varphi^\mathrm{Pek}=4\sqrt{\pi}\left(-\Delta\right)^{-\frac{1}{2}}\left|\psi^\mathrm{Pek}\right|^2$ where $\psi^\mathrm{Pek}$ is the ground state of the operator $H_{V^\mathrm{Pek}}$. Consequently $\varphi^\mathrm{Pek}_x=4\sqrt{\pi}\left(f_x+g_x\right)$ with the definitions $\widehat{f_x}(k)=\mathds{1}_{B_{\Lambda}}(k)\frac{\widehat{\left|\psi^\mathrm{Pek}\right|^2}(k)}{|k|}e^{ik\cdot x}$ and $\widehat{g_x}(k)=\mathds{1}_{\mathbb{R}^3\setminus B_{\Lambda}}(k)\frac{\widehat{\left|\psi^\mathrm{Pek}\right|^2}(k)}{|k|}e^{ik\cdot x}$, where $\widehat{\cdot}$ denotes the Fourier transform. In the first step we are going to estimate $\|\left(1-\Pi \right)g_x\|=\|g_x\|$ by
\begin{align} 
\label{Equation-gEstimate}
\|g_x\|^2\! =\! \int_{|k|\geq \Lambda}\frac{\Big|\widehat{\left|\psi^\mathrm{Pek}\right|^2}(k)\Big|^2}{|k|^2}\mathrm{d}k\! \leq\!  \left\||k|^2\widehat{\left|\psi^\mathrm{Pek}\right|^2}(k)\right\|^2_\infty \int_{|k|\geq \Lambda}\frac{1}{|k|^6}\mathrm{d}k\lesssim \frac{1}{\Lambda^3}\! =\!  \alpha^{-\frac{12}{5}(1+\sigma)},
\end{align}
where we have used that $\psi^\mathrm{Pek}\in H^2\! \left(\mathbb{R}^3\right)$, see \cite{Li,MS}, and therefore $\left\||k|^2\widehat{\left|\psi^\mathrm{Pek}\right|^2}(k)\right\|_\infty<\infty$. In order to estimate the remaining part $\|\left(1-\Pi \right)f_x\|$, let us first compute
\begin{align*}
f_x(y)\! &=\! \frac{1}{\sqrt{(2\pi)^3}}\! \int_{|k|\leq \Lambda}\! \! \frac{\widehat{\left|\psi^\mathrm{Pek}\right|^2}(k)}{|k|}e^{ik\cdot (x-y)} \mathrm{d}k=\frac{1}{(2\pi)^3}\int_{|k|\leq \Lambda}\!  \! \frac{e^{ik\cdot (x-y)}}{|k|}\! \int_{\mathbb{R}^3}\!  \left|\psi^\mathrm{Pek}(z)\right|^2\! e^{ik\cdot z} \mathrm{d}z \mathrm{d}k\\
&=\frac{1}{(2\pi)^3}\int_{\mathbb{R}^3}\left|\psi^\mathrm{Pek}(z)\right|^2 \int_{|k|\leq \Lambda}\! \! \frac{e^{ik\cdot (x+z-y)}}{|k|}\mathrm{d}k\,  \mathrm{d}z=\frac{1}{\sqrt{4\pi}}\int_{\mathbb{R}^3}\left|\psi^\mathrm{Pek}(z)\right|^2 \Pi_\Lambda w_{x+z}(y)\,  \mathrm{d}z
\end{align*}
using the projection $\Pi_\Lambda$ from Definition \ref{Definition-Pi} and the function $w_{x}$ from Lemma \ref{Lemma-Norm Estimate}. Consequently we obtain by Lemma \ref{Lemma-Norm Estimate}
\begin{align*}
\|\left(1-\Pi \right)f_x\|&\leq \frac{1}{\sqrt{4\pi}}\int_{\mathbb{R}^3}\left|\psi^\mathrm{Pek}(z)\right|^2 \|\Pi_\Lambda w_{x+z}-\Pi w_{x+z}\|\,  \mathrm{d}z\\\\
&\lesssim \ell \sqrt{\Lambda}\int_{\mathbb{R}^3} |z|\left|\psi^\mathrm{Pek}(z)\right|^2\,  \mathrm{d}z+\ell \sqrt{\Lambda}|x|+\sqrt{\ell},
\end{align*}
where we have used $\left(1-\Pi \right)\Pi_\Lambda=\Pi_\Lambda-\Pi$ and $\int_{\mathbb{R}^3}\left|\psi^\mathrm{Pek}(z)\right|^2\,  \mathrm{d}z=1$. This concludes the proof of the first part, since the terms $\ell \sqrt{\Lambda}$ and $\sqrt{\ell}$ are all bounded by $\alpha^{-\frac{6}{5}(1+\sigma)}$, and the state $\psi^\mathrm{Pek}$ satisfies $\int_{\mathbb{R}^3} |z|^p \left|\psi^\mathrm{Pek}(z)\right|^2\,  \mathrm{d}z<\infty$ for any $p\geq 0$, see \cite{MS}. 

In order to verify the second part, we write again $\partial_{x_n} \varphi^\mathrm{Pek}=4\sqrt{\pi}\left(\partial_{x_n}f_0+\partial_{x_n}g_0\right)$. In analogy to Eq.~(\ref{Equation-gEstimate}) we have $\|\partial_{x_n}g_0\|^2\lesssim \frac{1}{\Lambda}=\alpha^{-\frac{4}{5}(1+\sigma)}$. Furthermore $\partial_{x_n}f_0(x)=-\frac{1}{\sqrt{4\pi}}\int_{\mathbb{R}^3}\partial_{z_n}\left(\left|\psi^\mathrm{Pek}(z)\right|^2\right) \Pi_\Lambda w_{z}(x)\,  \mathrm{d}z$, hence proceeding as above yields 
\begin{align*}
\|\left(1-\Pi \right)\partial_{x_n}f_0\| & \lesssim \ell \sqrt{\Lambda}\int_{\mathbb{R}^3} |z|\big|\partial_{z_n}\left(\left|\psi^\mathrm{Pek}(z)\right|^2\right)\! \big|\,  \mathrm{d}z \\ & \quad + \left( \ell \sqrt{\Lambda}|x|+\sqrt{\ell}\right) \int_{\mathbb{R}^3} \big|\partial_{z_n}\left(\left|\psi^\mathrm{Pek}(z)\right|^2\right)\! \big|\,  \mathrm{d}z.
\end{align*}
This concludes the proof, since
\begin{align*}
\int_{\mathbb{R}^3} |z|\Big|\partial_{x_n}\left(\left|\psi^\mathrm{Pek}(z)\right|^2\right)\! \Big|\,  \mathrm{d}z&=2\int_{\mathbb{R}^3} |z||\psi^\mathrm{Pek}(z)||\partial_{z_n}\psi^\mathrm{Pek}(z)|\,  \mathrm{d}z\\
&\leq \int_{\mathbb{R}^3} |z|^2|\psi^\mathrm{Pek}(z)|^2\,  \mathrm{d}z+\int_{\mathbb{R}^3} |\nabla \psi^\mathrm{Pek}(z)|^2\,  \mathrm{d}z<\infty
\end{align*}
and similarly with $|z|$ replaced by $1$.
\end{proof}

\begin{lem}
\label{Lemma-Boundedness of the density}
There exists a constant $C$ such that $\int\limits_{t\leq x_i\leq t+\epsilon}\left|\varphi^\mathrm{Pek}(x)\right|^2\, \mathrm{d}x\leq C\, \epsilon$ for all $t\in \mathbb{R}$, $\epsilon>0$ and $i\in \{1,2,3\}$.
\end{lem}
\begin{proof}
By the reflection symmetry of the Pekar minimizer, it is enough to prove the statement for $i=1$. For this purpose, let us define the function $D:\mathbb{R}\rightarrow \mathbb{R}$ as
\begin{align*}
D(t):=\int_{\mathbb{R}^2}\left|\varphi^\mathrm{Pek}(t,x_2,x_3)\right|^2\, \mathrm{d}x_2\mathrm{d}x_3
\end{align*}
In order to prove the Lemma, we are going to show that $D$ is a bounded function. Since $D(t)\underset{t\rightarrow \pm \infty}{\longrightarrow}0$, we have $\|D\|_\infty\leq \int |D'(t)|\mathrm{d}t$ and furthermore
\begin{align*}
\int |D'(t)|\mathrm{d}t&\leq \int \int_{\mathbb{R}^2}\left|\partial_t\left|\varphi^\mathrm{Pek}(t,x_2,x_3)\right|^2\right|\, \mathrm{d}x_2\mathrm{d}x_3\mathrm{d}t\leq \int_{\mathbb{R}^3}\left|\nabla_x \left|\varphi^\mathrm{Pek}\right|^2\right|\, \mathrm{d}x\\
&=2\int_{\mathbb{R}^3}\varphi^\mathrm{Pek}(x)\left|\nabla_x\varphi^\mathrm{Pek}\right|\, \mathrm{d}x\leq \|\varphi^\mathrm{Pek}\|^2+\|\nabla\varphi^\mathrm{Pek}\|^2<\infty,
\end{align*}
where we have used that $\varphi^\mathrm{Pek}\in H^1\! \left(\mathbb{R}^3\right)$.
\end{proof}

\begin{lem}
\label{Lemma-Results on the Pekar minimizer}
The Pekar minimizers $\varphi^\mathrm{Pek}_x$ satisfy $\left\|\varphi^\mathrm{Pek}_x-\varphi^\mathrm{Pek}\right\|^2\lesssim \sum_{i=1}^3 P^\epsilon_i\left(\left|\varphi^\mathrm{Pek}_x\right|^2\right)+\alpha^{-u}$, where $P^\epsilon_i$ is defined in Eq.~(\ref{Equation-Definition H_i}). 
\end{lem}
\begin{proof}
Since $\left\|\varphi^\mathrm{Pek}_x-\varphi^\mathrm{Pek}\right\|\leq \left\|\varphi^\mathrm{Pek}_x\right\|+\left\|\varphi^\mathrm{Pek}\right\|=2\left\|\varphi^\mathrm{Pek}\right\|$ and $\left\|\varphi^\mathrm{Pek}_x-\varphi^\mathrm{Pek}\right\|^2\leq |x|^2 \left\|\nabla \varphi^\mathrm{Pek}\right\|^2$, we have $\left\|\varphi^\mathrm{Pek}_x-\varphi^\mathrm{Pek}\right\|^2\lesssim \min\{|x|^2,1\}$. Therefore it is enough to show that we have $\min\{x_i^2,1\}\lesssim P^\epsilon_i\left(\left|\varphi^\mathrm{Pek}_x\right|^2\right)+\epsilon$. By the reflection symmetry of $\varphi^\mathrm{Pek}$, we can assume w.l.o.g. that $i=1$. We identify $\frac{1}{\left\|\varphi^\mathrm{Pek}\right\|^4}P^\epsilon_1\left(\left|\varphi^\mathrm{Pek}_x\right|^2\right)$
\begin{align*}
&\frac{1}{4}- \! \frac{1}{\left\|\varphi^\mathrm{Pek}\right\|^2}\underset{y_1\leq x_1+\epsilon}{\int}\left|\varphi^\mathrm{Pek}(y)\right|^2\mathrm{d}y\ \bigg(1- \! \frac{1}{\left\|\varphi^\mathrm{Pek}\right\|^2}\underset{y_1\leq  x_1-\epsilon}{\int}\left|\varphi^\mathrm{Pek}(y)\right|^2\mathrm{d}y\bigg)\\
&\ \ =\left(\frac{1}{2}\!-\!F(x_1)\right)^2+F(x_1)\big(F(x_1\!-\!\epsilon)\!-\!F(x_1)\big)\!+\!\big(F(x_1)\!-\!F(x_1\!+\!\epsilon)\big)\big(1\!-\!F(x_1-\epsilon)\big)\\
&\ \ \geq \left(\frac{1}{2}\!-\!F(x_1)\right)^2+\big(F(x_1\!-\!\epsilon)\!-\!F(x_1)\big)+\big(F(x_1)\!-\!F(x_1\!+\!\epsilon)\big)\geq \left(\frac{1}{2}\!-\!F(x_1)\right)^2-2C\, \epsilon
\end{align*}
with $F(t):=\frac{1}{\left\|\varphi^\mathrm{Pek}\right\|^2}\int_{y_1\leq t}\left|\varphi^\mathrm{Pek}(y)\right|^2\mathrm{d}y$, where $C$ is the constant from Lemma \ref{Lemma-Boundedness of the density}. Since $\varphi^\mathrm{Pek}$ is radially decreasing, see \cite{Li}, it is clear that $|\varphi^\mathrm{Pek}(x)|^2\geq c>0$ for all $x\in [-\delta,\delta]^3$ where $\delta,c>0$ are suitable constants. Assuming $x_1>0$ w.l.o.g. we conclude that $\left\|\varphi^\mathrm{Pek}\right\|^2\left(F(x_1)-\frac{1}{2}\right) \geq c\int_{0\leq y_1\leq x_1}\mathds{1}_{[-\delta,\delta]^3}(y)\, \mathrm{d}y =4c\delta^2\mathrm{min}\{x_1,\delta\}\gtrsim \min\{x_1,1\}$.
\end{proof}

\section{Properties of the Projection $\Pi $}
\label{Appendix-Estimates on Operator Norms}
In the following section we discuss properties of the Projections $\Pi $ defined in Eq.~(\ref{Equation-Definition of strong projection}) and $\Pi_K$ defined in Definition~\ref{Definition-Pi}. The first two results in Lemma \ref{Lemma-BasisElements} and Corollary \ref{Corollary-BasisElements} concern the space confinement of elements in the range of $\Pi $, to be precise we show that the associated potentials $V_\varphi$ defined in Eq.~(\ref{Equation-Potential}) are concentrated in a ball of radius $\alpha^q$ for a suitable $q>0$. While Lemma \ref{Lemma-Integral estimates} is an auxiliary result, we will show in the subsequent Lemmas \ref{Lemma-Auxillary Convergence of the trace} and \ref{Lemma-Convergence of the trace} that the operator $ J_{t,\epsilon} $ is an approximation of the Hessian $\mathrm{Hess}|_{\varphi^\mathrm{Pek}}\mathcal{F}^\mathrm{Pek}$, where $J_{t,\epsilon}$ is the operator defined in Eq.~(\ref{Equation-Pekar Functional}). Finally, we will show in Lemma \ref{Lemma-Uniform_singularity_control} that the functions $\Pi_K w_x$, which appear in the definition of $\mathbb{H}_K$ in Eq.~(\ref{Equation-Momentum Cut off Hamiltonian}), are confined in space around the origin. We will then use this result in order to quantify the energy cost of having the electron and the phonon field localized in different regions of space, see Corollary \ref{Corollary-Outside_Mass}.

The proof of the following auxiliary Lemma \ref{Lemma-BasisElements} is an easy analysis exercise and is left to the reader.
\begin{lem}
\label{Lemma-BasisElements}
There exists a constant $C>0$ such that for $f\in C^3\left(\mathbb{R}^3\right)$ and $K:=(k_1,k'_1)\times (k_2,k'_2)\times (k_3,k'_3)\subset \mathbb{R}^3$ with $k_i<k_i'< k_i+2$
\begin{align*} 
\left|\widehat{\left(\mathds{1}_K f\right)}(x)\right|\leq C\frac{\|f\|_{C^3\left(K\right)}}{\left(1+|x_1|\right)\left(1+|x_2|\right)\left(1+|x_3|\right)}
\end{align*}
for all $x=(x_1,x_2,x_3)\in \mathbb{R}^3$, where $\|f\|_{C^3\left(K\right)}:=\max_{|\alpha|\leq 3}\sup_{x\in K}|\partial^{\alpha} f(x)|$ and $\widehat{\cdot}$ denotes the Fourier transform.
\end{lem}

\begin{cor}
\label{Corollary-BasisElements}
There exists a constant $v>0$, such that for all $r>0$ and $\varphi\in \Pi  L^2\! \left(\mathbb{R}^3\right)$
\begin{align}
\label{Equation-Decay on subspace}
\left\|\mathds{1}_{\mathbb{R}^3 \setminus B_r(0)}V_\varphi\right\|\lesssim \frac{\alpha^v \|\varphi\|}{\sqrt{r}},
\end{align}
where $\Pi $ is defined in Eq.~(\ref{Equation-Definition of strong projection}) and $V_\varphi$ is defined in Eq.~(\ref{Equation-Potential}).
\end{cor}
\begin{proof}
Let $e_n$ be the basis from Definition \ref{Definition-Pi} corresponding to concrete choices of $\Lambda$ and $\ell$ defined above Eq.~(\ref{Equation-Definition of strong projection}). Given $\varphi=\sum_{n=1}^{N} \lambda_n e_n\in \Pi  L^2\! \left(\mathbb{R}^3\right)$, $\lambda_n\in  \mathbb{C}$, we have the rough estimate 
\begin{align*}
\left\|\mathds{1}_{\mathbb{R}^3 \setminus B_r(0)}V_\varphi\right\|\leq \sum_{n=1}^{N} |\lambda_n| \left\|\mathds{1}_{\mathbb{R}^3 \setminus B_r(0)}V_{e_n}\right\|\leq \sqrt{{N}}\|\varphi\| \sup_{n\in \{1,\dots , {N}\}}\left\|\mathds{1}_{\mathbb{R}^3 \setminus B_r(0)}V_{e_n}\right\|.
\end{align*}
Since ${N}\leq \alpha^p$ for a suitable constant $p$, it is enough to verify Eq.~(\ref{Equation-Decay on subspace}) for $\varphi=e_n$. Making use of $V_{e_n}=\widehat{\mathds{1}_{K_n} f}$ with $K_n:=\left(z^n_1-\ell,z^n_1+\ell\right)\times \left(z^n_2-\ell,z^n_2+\ell\right)\times \left(z^n_3-\ell,z^n_3+\ell\right)$ and $f(k)=\frac{-2}{\sqrt{(2\pi)^3 \int_{K_n}\frac{1}{|k|^2}\,\mathrm{d}k}}\frac{1}{|k|^2}$, and the fact that $(z_k^n+\ell)-(z_k^n-\ell)=2\ell\leq 2$, we obtain by Lemma \ref{Lemma-BasisElements}
\begin{align*}
\left\|\mathds{1}_{\mathbb{R}^3 \setminus B_r(0)}V_{e_n}\right\|^2\lesssim \alpha^{2p'}\int_{|x|> r}\frac{1}{(1+|x_1|)^2(1+|x_3|)^2(1+|x_3|)^2}\, \mathrm{d}x\lesssim \alpha^{2p'}\, \frac{1}{r},
\end{align*}
where we have used $K_n\subset \mathbb{R}^3 \setminus B_{2\ell}(0)$ and therefore $\|f\|_{C^3(K)}\lesssim \ell^{-\frac{3}{2}}\Lambda \left(\ell\right)^{-5}=\alpha^{p'}$ for a suitable $p'>0$.
\end{proof}

\begin{lem}
\label{Lemma-Integral estimates}
For $\psi\in L^2\! \left(\mathbb{R}^3\right)$ 
and $T>0$, 
\begin{align}
\label{Equation-Small momentum part of the trace}
\int\int_{|k'|\leq T}\frac{|\widehat{\psi}(k-k')|^2}{(1+|k|^2)\, |k'|^2}\, \mathrm{d}k'\mathrm{d}k&\lesssim \|\psi\|^2 T,\\
\label{Equation-Large momentum part of the trace}
\int\int_{|k'|> T}\frac{|\widehat{\psi}(k-k')|^2}{(1+|k|^2)\, |k'|^2}\, \mathrm{d}k'\mathrm{d}k&\lesssim \frac{\|\psi\|^2}{\sqrt{T}}.
\end{align}
Furthermore, interpreting $\psi$ as a multiplication operator we have 
\begin{align}
\label{Equation-First HS norm}
&\left\|\left(1-\Delta\right)^{-\frac{1}{2}}\psi \left(-\Delta\right)^{-\frac{1}{2}}\right\|_{\mathrm{HS}}\lesssim \|\psi\|,\\
\label{Equation-Second HS norm}
&\left\|\left(1-\Delta\right)^{-\frac{1}{2}} \left(-\Delta\right)^{-\frac{1}{2}}\psi\right\|_{\mathrm{HS}}=\sqrt{2}\pi\|\psi\|.
\end{align}
\end{lem}
\begin{proof}
Eq.~(\ref{Equation-Small momentum part of the trace}) and Eq.~(\ref{Equation-Large momentum part of the trace}) immediately follow from the estimates
\begin{align*}
\int\int_{|k'|\leq T}\frac{|\widehat{\psi}(k-k')|^2}{(1+|k|^2)\, |k'|^2}\, \mathrm{d}k'\mathrm{d}k&\leq \int\int_{|k'|\leq T}\frac{|\widehat{\psi}(k-k')|^2}{|k'|^2}\, \mathrm{d}k'\mathrm{d}k=\|\psi\|^2 4\pi T,\\
\int\! \int_{|k'|> T}\frac{|\widehat{\psi}(k-k')|^2}{(1+|k|^2)\, |k'|^2}\, \mathrm{d}k'\mathrm{d}k&\leq \frac{1}{2}\int\! \int_{|k'|> T}\! \left(\frac{1}{\sqrt{T}(1+|k|^2)^2}\! +\! \frac{\sqrt{T}}{|k'|^4}\right)\! |\widehat{\psi}(k-k')|^2\, \mathrm{d}k'\mathrm{d}k\\
&\leq \frac{1}{2}\left(\int\frac{1}{(1+|k|^2)^2}\, \mathrm{d}k+4\pi\right)\frac{\|\psi\|^2}{\sqrt{T}}.
\end{align*}
By making use of the fact that the integral kernel of $\left(1-\Delta\right)^{-\frac{1}{2}}\psi \left(-\Delta\right)^{-\frac{1}{2}}$ in Fourier space is given as $\frac{\widehat{\psi}(k-k')}{\sqrt{1+|k|^2}|k'|}$,  Eq.~(\ref{Equation-First HS norm}) immediately follows from Eq.~(\ref{Equation-Large momentum part of the trace}) and Eq.~(\ref{Equation-Small momentum part of the trace}) with the concrete choice $T=1$. Finally Eq.~(\ref{Equation-Second HS norm}) follows from the fact that the corresponding integral kernel is given by $\frac{\widehat{\psi}(k-k')}{\sqrt{1+|k|^2}|k|}$ and the identity $\int \int \frac{|\widehat{\psi}(k-k')|^2}{|k|^2(1+|k|^2)}\, \mathrm{d}k'\mathrm{d}k=\int \frac{1}{|k|^2(1+|k|^2)}\, \mathrm{d}k\, \|\psi\|^2=2\pi^2 \|\psi\|^2$.
\end{proof}

\begin{lem}
\label{Lemma-Auxillary Convergence of the trace}
We have $\mathrm{Tr}\left[\left(1-\Pi\right)L^\mathrm{Pek}_x\left(1-\Pi\right)\right]\lesssim  \alpha^{-\frac{2}{5}}$ for $|x|\lesssim 1$, where $L^\mathrm{Pek}_x$ is the operator defined above Eq.~(\ref{Equation-Definition of J}).
\end{lem}
\begin{proof}
With the definition $\psi^\mathrm{Pek}_x(y):=\psi^\mathrm{Pek}(y-x)$, we can express the operator $L^\mathrm{Pek}_x$ as $L^\mathrm{Pek}_x=2\left|\left(1-\Delta\right)^{-\frac{1}{2}}\psi^\mathrm{Pek}_x\left(-\Delta\right)^{-\frac{1}{2}}\right|^2$. Since the integral kernel of $\left(1-\Delta\right)^{-\frac{1}{2}}\psi^\mathrm{Pek}_x\left(-\Delta\right)^{-\frac{1}{2}}$ is given by $\frac{\widehat{\psi}^\mathrm{Pek}_x(k-k')}{\sqrt{1+|k|^2} |k'|}$ in Fourier space and since the one of $\Pi$ reads $\sum_{n=1}^{N}\frac{\mathds{1}_{C_{z^n}}(k)\mathds{1}_{C_{z^n}}(k')}{\int_{C_{z^n}}\frac{1}{|q|^2}\mathrm{d}q\, |k|\, |k'|}$, where $C_{z^n}$ is as in Definition \ref{Definition-Pi}, we can further express the integral kernel of the operator $\left(1-\Delta\right)^{-\frac{1}{2}}\psi^\mathrm{Pek}_x\left(-\Delta\right)^{-\frac{1}{2}}(1-\Pi )$ as
\begin{align}
\label{Equation-Integral kernel}
\sum_{n=1}^{N}\frac{\int_{C_{z^n}} \frac{\widehat{\psi}^\mathrm{Pek}_x(k-k')-\widehat{\psi}^\mathrm{Pek}_x(k-q)}{\sqrt{1+|k|^2} |k'|}\frac{1}{|q|^2}\mathrm{d}q}{\int_{C_{z^n}}\frac{1}{|q'|^2}\mathrm{d}q} \mathds{1}_{C_{z^n}}(k')+\frac{\widehat{\psi}^\mathrm{Pek}_x(k-k')}{\sqrt{1+|k|^2} |k'|}\mathds{1}_{\mathbb{R}^3\setminus \left(\bigcup_n C_{z^n}\right)}(k').
\end{align}
In the following we need to show that the $L^2\! \left(\mathbb{R}^3\times \mathbb{R}^3\right)$ norm of the expression in Eq.~(\ref{Equation-Integral kernel}) is of order $\alpha^{-\frac{1}{5}}$. As in the proof of Lemma \ref{Lemma-Norm Estimate}, we will use $\mathbb{R}^3\setminus \left(\bigcup_n C_{z^n}\right)\subset B_{2\ell}\cup \left(\mathbb{R}^3\setminus B_{\Lambda-4\ell}\right)$, where $\Lambda$ and $\ell$ are defined above Eq.~(\ref{Equation-Definition of strong projection}). Applying Eq.~(\ref{Equation-Small momentum part of the trace}) with $T=2\ell$ and Eq.~(\ref{Equation-Large momentum part of the trace}) with $T=\Lambda-4\ell$ yields
\begin{align*}
\int\int_{\mathbb{R}^3\setminus \left(\bigcup_n C_{z^n}\right)}\frac{|\widehat{\psi}^\mathrm{Pek}_x(k-k')|^2}{(1+|k|^2) |k'|^2}\, \mathrm{d}k'\mathrm{d}k\lesssim 2\ell+\frac{1}{\sqrt{\Lambda-4\ell}}\lesssim \alpha^{-\frac{2}{5}}.
\end{align*}
In order to estimate the $L^2$ norm of $f(k,k'):=\sum_{n=1}^{N}\frac{\int_{C_{z^n}} \frac{\widehat{\psi}^\mathrm{Pek}_x(k-k')-\widehat{\psi}^\mathrm{Pek}_x(k-q)}{\sqrt{1+|k|^2} |k'|}\frac{1}{|q|^2}\mathrm{d}q}{\int_{C_{z^n}}\frac{1}{|q|^2}\mathrm{d}q} \mathds{1}_{C_{z^n}}(k')$, let us define $\psi_{x,s,\eta}(y):=\frac{\eta}{|\eta|}\cdot y e^{is\eta \cdot y}\psi^\mathrm{Pek}_x(y)$ for $s\in \mathbb{R},\eta\in \mathbb{R}^3$ and $\xi:=q-k'$, and  compute
\begin{align*}
\widehat{\psi}^\mathrm{Pek}_x(k\! -\! k')\! -\! \widehat{\psi}^\mathrm{Pek}_x(k\! -\! q)= \! \int_0^1 \xi\cdot \nabla \widehat{\psi}^\mathrm{Pek}_x(k\! -\! k'+s\xi)\mathrm{d}s=|\xi|\! \int_0^1  \widehat{\psi}_{x,s,\xi}(k\! -\! k')\mathrm{d}s.
\end{align*}
Making use of the inequality $\frac{\frac{1}{|q|^2}}{\int_{C_{z^n}}\frac{1}{|q'|^2}\, \mathrm{d}q'}\lesssim \ell^{-3}$ for $q\in C_{z^n}$ and the fact that $\xi=q-k'\in K:=(-2\ell,2\ell)^3$ for all $k',q\in C_{z^n}$, yields
\begin{align*}
&\left|f(k,k')\right|^2\lesssim \sum_{n=1}^{N}\mathds{1}_{C_{z^n}}(k')\ell^{-4}\left|\int_{K}\int_{0}^1 \frac{\big|\widehat{\psi}_{x,s,\xi}(k - k')\big|}{\sqrt{1+|k|^2} |k'|}\mathrm{d}s\mathrm{d}\xi\right|^2\\
&\ \ \leq \sum_{n=1}^{N}\mathds{1}_{C_{z^n}}(k')8\ell^{-1}\int_{K}\int_{0}^1 \frac{\big|\widehat{\psi}_{x,s,\xi}(k - k')\big|^2}{(1+|k|^2) |k'|^2}\mathrm{d}s \mathrm{d}\xi\leq 8\ell^{-1}\int_{K}\int_{0}^1 \frac{\big|\widehat{\psi}_{x,s,\xi}(k - k')\big|^2}{(1+|k|^2) |k'|^2}\mathrm{d}s \mathrm{d}\xi,
\end{align*}
where we have applied the Cauchy--Schwarz inequality. An application of Lemma \ref{Lemma-Integral estimates} with $T=1$ then yields
\begin{align*}
\iint \left|f(k,k')\right|^2\, \mathrm{d}k'\mathrm{d}k\lesssim \ell^{-1}\int_{K}\int_{0}^1 \|\psi_{x,s,\xi}\|^2\mathrm{d}s\mathrm{d}\xi\leq C\ell^2\lesssim \alpha^{-8},
\end{align*}
where we used that $\|\psi_{x,s,\eta}\|\leq C$ for all $|x|\lesssim 1$ and a suitable constant $C<\infty$.
\end{proof}

\begin{lem}
\label{Lemma-Convergence of the trace}
Recall the operator $H^\mathrm{Pek}$ from Eq.~(\ref{Equation-Definition_Hessian}). Then there exists a constant $c>0$ such that $J_{t,\epsilon}\geq c\, \pi$ for $\epsilon$ small enough and $\alpha$ large enough. Furthermore
\begin{align}
\label{Equation-Convergence of the trace}
\left|\mathrm{Tr}_{\Pi L^2\! (\mathbb{R}^3)}\! \!\left[1 -\sqrt{  J_{t,\epsilon} }\, \right]-\mathrm{Tr}\left[1-\sqrt{H^\mathrm{Pek}}\, \right]\right|\lesssim \epsilon+\alpha^{-\frac{1}{5}}
\end{align}
for $|t|<\epsilon$, $\epsilon$ small enough and $\alpha$ large enough. 
\end{lem}
\begin{proof}
Recall the definition of $\pi$ and $J_{t,\epsilon}$ in, respectively below, Eq.~(\ref{Equation-Definition of J}) for $|t|<\epsilon<\delta_*$, where $\delta_*$ is defined before Definition~\ref{Definition-tau}. In the following we are going to verify that  $\|(1+\epsilon)\pi  \left(K^\mathrm{Pek}_{x_t}+\epsilon L^\mathrm{Pek}_{x_t}\right) \pi\|_{\mathrm{op}}\leq 1-c$ for a suitable constant $c>0$, small $\epsilon$ and $|t|<\epsilon$, which immediately implies $J_{t,\epsilon}\geq c\, \pi$. Let $\pi_{x}$ be the orthogonal projection onto $\{\partial_{x_1} \varphi^\mathrm{Pek}_x,\partial_{x_2} \varphi^\mathrm{Pek}_x,\partial_{x_3} \varphi^\mathrm{Pek}_x\}^\perp$ and let $\varphi_n$ be defined in Eq.~(\ref{Equation-Orthonormal Basis}). Then we estimate
\begin{align}
\label{Equation-pi differences}
\mathrm{Tr}\left[\left|\pi_0-\pi_{ x}\right|\right]&\leq 2\sum_{n=1}^3 \left\|\varphi_n-\frac{\partial_{x_n} \varphi^\mathrm{Pek}_x}{\|\partial_{x_n} \varphi^\mathrm{Pek}_x\|}\right\|\\
\nonumber 
&\leq 2\sum_{n=1}^3 \left\|\frac{\partial_{x_n} \varphi^\mathrm{Pek}}{\|\partial_{x_n} \varphi^\mathrm{Pek}\|}-\frac{\partial_{x_n} \varphi^\mathrm{Pek}_x}{\|\partial_{x_n} \varphi^\mathrm{Pek}_x\|}\right\|+2\sum_{n=1}^3 \left\|\varphi_n-\frac{\partial_{x_n} \varphi^\mathrm{Pek}}{\|\partial_{x_n} \varphi^\mathrm{Pek}\|}\right\|\lesssim |x|+\alpha^{-\frac{2}{5}},
\end{align}
where we have used Lemma \ref{Lemma-Cut-off residuum estimate} in order to obtain $\|\partial_{x_n} \varphi^\mathrm{Pek}-\Pi\partial_{x_n} \varphi^\mathrm{Pek}\|\lesssim \alpha^{-\frac{2}{5}}$ and the fact that $ \varphi^\mathrm{Pek}\in H^2\! \left(\mathbb{R}^3\right)$, which yields $\|\partial_{x_n} \varphi^\mathrm{Pek}_x-\partial_{x_n} \varphi^\mathrm{Pek}\|\leq |x|\left\|\nabla \partial_{x_n} \varphi^\mathrm{Pek}\right\|\lesssim |x|$. Hence $\mathrm{Tr}\left[\left|\pi_0-\pi_{\pm x_t}\right|\right]\lesssim |t|+\alpha^{-\frac{2}{5}}$ for $t$ small enough. It is a straightforward consequence of \eqref{Equation-CoercivityPekarFunctional}  that the operator norm of $\pi_0 K^\mathrm{Pek}\pi_0$ is bounded by $\|\pi_0 K^\mathrm{Pek}\pi_0\|_\mathrm{op}<1$ (see also \cite[Lemma~1.1]{MMS}). Therefore we obtain, using $\pi=\Pi \pi_0=\pi_0\Pi$,
\begin{align}
\nonumber
&\left\|(1\! +\! \epsilon)\pi\!  \left(\! K^\mathrm{Pek}_{x_t}\! +\! \epsilon  L^\mathrm{Pek}_{x_t}\! \right)\!  \pi\right\|_\mathrm{op}\! \leq \left\|(1\! +\! \epsilon)\pi_0 \! \left(\! K^\mathrm{Pek}_{x_t}\! +\! \epsilon  L^\mathrm{Pek}_{x_t}\! \right)\!  \pi_0\right\|_\mathrm{op}\! = \left\|\pi_0 K^\mathrm{Pek}_{x_t}\pi_0\right\|_\mathrm{op}\! +O\left(\epsilon\right) \\\label{Equation: Strict Operator Bound}
&\ \ \ \ \ = \left\|\pi_{-x_t} K^\mathrm{Pek}\pi_{-x_t}\right\|_\mathrm{op}\! \! \! +O\left(\epsilon\right)
= \left\|\pi_0 K^\mathrm{Pek}\pi_0\right\|_\mathrm{op}\! \! \! +O\left(\epsilon\right) + O\big(\alpha^{-2/5}\big) \leq 1\! -\! c
\end{align}
for a suitable constant $c>0$, $\epsilon$ small enough, $|t|< \epsilon$ and $\alpha$ large enough. 

 In order to verify Eq.~(\ref{Equation-Convergence of the trace}), let $|t|< \epsilon$ and $\epsilon$ be small enough such that $J_{t,\epsilon}\geq 0$, and let us compute
\begin{align*}
\mathrm{Tr}_{\Pi L^2\! (\mathbb{R}^3)}\! \!\left[1 -\sqrt{  J_{t,\epsilon} }\, \right]=\mathrm{Tr}\left[1+\pi_0^\perp-\sqrt{1-(1+\epsilon)\pi \left(K^\mathrm{Pek}_{x_t}+\epsilon L^\mathrm{Pek}_{x_t}\right) \pi}\right],
\end{align*}
Furthermore we have the identity $\mathrm{Tr}\left[1-\sqrt{1-K^\mathrm{Pek}}\right]=\mathrm{Tr}\left[1+\pi_0^\perp-\sqrt{1- \pi_0 K^\mathrm{Pek}\pi_0 }\right]=\mathrm{Tr}\left[1-\sqrt{1- \pi_{x_t} K^\mathrm{Pek}_{x_t}\pi_{x_t} }\, \right]+\mathrm{Tr}\left[\pi_0^\perp\right]$. Using the definition of $K^\mathrm{Pek}$ in Eq.~(\ref{Equation-K}), we can express $\mathrm{Tr}_{\Pi L^2\! (\mathbb{R}^3)}\! \!\left[1 -\sqrt{  J_{t,\epsilon} }\, \right]-\mathrm{Tr}\left[1-\sqrt{H^\mathrm{Pek}}\, \right]$ as
\begin{align}
\label{Equation-Trace_Difference}
\mathrm{Tr}\left[1-\sqrt{1-(1+\epsilon)\pi \left(K^\mathrm{Pek}_{x_t}+\epsilon L^\mathrm{Pek}_{x_t}\right) \pi}\right]-\mathrm{Tr}\left[1-\sqrt{1- \pi_{x_t} K^\mathrm{Pek}_{x_t}\pi_{x_t} }\right].
\end{align}
In the following let $f$ be a smooth function with compact support satisfying $f(x)=1-\sqrt{1-x}$ for $0\leq x\leq 1-c$, where $c$ is as in Eq.~(\ref{Equation: Strict Operator Bound}), and let us define the operators $A:=(1\! +\! \epsilon)\pi \left(K^\mathrm{Pek}_{x_t}\! +\! \epsilon  L^\mathrm{Pek}_{x_t}\right) \pi$ and $B:=\pi_{x_t} K^\mathrm{Pek}_{x_t}\pi_{x_t}$. Using Eq.~(\ref{Equation-Trace_Difference}) and $\left\|(1+\epsilon)\pi \left(K^\mathrm{Pek}_{x_t}+\epsilon  L^\mathrm{Pek}_{x_t}\right) \pi\right\|_\mathrm{op}\leq 1-c$ for $t$ and $\epsilon$ small enough, we obtain
\begin{align}
\nonumber
&\left|\mathrm{Tr}_{\Pi L^2\! (\mathbb{R}^3)}\! \!\left[1 -\sqrt{  J_{t,\epsilon} }\, \right] -\mathrm{Tr}\left[ 1- \sqrt{H^\mathrm{Pek}}\, \right]\right|  =\left|\mathrm{Tr}\left[f(A)-f(B)\right]\right|\\
\label{Equation-Trace_with_f}
&\ \ \ \ \leq \left\|f(A)-f(B)\right\|_1\leq \frac{1}{\sqrt{2\pi}}\int_\mathbb{R} |t \widehat{f}(t)|\, \mathrm{d}t \left\|A-B\right\|_1,
\end{align}
where $\|\cdot \|_1$ is the trace norm and $\widehat{f}$ is the Fourier transformation of $f$. In order to estimate the right hand side of Eq.~(\ref{Equation-Trace_with_f}), we  write $A-B=T_1+\pi_0 T_2\pi_0+\pi T_3\pi$ with $T_1:=(\pi_0-\pi_{x_t}) K_{x_t}^\mathrm{Pek}\pi_0+\pi_{x_t} K_{x_t}^\mathrm{Pek}(\pi_0-\pi_{x_t})$, $T_2:= (\Pi-1) K^\mathrm{Pek}_{x_t}\Pi + K_{x_t}^\mathrm{Pek}(\Pi-1)$ and $T_3:=\epsilon \left(K^\mathrm{Pek}_{x_t}+(1+\epsilon)L^\mathrm{Pek}_{x_t}\right)$. Clearly we have the estimates $\|\pi T_3\pi\|_1\leq \|T_3\|_1\lesssim \epsilon$ and $\|T_1\|_1\lesssim \|\pi_0-\pi_{x_t}\|_1\lesssim t+\alpha^{-\frac{2}{5}}$ by Eq.~(\ref{Equation-pi differences}), using the fact that $K^\mathrm{Pek}_{x_t}$ is trace-class, which follows from $K^\mathrm{Pek}_{x_t}\lesssim  L^\mathrm{Pek}_{x_t}$ and the fact that $L^\mathrm{Pek}_{x_t}$ is trace-class, see Eq.~(\ref{Equation-First HS norm}) with $\psi:=\psi^\mathrm{Pek}$. 
 Using Lemma \ref{Lemma-Auxillary Convergence of the trace} together with a Cauchy--Schwarz estimate for the trace norm, we can bound the final contribution $\pi_0 T_2\pi_0$ by
\begin{align*}
\|\pi_0 T_2\pi_0\|_1\leq \|T_2\|_1\leq 2\mathrm{Tr}\left[\Pi K^\mathrm{Pek}_{x_t}\Pi \right]^{\frac{1}{2}} \mathrm{Tr}\left[\left(1-\Pi\right)K^\mathrm{Pek}_{x_t}\left(1-\Pi\right)\right]^{\frac{1}{2}}\lesssim \alpha^{-\frac{1}{5}}.
\end{align*}
\end{proof}

The following Lemma \ref{Lemma-Uniform_singularity_control} is an auxiliary result, which we will use to quantify the energy cost of having the electron and the phonon field localized in different regions of space, see Corollary \ref{Corollary-Outside_Mass}.

\begin{lem}
\label{Lemma-Uniform_singularity_control}
Let $w_0(y)=\pi^{-\frac{3}{2}}\frac{1}{|y|^2}$ and let $\Pi_K$ be the projection defined in Definition \ref{Definition-Pi}. Then there exist a constant $D$ such that 
\begin{align*}
\|\mathds{1}_{\mathbb{R}^3\setminus B_r(0)}\Pi_K w_0\|\leq \frac{D}{\sqrt{r}}
\end{align*}
for all $K,r>0$.
\end{lem}
\begin{proof}
The Fourier transform of $\Pi_K w_0$ is given by $\frac{\chi\left(|k|\leq K\right)}{\sqrt{2\pi^2}|k|}$. Defining the function $u$ via its Fourier transform as $\widehat{u}(k):=\frac{\chi^\epsilon\left(2\epsilon\leq |k|\leq K\right)}{\sqrt{2\pi^2}|k|}$, where $\epsilon>0$ and $\chi^\epsilon$ is defined in Eq.~(\ref{Equation-Epsilon cut-off}), we have
\begin{align*}
\left\|\Pi_K w_0-u\right\|^2\leq \frac{1}{2\pi^2}\int_{|k|\leq 3\epsilon}\frac{1}{|k|^2}\mathrm{d}k+\frac{1}{2\pi^2}\int_{K-\epsilon\leq |k|\leq K+\epsilon}\frac{1}{|k|^2}\mathrm{d}k=\frac{6\epsilon}{\pi},
\end{align*}
and consequently $\|\mathds{1}_{\mathbb{R}^3\setminus B_r(0)}\Pi_K w_0\|\leq \sqrt{\frac{6\epsilon}{\pi}}+\|\mathds{1}_{\mathbb{R}^3\setminus B_r(0)}u\|$. Making use of the observation that $\frac{1}{|y|}\mathds{1}_{\mathbb{R}^3\setminus B_r(0)}(y)\leq \frac{1}{r}$ yields
\begin{align}
\nonumber
&\|\mathds{1}_{\mathbb{R}^3\setminus B_r(0)}u\|^2\leq \frac{1}{r^2}\int_{\mathbb{R}^3} |y|^2|u(y)|^2\mathrm{d}y=\frac{1}{r^2}\left\|\nabla_k\hat{u}\right\|^2=\frac{1}{2\pi^2 r^2}\left\|f_1- f_2\right\|^2
\end{align}
with $f_1(k):=\frac{\chi^\epsilon\left(2\epsilon\leq |k|\leq K\right)}{|k|^2}$ and $f_2(k):=\frac{\nabla_k\chi^\epsilon\left(2\epsilon\leq |k|\leq K\right)}{|k|}$.  Clearly we can bound $\left\|f_1\right\|^2\leq \int_{|k|\geq \epsilon}\frac{1}{|k|^4}\mathrm{d}k=\frac{4\pi}{\epsilon}$. Furthermore we obtain, using $\left\|\nabla_k\chi^\epsilon\left(2\epsilon\leq |k|\leq K\right)\right\|_\infty\lesssim \frac{1}{\epsilon}$, 
\begin{align*}
\left\|f_2\right\|^2\lesssim \frac{1}{\epsilon^2}\left(\int_{\epsilon\leq |k|\leq 3\epsilon}\frac{1}{|k|^2}\mathrm{d}k+\int_{K-\epsilon\leq |k|\leq K+\epsilon}\frac{1}{|k|^2}\mathrm{d}k\right)=\frac{4}{\epsilon}.
\end{align*}
In combination this yields $\|\mathds{1}_{\mathbb{R}^3\setminus B_r(0)}\Pi_K w_0\|^2\lesssim \epsilon+\frac{1}{r^2\epsilon}$, which concludes the proof with the concrete choice $\epsilon:=\frac{1}{r}$.
\end{proof}

\begin{cor}
\label{Corollary-Outside_Mass}
Given $A\subset \mathbb{R}^3$, let us define the operator $\mathcal{N}_A:=\widehat{D_A}$ with $D_A(\rho):=\int_{A}\mathrm{d}\rho(y)$, using the notation of Definition \ref{Definition-Hat operators}, i.e. $\alpha^2\mathcal{N}_A$ counts the number of particles in the region $A$. Furthermore let $A'\subset \mathbb{R}^3$. Then given a constant $C>0$, there exists a constant $D>0$ such that for all states $\Psi$ with $\mathrm{supp}\left(\Psi\right)\subset A'$ and $\chi\left(\mathcal{N}\leq C\right)\Psi=\Psi$
\begin{align*}
\braket{\Psi|\mathbb{H}_K|\Psi}\geq E_\alpha+\braket{\Psi|\mathcal{N}_A|\Psi}-\sqrt{\frac{D}{\mathrm{dist}(A,A')}},
\end{align*}
where $K>0$.
\end{cor}
\begin{proof}
Let us define the function $v_x:=\mathds{1}_A \Pi_K w_x$ and rewrite $\mathbb{H}_K-\mathcal{N}_A$ as
\begin{align*}
\mathbb{H}_K-\mathcal{N}_A=-\Delta_x-a\left(\Pi_K w_x-v_x\right)-a^\dagger\left(\Pi_K w_x-v_x\right)+\mathcal{N}-\mathcal{N}_A-a\left(v_x\right)-a^\dagger\left(v_x\right).
\end{align*}
Identifying $L^2\Big(\mathbb{R}^3,\mathcal{F}\big(L^2\! \left(\mathbb{R}^3\right)\big)\Big)\cong L^2\Big(\mathbb{R}^3,\mathcal{F}\Big( L^2\! \left(\mathbb{R}^3\setminus A\right)\Big)\Big)\otimes \mathcal{F}\Big(L^2\! \left(A\right)\Big)$, we observe that $-\Delta_x-a\left(\Pi_K w_x-v_x\right)-a^\dagger\left(\Pi_K w_x-v_x\right)+\mathcal{N}-\mathcal{N}_A$ is the restriction (in the sense of quadratic forms) of $\mathbb{H}_K$ to states of the form $\Psi'\otimes \Omega$, where $\Omega$ is the vacuum in $\mathcal{F}\Big(L^2\! \left(A\right)\Big)$, and therefore we have the operator inequality $-\Delta_x-a\left(\Pi_K w_x-v_x\right)-a^\dagger\left(\Pi_K w_x-v_x\right)+\mathcal{N}-\mathcal{N}_A\geq E_\alpha$. Consequently
\begin{align*}
\braket{\Psi|\mathbb{H}_K-\mathcal{N}_A|\Psi}\geq E_\alpha-\braket{\Psi|a\left(v_x\right)+a^\dagger\left(v_x\right)|\Psi}\geq E_\alpha-\sup_{x\in A'}\|v_x\|\left(1+C\right),
\end{align*}
where we have used the operator inequality $a\left(v_x\right)+a^\dagger\left(v_x\right)\geq -\|v_x\|\left(1+\mathcal{N}\right)$, as well as the assumptions $\mathrm{supp}\left(\Psi\right)\subset A'$ and $\chi\left(\mathcal{N}\leq C\right)\Psi=\Psi$, in the second inequality. This concludes the proof, since $\|v_x\|^2=\int_{A}\left|\Pi_K w_0(y-x)\right|^2\mathrm{d}y\leq \int_{|y|\geq \mathrm{dist}(A,A')}\left|\Pi_K w_0(y)\right|^2\mathrm{d}y$ for all $x\in A'$ and $\int_{|y|\geq \mathrm{dist}(A,A')}\left|\Pi_K w_0(y)\right|^2\mathrm{d}y\lesssim \frac{1}{\mathrm{dist}(A,A')}$, see Lemma \ref{Lemma-Uniform_singularity_control}.
\end{proof}

\begin{center}
\textsc{Acknowledgments}
\end{center}
Funding from the European Union’s Horizon 2020 research and innovation programme
under the ERC grant agreement No 694227 is  acknowledged.

\bibliographystyle{plain}

\end{document}